\documentclass[pdflatex,sn-mathphys-num]{sn-jnl}

\usepackage{graphicx}%
\usepackage{multirow}%
\usepackage{amsmath,amssymb,amsfonts}%
\usepackage{amsthm}%
\usepackage{mathrsfs}%
\usepackage[title]{appendix}%
\usepackage{xcolor}%
\usepackage{textcomp}%
\usepackage{manyfoot}%
\usepackage{booktabs}%
\usepackage{algorithm}%
\usepackage{algorithmicx}%
\usepackage{algpseudocode}%
\usepackage{listings}%
\usepackage[capitalize]{cleveref}
\RequirePackage[margin=1.15in]{geometry}

\usepackage{comment}

\theoremstyle{thmstyleone}%
\newtheorem{theorem}{Theorem}
\newtheorem{lemma}{Lemma}
\newtheorem{observation}{Observation}
\makeatletter
\newtheorem*{rep@theorem}{\rep@title}
\newcommand{\newreptheorem}[2]{%
\newenvironment{rep#1}[1]{%
 \def\rep@title{#2 \ref{##1}}%
 \begin{rep@theorem}}%
 {\end{rep@theorem}}}
\makeatother

\newreptheorem{theorem}{Theorem}
\newreptheorem{lemma}{Lemma}

\theoremstyle{definition}
\newtheorem{inline-table}{Table}
\crefname{inline-table}{Table}{Tables}
\Crefname{inline-table}{Table}{Tables}

\theoremstyle{thmstyletwo}%

\theoremstyle{thmstylethree}%
\newtheorem{definition}{Definition}%

\newcommand{\Le}{\left}
\newcommand{\Ri}{\right}

\newcommand{\lambdabar}{\overline{\lambda}}
\newcommand{\E}{\mathbb{E}}

\newcommand{\Q}{\hat{\boldsymbol{q}}}
\newcommand{\Qi}{\hat{q}}
\newcommand{\p}{\mathbb{P}}
\newcommand{\bpi}{\boldsymbol{\pi}}
\newcommand{\idle}{\text{idle}}
\newcommand{\piidle}{\bpi_\idle}

\newcommand{\indic}{\mathbb{I}}

\begin{document}

\title[Article Title]{Tandem Queues with Hybrid Abandonment in Heavy Traffic}


\author*[1]{\fnm{Cameron} \sur{Curtis}\orcid{0009-0007-2929-5497}}\email{cameroncurtis2028@u.northwestern.edu}

\author[1]{\fnm{Izzy} \sur{Grosof}\orcid{0000-0001-6205-8652}}\email{izzy.grosof@northwestern.edu}

\author[1]{\fnm{Seyed} \sur{Iravani}\orcid{0000-0003-2172-1205
}}\email{s-iravani@northwestern.edu}

\affil*[1]{\orgdiv{Department of Industrial Engineering and Management Sciences}, \orgname{Northwestern University}, \orgaddress{\street{2145 Sheridan Rd}, \city{Evanston}, \postcode{60208}, \state{Il}, \country{US}}}


\abstract{    Hybrid-abandonment systems, where jobs abandon from certain stages of service but not others, occur frequently in medical, service, and computing systems. Abandonment affects the heavy-traffic behavior of both the stages with and without abandonment. In contrast to existing heavy-traffic works that focus on cases where there is abandonment at either all stages or no stages, we focus on the heavy-traffic behavior for these hybrid-abandonment systems.

    We provide the first heavy-traffic characterization of the second queue length in a two-stage tandem queue with abandonment only in the first stage. We characterize the scaled mean of the second stage queue length and prove that the scaled second queue length is asymptotically exponentially distributed and asymptotically independent of the first stage queue length.

    We derive our results by analyzing Markovian arrival process queues with countably infinite phases. We do this by extending techniques used to analyze queues with Markov-modulated arrivals to the broader class of infinite Markovian arrival process queues. We prove asymptotic results for a more general infinite phase process condition, an infinite birth-death phase process, and the hybrid-abandonment tandem case.

    We also empirically demonstrate that our heavy-traffic approximations are good approximations across various system settings.}




\maketitle

\section{Introduction}\label{sec1}

\label{Sec1}

Existing literature on tandem queueing models, where jobs receive service at multiple stages in series, predominantly focus on systems where jobs can abandon from either no stages or all stages. These models are often a worthwhile proxy for certain real-world systems. This all-or-nothing approach to abandonment, though, leaves out systems where we observe abandonment in some stages but not in others.

In many real-world systems, the feasibility of abandoning may change significantly as a job moves through the system. We call systems like these with abandonment at some stages but not others \emph{hybrid-abandonment} systems. We specifically focus on tandem systems with abandonment in the first stage but not the second stage as seen in \cref{fig:HADiag}. In a service context, an example of this structure is a coffee shop. While a customer waits in line, they can abandon but once their order is handed to the barista, abandonment is no longer possible. As another example, a prospective surgical patient on a waitlist might abandon to seek a different treatment. However, after surgery the patient cannot abandon from the waitlist for a follow-up appointment. In a manufacturing or other industrial example, this same model represents a raw material that can be repurposed until it undergoes some initial processing (cutting, drilling, etc.) and is then only fit for a particular use.
In computing, a payment processing attempt may be abandoned before money moves between accounts, but once money leaves the payer's account, it must be placed into the recipient's account.
\begin{figure}
    \centering
    \includegraphics[width=0.7\linewidth]{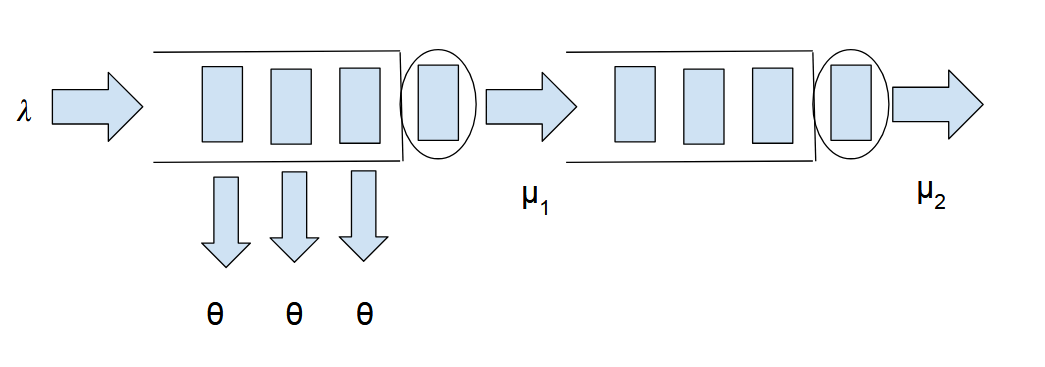}
    \caption{A hybrid-abandonment tandem queue with abandonment in the first stage}
    \label{fig:HADiag}
\end{figure}

We study hybrid-abandonment queues in heavy traffic. The heavy-traffic regime considers a sequence of queues whose load (the fraction of time the second server is busy) approaches $1$. This heavy-traffic limit is a worthwhile and valuable theoretical result on its own but also serves as a valuable approximation for the system's behavior in high-load settings (as we demonstrate in \cref{Sec6}). Existing techniques to analyze tandem queues in heavy traffic are applicable to systems with either abandonment at no stages (see \citet{jackson1957networks}) or abandonment at all stages (see \citet{reed2004diffusion}) and don't easily generalize to hybrid-abandonment systems. 
We characterize the asymptotic distribution of both queue lengths in the heavy-traffic limit. We show that the first queue length and the scaled second queue length are asymptotically independent and that the scaled second queue length is asymptotically exponentially distributed and we compute its mean (\cref{Thm3}).

We also generalize our results to provide sufficient conditions for analogous asymptotic independence and asymptotic exponentiality results for both Markovian arrival processes with countably infinite phases (\cref{Thm1}) and Markovian arrival processes with birth-death phase processes (\cref{Thm2}).

\subsection{Structure}

In \cref{Sec2}, we provide an overview of prior work relating to tandem queues, queues with Markovian-arrival processes, and heavy-traffic results for these systems and demonstrate the need for new methods in analyzing systems such as our hybrid-abandonment systems.

Next in \cref{Sec3} we define the necessary tools for our analysis: the infinite-phase Markovian arrival process (IMAP) queue and heavy-traffic families of IMAP queues. We also formally define the sufficient conditions we use in our analysis.

In \cref{Sec4}, having defined all the necessary concepts in \cref{Sec3}, we present our main theorems. We present our theorems from general to specific, beginning with our heavy-traffic analysis for general families of IMAP queues before narrowing to an analogous result for families of IMAP queues with birth-death phase processes. Finally, we present the narrowest result where we specifically consider hybrid-abandonment tandem systems.

In \cref{Sec5} we provide proofs of our main theorems and the necessary lemmas.

Lastly, in \cref{Sec6} we demonstrate via simulation that our heavy-traffic approximation is useful outside of the asymptotic regime.

\section{Prior Work}
\label{Sec2}
The prior work we discuss can be broadly divided into three categories: works considering tandem queueing models with abandonment (\cref{sec:TandQueueAb}), works considering systems with Markovian arrival processes (or MAPs) in heavy-traffic (\cref{sec:MAPHT}), and works considering matrix-analytic methods (\cref{sec:MatAnMethods}). 
\subsection{Tandem Queues with abandonment}
\label{sec:TandQueueAb}
While queueing networks with abandonment have been fairly well studied, far fewer works have considered hybrid-abandonment systems. This is especially true for systems with abandonment at earlier stages but not later stages.

Many of the works that have considered tandem queues with abandonment have done so in the dynamic control context. Zayas-Cab\'an and collaborators studied tandem queues in medical settings where patients proceed through two treatment stages. They considered dynamic control of servers in systems where customers can abandon in the second stage (as in \citet{zayas2019policies}) or in both stages (as in \citet{zayas2023scheduling,zayas2016dynamic}). \citet{andradottir2023optimal} studied a similar system with abandonment only from the second stage and a finite second stage buffer and characterized the optimal dynamic server assignment in this system. While these works consider dynamic control of tandem systems, we focus on asymptotic behavior of static systems.

\citet{reed2004diffusion} study a generalized Jackson network with abandonment at each station. They show that in heavy traffic such a system is well-approximated by a reflected Ornstein-Uhlenbeck process. They consider a different heavy traffic regime, however. While they consider a regime where arrival rates and service rates jointly approach infinity and the abandonment rate stays fixed, we consider a regime where the arrival rate rises to a critical level while the service rates and abandonment rates stay fixed.

\subsection{Markovian Arrival Processes in Heavy-Traffic}
\label{sec:MAPHT}
Queues with Markovian arrival processes (MAPs) and hybrid-abandonment tandem queues are conceptually similar and we draw on techniques designed for studying queues with Markovian arrival processes to study hybrid-abandonment tandem queues. In both cases, the arrivals to a single-server queue (the queue itself in the MAP case and the second stage queue in the tandem case) are determined by an independently evolving CTMC, the ``phase process" in the MAP case and the first queue in the tandem case. The primary differences are two-fold: First, MAP queues allow arrivals to the queue to occur either while the process in a state (``in-state arrivals") or alongside a transition of the phase process (``between-state arrivals"). In contrast, arrivals to the second stage of the tandem system only occur alongside transitions in the first stage. Second, the tandem queues are usually allowed to have infinite buffers while MAP queues (with extremely limited exceptions) have only been considered in the case of finitely many phase states.

\citet{choudhury1994heavy} study the tail behavior of the work in a queue with batch MAP arrivals, a generalization of MAPs that permits multiple arrivals at the same time, and general service times in heavy traffic. They show that in heavy traffic, the stationary distribution of the work in the queue has an exponential tail. We show that with an infinite-phase MAP and exponential service times, the scaled number of jobs in the queue is exponentially distributed rather than the work simply having an exponential tail.

Beyond \citet{choudhury1994heavy}, almost all of the work that considers MAP queues in heavy traffic specifically studies the special case of Markov-modulated arrival queues in heavy traffic. A Markov-modulated arrival queue is an MAP with only in-state arrivals, i.e., transitions of the phase process never cause arrivals. \citet{yechiali1971queuing} studied a system with two-level Markov-modulated arrival and service rates. In the two-level setting, they were able to derive the steady-state distribution and average queue length (up to a parameter being left as a root of a cubic equation). Their technique unfortunately doesn't generalize to systems with more than two modulation states.
\citet{burman1986asymptotic} provided the first heavy traffic result for queues with general finite-modulation-process Markov-modulated arrivals. They studied a system with Markov-modulated arrivals and general service times and provided a derivation of the mean queue length via a diffusion-scaling approach. They proved this result up to an interchange of limits assumption.

\citet{falin1999heavy} studied the same system and rigorously characterized the Laplace transform of the limiting heavy-traffic distribution. \citet{dimitrov2011single} extended the analysis of \citet{falin1999heavy} to a system with Markovian-service rates and arrival rates but still only considered the case of finite modulation processes. We extend  \citet{falin1999heavy} by generalizing their methods to MAPs with infinitely many phases.

\citet{mitrani2005approximate} derives asymptotically exact approximations for the average queue length in Markov-modulated arrival queues. Rather than using a transform method like \citet{falin1999heavy} and \citet{dimitrov2011single}, \citet{mitrani2005approximate} instead uses a spectral expansion-based technique to derive their approximations. As with \citet{falin1999heavy} and \citet{dimitrov2011single}, \citet{mitrani2005approximate} only considers Markov-modulated arrival queues with finitely many phases while we consider the broader class of MAPs with countably infinite phases.

\citet{thorsdottir2016markov} consider the heavy-traffic behavior of M/G/1 queues with Markov-modulated arrivals and show that the work in the system is asymptotically exponentially distributed. We focus on the number of jobs in queue rather than the workload and consider models with MAPs and countably infinite phases.

Following the works by \citet{falin1999heavy} and \citet{dimitrov2011single}, equivalent results were independently derived using the drift method as stepping stones to other results. \citet{mou2024heavy}, for example, analyzed input-switched Markovian models in heavy traffic. They independently reproved via the drift method that the expected scaled queue length, $\mathbb{E}[\Le(1-\rho\Ri)Q]$, converges in heavy traffic and expressed the limiting value as an infinite summation.

\citet{grosof2023reset} derived the waiting time in first-come-first-served multiserver-job models. As part of their analysis, they independently reproved expected queue length results for systems with Markov-modulated service rates via the drift method. We consider the more general case of MAP queues instead and the case of infinitely many phase states.

In work currently under submission, \citet{Hurtado-LangeDaniela2026MMJi} consider the heavy-traffic behavior of the join-the-shortest-queue load-balancing policy in queues with Markov-modulated arrivals, including arrival processes with countably infinite modulation states. While they assume sufficient conditions for their heavy-traffic results, we prove our heavy-traffic conditions for the class of hybrid-abandonment tandem systems.

\subsection{Matrix Analytic Methods}
\label{sec:MatAnMethods}
Matrix analytic methods can be used to numerically compute performance metrics for a variety of queueing systems via matrix computations \citep{neuts1978m, neuts1994matrix, latouche1999introduction, ramaswami1980n}. The methods were first used in \citet{neuts1978m} to compute queue length, busy period lengths, and other performance metrics. 
\citet{regterschot1986queue} used matrix analytic methods to find computationally tractable bounds and approximations for a number of performance metrics.

\citet{ciucu2018two} provided a more computationally efficient method to derive tight bounds on the transform of a system with Markov-modulated arrivals (among other systems).

\citet{wang2019tandem} studied tandem queues with abandonment at both stations as level-dependent quasi-birth-death processes. This allowed them to numerically compute performance metrics of such a system efficiently, but required the second queue have a finite buffer. Our work considers the case where the second queue has an infinite buffer and no abandonment.

Matrix analytic methods provide numerically efficient methods to bound performance metrics in a wide array of stochastic systems, but we focus on exact characterization of these metrics and, particularly, asymptotic analysis in the heavy-traffic limit rather than performance in any single system. Matrix analytic methods also tend to focus on systems with one-dimensionally infinite state spaces, such as Markov-modulated queues with finitely many modulation states, rather than systems with doubly-infinite state spaces, such as arrival processes with countably infinite state spaces as in our hybrid-abandonment system.

The few matrix-analytic works that do consider arrival processes with countably infinite state spaces do not characterize the system's heavy-traffic behavior.
By contrast, we characterize heavy-traffic behavior, which allows us to derive simple and accurate approximations for the queueing behavior at all loads, as we show in \cref{Sec6}.
\citet{ramaswami1996some} introduce an operator-theoretic method to characterize the rate matrix $\boldsymbol{R}$ and matrix $\boldsymbol{G}$ in level-dependent, infinite-phase quasi-birth-death processes, including MAP queues with countably infinite phases. Their techniques allow them to analyze particular infinite phase MAP queues but do not allow them to analyze heavy-traffic behavior or produce simple approximations. Other works that consider infinite phase MAP queues do so by building upon \citet{ramaswami1996some}. \citet{latouche2011queues} and \citet{kroese2004spectral}, for instance, consider when the behavior of a sequence of truncated finite-phase systems converge to the behavior of an untruncated infinite-phase system. \citet{lian2011departure} uses a similar framework to create an algorithm to compute inter-departure times and busy period of different stages in a tandem queue. In contrast to works like \citet{ramaswami1996some} that focus on numerically understanding the behavior of particular IMAP systems, we focus on the heavy-traffic behavior of tandem systems and IMAP queues.

\section{Model and Definition}
\label{Sec3}

Our ultimate goal is to derive structural insights into the behavior of the hybrid-abandonment system across various loads. To do so, we seek to characterize the heavy-traffic behavior of the system, which we will use in \cref{Sec6} to derive such insights. To prove our heavy-traffic results, we begin by defining a generalization of the hybrid-abandonment queue, the infinite-phase Markovian arrival process queue, or IMAP queue, in \cref{IMAP}. We show that the hybrid-abandonment queue is a special case of the infinite-phase Markovian arrival process queue in \cref{HA_IMAP}. Finally, we define the near-inverse matrix and $M$ function for an IMAP queue (\cref{Rdefn}) and the families of IMAP queues we consider in our analysis (\cref{IMAP_families}).

\subsection{The Infinite-phase Markovian Arrival Queue (IMAP) Queue}
\label{IMAP}

We begin by defining a finite-phase Markovian arrival process (or ``MAP") queue before distinguishing our infinite-phase Markovian arrival process queue. The arrivals to an MAP queue are controlled by a finite-state CTMC $Y$, the ``phase" of the arrival process, which evolves independently of the queue length process $Q$. While $Y$ is in state $i$, arrivals to $Q$ occur at rate $\lambda_i$. We will call these ``in-state" arrivals. When $Y$ transitions from state $i$ to state $j$, there is some probability $p_{ij}$ that an arrival to $Q$ is triggered. We will call these arrivals ``between-state" arrivals.
If $p_{ij} = 0$ for all $i,j$, the MAP queue is simply a Markov-modulated arrival queue. 

The infinite-phase Markovian arrival process (or ``IMAP") queue generalizes the MAP queue to allow $Y$ to have countably infinite states. Without loss of generality, we will assume that these states are $\mathbb{N}_0 = \{0, 1, 2, \dots\}$.

The phase process $Y(t)$ has a generator matrix $\boldsymbol{A} = [\alpha_{i,j}]_{i,j\in \mathbb{N}_0}$.
Note that throughout this paper, we use bold font to indicate matrices and vectors.
When $(Y(t), Q(t))$ has a unique steady-state distribution, we denote the steady-state probability of $\Le(Y, Q\Ri)$ as $\pi_{i,q} = \p\Le(Y=i,Q=q\Ri)$ and the steady-state distribution of $Y$ alone as $\boldsymbol{\pi}$ where $[\boldsymbol{\pi}]_i=\pi_i = \p\Le(Y=i\Ri)$. When $Y\Le(t\Ri)=i$, the instantaneous arrival rate is $\lambda_i+\sum_{j\ne i} p_{i,j}\alpha_{i,j}$ (the instantaneous rate of in-state arrivals plus the instantaneous rate of between-state arrivals). Thus, the long-run arrival rate to the queue is 
\begin{align*}
    \lambdabar = \E\Big[\lambda_Y+\sum_{j\ne Y} p_{Y,j} \alpha_{Y,j} \Big] = \sum_{i} \pi_i \Big(\lambda_Y+\sum_{j\ne i} p_{i,j} \alpha_{i,j}\Big).
\end{align*}
We often collect the in-state and between-state arrival rates into the $\boldsymbol{F}$ matrix where the elements of $\boldsymbol{F}$ are given by:
\begin{align*}
    [\boldsymbol{F}]_{i,j} = F_{i,j} = \begin{cases}
        \lambda_i, \quad &i=j\\
        p_{i,j} \alpha_{i,j}, \quad& i\ne j
    \end{cases}
\end{align*}
for $i,j\in \mathbb{N}_0$.

Most of our analysis of this system is done in terms of the system's $z$-transform. The $z$-transform of the queue length random variable, \[\overline{q}(z) = \E[z^Q],\]
captures the marginal distribution of our queue length random variable while the partial $z$-transforms, \begin{align*}
    \Qi(z;i) = \E[z^Q\indic{\{Y=i\}}],
\end{align*}
captures the dependence between the queue length random variable and the phase process.

Often in our analysis, it is useful to refer to $\piidle = [P(Y=i, Q=0)]_{i\in \mathbb{N}_0}$, the idleness vector, and $\varepsilon = P(Q=0)$, the system's total idleness. We also let $e^{T}$ be the ones vector, i.e., $e^T = [1, 1, \dots]^T$.

\subsection{Hybrid-Abandonment Tandem Queues}
\label{HA_IMAP}
We are particularly interested in the hybrid-abandonment systems we discussed in \cref{Sec1}. These are tandem queues where jobs enter the first stage of queueing and either abandon the first queue or receive service and enter the second queue as depicted in \cref{fig:HADiag}.
These are a subset of the IMAP queues we introduced above. To model a hybrid-abandonment queue as an IMAP queue we will treat the first queue length as the phase process $Y$ and the second queue length as the queue length process $Q$. When $n$ jobs are in the first queue, the first queue transitions upward with rate $\alpha_{n, n+1} = \lambda$ and downward with rate $\alpha_{n, n-1} =\Le(\mu_1+\Le(n-1\Ri)\theta\Ri)\indic\{n>0\}$. When the first queue transitions downward it triggers a between-state arrival to the second queue with probability $p_{n,n-1}=\frac{\mu_1}{\mu_1 + \Le(n-1\Ri)\theta}$. These downward between-state arrivals are the only arrivals to the second queue. When discussing the hybrid-abandonment systems we refer to the service rate of the second queue as $\mu_2$ rather than simply $\mu$ for clarity.

\subsection{Near-Inverse Matrix $\boldsymbol{R}$ and $M$ Function}
\label{Rdefn}
Often in our analysis, a quantity we are interested in appears right-multiplied by the generator matrix $\boldsymbol{A}$.
A way to ``undo" the effects of this multiplication and recover the quantity of interest would therefore be extremely useful.
However, because $\boldsymbol{A}$ is a generator matrix, it is not invertible.
Lacking a true inverse, we often utilize a \textit{near-inverse} of $\boldsymbol{A}$. 
Given a generator matrix $\boldsymbol{A}$ with steady state distribution $\bpi$ and state space $\mathbb{N}_0$, define the vector $\bpi_{\overline{0}} = [\pi_i]_{i>0}$.
We do so without loss of generality: we may choose any mapping from the states of $\boldsymbol{A}$ to $\mathbb{N}_0$, and state 0 may be chosen arbitrarily.
We then define a near-inverse matrix $\boldsymbol{R}$ to be a matrix that satisfies our near-inverse equation:    \[\boldsymbol{A}\boldsymbol{R}= \begin{bmatrix}
        0 & -\frac{1}{\pi_0} \bpi_{\overline{0}}\\
        0& I
    \end{bmatrix}=
    \begin{bmatrix}
        0 & -\frac{\pi_1}{\pi_0} & -\frac{\pi_2}{\pi_0} & \dots &-\frac{\pi_n}{\pi_0} &\dots\\
        0 & 1 & 0 &\dots &0&\dots\\
        0&0&1&\dots&0&\dots\\
        \vdots\\
        0&0&0&\dots &1&\dots\\
        \vdots
    \end{bmatrix}\]
    or, equivalently,
    \[\boldsymbol{A}\boldsymbol{R}x=x-\frac{x_0}{\pi_0}\bpi\]
    for any vector $x$.

In \cref{R_exst_lemma}, we construct $\boldsymbol{R}$ as the time-in-state matrix by letting $[\boldsymbol{R}]_{i,j} = r_{i,j}$ be the expected time in state $j$ between starting in state $i$ and reaching state $0$ for the first time and showing that it satisfies our near-inverse equation. 

We use the notation $\boldsymbol{R}$ for our near-inverse to match the notation used by \citet{falin1999heavy}.
However, it should be noted that our near-inverse matrix $\boldsymbol{R}$ is not the same as the rate matrix $\boldsymbol{R}$ that often appears in matrix-analytic works. While they share a symbol and have some superficial similarities (both are concerned with time spent in or visits to particular states before a certain event occurs), they are not meaningfully related. In the rate matrix $\boldsymbol{R}$, the element $[\boldsymbol{R}]_{i,j}$ is the expected number of visits to state $(m+1, j)$ between starting in state $(m, i)$ and first returning to queue length $m$ (which is independent of the particular queue length $m$) \citep{he2014fundamentals}.

Once we define a near-inverse matrix $\boldsymbol{R}$ for an IMAP queue, we can also define the $M$ function for the same system: 
\begin{align}
M(z)=1-\mu^{-1}\bpi\Le(-z \boldsymbol{F}+\mu \boldsymbol{I}\Ri)\boldsymbol{R}\Le(-z \boldsymbol{F}+\mu \boldsymbol{I}\Ri)e^T.
\end{align}


In families of queues with the necessary convergence conditions (including our original hybrid-abandonment queues as we demonstrate in \cref{Thm3Sup}), we prove that $M$'s asymptotic behavior is the same as the asymptotic scaled mean queue length. Thus, characterizing $M$ and $\boldsymbol{R}$ are key points in our analysis.

\subsection{Families of IMAP Queues}
\label{IMAP_families}
We frequently consider families of IMAP queues indexed by their idleness in our analysis. A family of IMAP queues indexed by idleness is a map from $\varepsilon$ in some $\mathcal{E}=(0, \varepsilon_{\text{max}})$ to IMAP queues such that the $\varepsilon^\text{th}$ queue has idleness $\varepsilon$. The heavy-traffic limit in the system is a limit of the IMAP queues as $\varepsilon\to 0$. 
We are specifically interested in families of IMAP queues that, like our hybrid-abandonment queues, have a particular notion of ``convergence" in the $\varepsilon\to 0$ limit, namely that the phase process and the queue process are well-behaved in this limit. Below we present the definition of our \emph{heavy-traffic conditions} for an IMAP family that we use throughout our analysis. We defer proving that our original hybrid-abandonment queues have these features to \cref{Thm3Sup}.
\begin{definition}[Heavy-Traffic Conditions]
\label{def:HT_cond}
    A family of IMAP queues parameterized by their idleness $\varepsilon$ for $\varepsilon\in \mathcal{E}$ that satisfies our \emph{heavy-traffic conditions} is a family of IMAP queues where there exist matrices $\boldsymbol{F}^{(0)}$ and $\boldsymbol{R}^{(0)}$, $\mu^{(0)}\in (0,\infty)$, and a distribution $\bpi^{(0)}$ such that:
    \begin{enumerate}
        \item 
        \begin{enumerate}
            \item For all $i\in \mathbb{N}_0$, $\lim_{\varepsilon\to 0}\pi_i^{(\varepsilon)}= \pi_i^{(0)}$.
            
            \item For all $i,j\in \mathbb{N}_0$, $\lim_{\varepsilon\to 0} r_{i,j}^{(\varepsilon)}= r_{i,j}^{(0)}$.

            \item For all $i,j\in \mathbb{N}_0$, $\lim_{\varepsilon\to 0} \boldsymbol{F}_{i,j}^{(\varepsilon)}= \boldsymbol{F}_{i,j}^{(0)}$.

            \item $\lim_{\varepsilon\to 0}\mu^{(\varepsilon)}= \mu^{(0)}$.
        \end{enumerate}


        \item \label{M_conv}Let $M^{(0)}=1-(\mu^{(0)})^{-1}\bpi^{(0)}\Le(-\boldsymbol{F}^{(0)}+\mu^{(0)} \boldsymbol{I}\Ri)\boldsymbol{R}^{(0)}\Le(-\boldsymbol{F}^{(0)}+\mu^{(0)} \boldsymbol{I}\Ri)e^T$. \\
        For all $z\in (0,1)$,
        $\lim_{\varepsilon\to 0}M(z^\varepsilon;\varepsilon) =M^{(0)}$.
        
        \item $Y^{(\varepsilon)}$ is ergodic for all $\varepsilon\in \mathcal{E}$.
    \end{enumerate}
\end{definition}

To aid in our analysis, we establish a few additional conditions for families of IMAP queues. These turn out to be useful sufficient conditions for families of IMAP queues to have well-defined heavy traffic asymptotic distributions.

\begin{definition}[Rate Bounded]
\label{def:rtbdd}
    A \emph{rate-bounded} family of IMAP queues is a family of IMAP queues where, there exist some $\hat{\lambda}>0, T_{\max}>0,$ and $\hat{\mu} > 0$ such that for all modulation states and all possible idleness $\varepsilon\in \mathcal{E}$,
    \begin{enumerate}
        \item The instantaneous arrival rate to the queue is upper bounded by $\hat{\lambda}$.

        \item The expected time until transition in any modulation state is upper bounded by $T_{\max}$.

        \item The service rate is lower bounded by $\hat{\mu}$.
    \end{enumerate}
\end{definition}

The \textit{rate-bounded} condition essentially exists to rule out a few particularly ill-behaved families of IMAP queues. A family being rate bounded guarantees that the phase process does not linger too long in any state or become explosive both for particular $\varepsilon$ values and in the $\varepsilon\to 0$ limit.

\begin{definition}[Birth-Death]
\label{def:bd}
    A \emph{birth-death} family of IMAP queues is a family of IMAP queues where $Y^{(\varepsilon)}$ is a birth-death process on $\mathbb{N}_0$ for all $\varepsilon\in \mathcal{E}$, i.e., for all $\varepsilon\in \mathcal{E}$, $Y^{(\varepsilon)}(t)\in \{0, 1, 2, \dots\}$ and from $Y^{(\varepsilon)} = i$, $Y^{(\varepsilon)}$ can only transition to $i-1$ or $i+1$. 
\end{definition}

A \textit{birth-death} family of IMAP queues is simply an extension of the concept of a birth-death phase process to cover all $\varepsilon\in \mathcal{E}$.

\begin{definition}[Uniformly Stable]
\label{def:unifstab}
    A \emph{uniformly stable} birth-death family of IMAP queues is a family of IMAP queues where there is some $N$, $a\in(0,1)$, and $b\in (0, \frac{1}{2})$ such that:
        \begin{enumerate}
            \item For all $\varepsilon\in \mathcal{E}$ and all $n\geq N$, $\frac{\alpha^{(\varepsilon)}_{n,n+1}}{\alpha^{(\varepsilon)}_{n+1,n}} \leq a$

            \item For all $\varepsilon\in \mathcal{E}$ and all $n\geq N$, $\frac{\alpha^{(\varepsilon)}_{n,n+1}}{\alpha^{(\varepsilon)}_{n,n-1}+\alpha^{(\varepsilon)}_{n,n+1}} \leq b$
        \end{enumerate}
\end{definition}

Intuitively, a family of  IMAP queues is \textit{uniformly stable} if the phase process has a downward drift for all $\varepsilon \in \mathcal{E}$ when $n$ is ``big enough".

Lastly, we have our desired characteristic: that a family of IMAP queues is asymptotically exponential and asymptotically independent.

\begin{definition}[Asymptotically Exponential and Independent]
\label{def:asymp_exp_ind}
    An \emph{asymptotically exponential and independent} family of IMAP queues with a heavy-traffic convergent phase process with mean $M$ is a family of IMAP queues with heavy-traffic convergent phase process where
    \begin{equation}    \label{eq:asympexp}
         \lim_{\varepsilon\to 0} \overline{q}(z^{\varepsilon};\varepsilon) = \frac{1}{1+M\ln z}
    \end{equation}
    and
    \begin{equation}
    \label{eq:asymp_ind}
        \lim_{\varepsilon\to 0} \hat{q}(z^\varepsilon;i,\varepsilon) = \frac{\pi^{(0)}_i}{1+M\ln z}
    \end{equation}
    for all $z\in (0,1)$ and all $i\in \mathbb{N}_0$.
\end{definition}

\cref{eq:asympexp} states that the $z$-transform of the random variable $\varepsilon Q$ approaches that of an exponential random variable with mean $M$. This is exactly equivalent to $\varepsilon Q$ converging in distribution to an exponential random variable with mean $M$. 

\cref{eq:asymp_ind} states that for any modulation state $i$, the queue length conditioned on the modulation state also is an exponential random variable with mean $M$. Thus, the queue length and phase process are independent.

The asymptotic exponential and independence properties are the focuses of our analysis. In each of \cref{Thm1,Thm2,Thm3} we will focus on showing that, under particular conditions, a family of IMAP queues has these asymptotic properties.

\section{Main Results}

\label{Sec4}

We characterize the asymptotic behavior of hybrid abandonment queues as the second queue's idleness approaches zero. We begin by analyzing a considerably broader class of models, infinite-phase Markovian arrival processes with ``light-tailed phase processes" (\cref{def:light-tailed}) and describing their behavior in the same heavy-traffic limit. This ``light-tailed" condition can be roughly described as a weighted average of $\boldsymbol{R}$ values having finite first and second moments.

Here we present statements of each of Theorems 1, 2, and 3 before presenting a full proof of each in \cref{Full_Thm1,Full_Thm2,Full_Thm3}, respectively.

\begin{theorem} \label{Thm1}
    If a family of IMAP queues with a heavy-traffic convergent phase process has a light-tailed phase process (\cref{def:light-tailed}, which will be defined in \cref{Full_Thm1}), then $\varepsilon Q^{(\varepsilon)}$ is asymptotically exponential with mean $M^{(0)}$ (where $M^{(0)}$ is defined in \cref{def:HT_cond}) and asymptotically independent from $Y^{(\varepsilon)}$ in the $\varepsilon\to 0$ limit. Equivalently,
    \begin{align*}
        \lim_{\varepsilon\to 0} \overline{q}(z;\varepsilon) &= \frac{1}{1+M^{(0)}\ln z}, \qquad
        \lim_{\varepsilon\to 0} \hat{q}(z;i,\varepsilon) = \frac{\pi^{(0)}_i}{1+M^{(0)}\ln z}.
    \end{align*}
\end{theorem}
Proof deferred to \cref{Full_Thm1}.

Next, we specialize our theorem to specifically focus on birth-death phase processes. We show that when the queue has a birth-death phase-process a set of easier-to-check conditions, rate-boundedness (see \cref{def:rtbdd}) and uniform stability (see \cref{def:unifstab}), implies our light-tailed condition and thus implies the same asymptotic behavior.
\begin{theorem}
\label{Thm2}
In a family of IMAP queues with birth-death phase processes that meets our heavy traffic conditions, if the family is uniformly stable and rate bounded then it also has a light-tailed phase process (\cref{def:light-tailed}). Consequently $\varepsilon Q^{(\varepsilon)}$ is asymptotically exponential with mean $M^{(0)}$ (where $M^{(0)}$ is defined in  \cref{Rdefn}) and asymptotically independent from $Y^{(\varepsilon)}$ in the $\varepsilon\to 0$ limit. Equivalently,
    \begin{align*}
        \lim_{\varepsilon\to 0} \overline{q}(z;\varepsilon) &= \frac{1}{1+M^{(0)}\ln z}, \qquad
        \lim_{\varepsilon\to 0} \hat{q}(z;i,\varepsilon) = \frac{\pi^{(0)}_i}{1+M^{(0)}\ln z}.
    \end{align*}
\end{theorem}
Proof deferred to \cref{Full_Thm2}.

Intuitively, rate-boundedness and uniform stability ensure that the queue isn't ill-behaved at any $\varepsilon$ or in the $\varepsilon\to 0$ limit. Rate-boundedness guarantees that transitions don't occur arbitrarily quickly (to prevent explosivity) or arbitrarily slowly (to prevent any states becoming isolated). Uniform stability guarantees that for ``large enough" states the queue always has a downward transition bias across all $\varepsilon$ values.

Finally, we return to our original system of interest, the tandem queue with hybrid abandonment. We show that a tandem queue with hybrid abandonment meets the birth-death conditions of \cref{Thm2} and satisfies our heavy-traffic conditions and therefore has the desired asymptotic behavior.
\begin{theorem}
\label{Thm3}
     In a tandem queue with hybrid abandonment, as the arrival rate rises to a critical threshold and the second queue's idleness approaches $0$, $\varepsilon Q^{(\varepsilon)}$ is asymptotically exponential with mean $M^{(0)}$ (where $M^{(0)}$ is defined in  \cref{def:HT_cond}) and asymptotically independent from $Y^{(\varepsilon)}$. Equivalently,
         \begin{align*}
        \lim_{\varepsilon\to 0} \overline{q}(z;\varepsilon) &= \frac{1}{1+M^{(0)}\ln z}, \qquad
        \lim_{\varepsilon\to 0} \hat{q}(z;i,\varepsilon) = \frac{\pi^{(0)}_i}{1+M^{(0)}\ln z}.
    \end{align*}
\end{theorem}
Proof deferred to \cref{Full_Thm3}.

\section{Proofs}
\label{Sec5}
We present the proofs for \cref{Thm1,Thm2,Thm3} as well as the necessary lemmas for these proofs.
\medskip

\cref{R_exst_lemma} and \cref{ind_exp_lemma} derive near-inverse matrix and the near-independence and near-exponential equations, respectively, that serve as the bedrock of our asymptotic analyses.

\cref{Thm1} analyzes these near-independence and near-exponential equations to show that if the phase processes of a heavy-traffic family are ``light-tailed" (which we will define in \cref{def:light-tailed} in \cref{Full_Thm1}), i.e., necessary quantities are negligible, then the near-independence and near-exponential equations imply the asymptotic-independence and asymptotic-exponential equations of \cref{def:asymp_exp_ind}.

\medskip
\cref{FundBDBds} provides three fundamental bounds on the quantities $\bpi$ and $\boldsymbol{R}$ when considering birth-death heavy-traffic families that are rate-bounded (\cref{def:rtbdd}) and uniformly stable (\cref{def:unifstab}). \cref{DerBDBds} uses these fundamental bounds to verify that these birth-death heavy-traffic families satisfy the light-tailed phase process conditions necessary for \cref{Thm1}. \cref{Thm2} applies \cref{Thm1} to these birth-death heavy-traffic families to show the desired asymptotic exponential and asymptotic-independence results.

\medskip
\cref{Thm3Sup} proves that hybrid-abandonment families are birth-death heavy-traffic families that are uniformly stable and rate-bounded. \cref{Thm3} then applies \cref{Thm2} to these hybrid-abandonment families to finally reach the asymptotic exponential and asymptotic independence equations for our original hybrid-abandonment system.
\subsection{Existence of the Near-Inverse Matrix}
\label{sec:near-inverse}


We begin by showing that a near-inverse matrix $\boldsymbol{R}$ exists for all the systems we consider in \cref{Thm1,Thm2,Thm3}.
\begin{lemma}
\label{R_exst_lemma}
    For any ergodic continuous time Markov chain $Y$ with generator matrix $\boldsymbol{A}$, there exists a near-inverse matrix $\boldsymbol{R} = [r_{i,j}]_{i,j\in \mathbb{N}_0}$ satisfying
    \[\boldsymbol{A}\boldsymbol{R}= \begin{bmatrix}
        0 & -\frac{1}{\pi_0} \bpi_{\overline{0}}\\
        0& I
    \end{bmatrix}=
    \begin{bmatrix}
        0 & -\frac{\pi_1}{\pi_0} & -\frac{\pi_2}{\pi_0} & \dots &-\frac{\pi_n}{\pi_0} &\dots\\
        0 & 1 & 0 &\dots &0&\dots\\
        0&0&1&\dots&0&\dots\\
        \vdots\\
        0&0&0&\dots &1&\dots\\
        \vdots
    \end{bmatrix}\]
    or, equivalently,
    \[\boldsymbol{A}\boldsymbol{R}x=x-\frac{x_0}{\pi_0}\bpi\]
    for any vector $x$.
\end{lemma}
Note that $\bpi_{\overline{0}}$ is defined in \cref{Rdefn}.

We give a full proof in Appendix~\ref{R_exst_proof}. A sketch of our proof is the following. We begin by constructing the matrix $\boldsymbol{R} = [r_{i,j}]_{i,j\in \mathbb{N}_0}$ by making $r_{i,j}$ be the expected time the chain spends in state $j$ between starting in state $i$ and reaching state $0$ for the first time. We show that when $\boldsymbol{A}$ is ergodic each of these elements exist and then that $\boldsymbol{R}$ satisfies our desired \textit{near-inverse} property.

Intuitively, this near-inverse matrix allows us to mostly cancel out the matrix $\boldsymbol{A}$. This allows us to convert the balance equations of \cref{ind_exp_lemma} into the more useful near-independence equation \cref{near-ind-eq} and near-exponential equation \cref{near-exp-eq}.

In the interest of clarity, we do not rule out the possibility that there exists another matrix $\boldsymbol{R}'\ne \boldsymbol{R}$ that satisfies the near-inverse property. Whenever we invoke $\boldsymbol{R}$ or the near-inverse matrix we consider this formulation of $\boldsymbol{R}$.

\subsection{Near-Independence and Near-Exponential Equations}
With our near-inverse matrix guaranteed to exist by \cref{R_exst_lemma}, we can establish analogous equations to those in \citet{falin1999heavy}. These equations are the basis of our heavy traffic analysis. In particular the most general of the theorems, \cref{Thm1}, requires only that we establish the limiting behavior of each of the terms in the near-independent and near-exponential equations. 

\begin{lemma}
\label{ind_exp_lemma}
    For any infinite-phase Markovian arrival process, if there exists a near-inverse matrix $R$ then the system has the following near-independence and near-exponential equations for any $M>0$:
    \begin{equation}
         \Q\Le(z\Ri) = \Qi\Le(z;0\Ri)\frac{\bpi}{\pi_0} +\Qi\Le(z;0\Ri)\frac{\bpi}{\pi_0}\Le( \boldsymbol{F}\Le(1-z\Ri)+\mu\Le(1-\frac{1}{z}\Ri)\boldsymbol{I}\Ri)R + \overline{\boldsymbol{Y}}\Le(z\Ri)\label{near-ind-eq}       
    \end{equation}
    \begin{equation}
        \Qi\Le(z;0\Ri) = \pi_0\frac{1+U\Le(z\Ri)}{1+M\ln\Le(z\Ri)/\varepsilon +V\Le(z\Ri)}\label{near-exp-eq}
    \end{equation}
    where
    \begin{align*}
                \overline{\boldsymbol{Y}}\Le(z\Ri)
                &=\Le(1-\frac{1}{z}\Ri)\Q\Le(z\Ri)
                \Big(\Le(-z \boldsymbol{F}+\mu \boldsymbol{I}\Ri)\boldsymbol{R}\Big)^2 - \mu\Le(1-\frac{1}{z}\Ri)\piidle\boldsymbol{R}\Le(-z \boldsymbol{F}+\mu \boldsymbol{I}\Ri)\boldsymbol{R} - \mu\piidle\boldsymbol{R}\\
                U\Le(z\Ri) &=-\mu^{-1}\varepsilon^{-1}\Le(1-\frac{1}{z}\Ri)\overline{\boldsymbol{Y}}\Le(z\Ri)\Le(-z \boldsymbol{F}+\mu \boldsymbol{I}\Ri)e^T\\
        V\Le(z\Ri) &= \varepsilon^{-1}\Le(1-\frac{1}{z}\Ri)\Le(-z\frac{\overline{\lambda}}{\mu}+\mu^{-1}\bpi\Le(-z \boldsymbol{F}+\mu \boldsymbol{I}\Ri)\boldsymbol{R}\Le(-z \boldsymbol{F}+\mu \boldsymbol{I}\Ri)e^T\Ri) - M\ln \Le(z\Ri)/\varepsilon.
    \end{align*}
\end{lemma}

We give a full proof in Appendix~\ref{ind_exp_proof}, but here we present a brief proof sketch. By careful manipulation and summation, the balance equation for the Markov chain $(Y, Q)$ can be converted into a balance equation for the partial $z$-transform vector $\hat{\boldsymbol{q}}(z)$, \cref{eqn:qh0}, providing $\hat{\boldsymbol{q}}(z)$ in terms of $\hat{\boldsymbol{q}}(z)$ and $\hat{q}(z;0)$. By substituting \cref{eqn:qh0} for $\hat{\boldsymbol{q}}(z)$ in \cref{eqn:qh0} itself, we now have a new formula for $\hat{\boldsymbol{q}}(z)$ which we use to derive our near-independence and near-exponential equations.

We call these equations \textit{near-independence} and \textit{near-exponential} to specifically draw parallels with ``true-independence" and ``true-exponential" equations,
\begin{equation}
    \hat{\boldsymbol{q}}(z) = \hat{q}(z;0) \frac{\bpi}{\pi_0} \label{eq:true-ind}
\end{equation}
\begin{equation}
    \hat{q}(z;0) = \pi_0 \frac{1}{1+M\ln z} \label{eq:true-exp}.
\end{equation}
\cref{eq:true-exp,eq:true-ind} exactly state that the variable $Q$ is independent of $Y$ and $Q$ has an exponential distribution with mean $M$, respectively. In \cref{Thm1} we show that, under specified conditions, \cref{eq:asymp_ind,eq:asympexp} asymptotically become \cref{eq:true-ind,eq:true-exp}, respectively.

\subsection{Asymptotic Independence and Exponentiality for Light-tailed Modulation Processes}
\label{Full_Thm1}
With our lemmas finished, we are now ready to begin considering the heavy traffic limit. We no longer consider a single system, but a family of systems parameterized by idleness. In our proofs for \cref{Thm1,Thm2,Thm3}, we let $x^{\Le(\varepsilon\Ri)}$ refer to the quantity $x$ in the system with idleness $\varepsilon$. For example, $\mu^{\Le(\varepsilon\Ri)}$ refers to the service rate of the queue with idleness $\varepsilon$ and $\bpi^{\Le(\varepsilon\Ri)}$ refers to the stationary distribution of the phase process in the same system. 

We start with the most general case: considering any convergent heavy-traffic family which satisfies the light-tailed phase process conditions described below. We show that these light-tailed conditions are sufficient to characterize the heavy-traffic limit of the family even when the phase process has a countably infinite state space.
\begin{definition}[Light-Tailed Phase Process]\label{def:light-tailed}
    A family of IMAP queues that meets our heavy-traffic conditions (\cref{def:HT_cond}) with a \emph{light-tailed phase process} is a family of IMAP queues that meets our heavy-traffic conditions where:
    \begin{enumerate}
        \item 
        \label{item:light-tailed-1} For all $z\in\Le(0,1\Ri)$, as $\varepsilon\to 0$:
        \begin{enumerate}
            \item \label{item:light-tailed-1a} $\Le(1-z^{-\varepsilon}\Ri)\Q\Le(z^\varepsilon;\varepsilon\Ri) \Le(\Le(-z^\varepsilon \boldsymbol{F}^{(\varepsilon)}+\mu^{\Le(\varepsilon\Ri)} \boldsymbol{I}\Ri)\boldsymbol{R}^{(\varepsilon)}\Ri)^2\Le(-z^\varepsilon \boldsymbol{F}^{(\varepsilon)}+\mu \boldsymbol{I}\Ri)e^T\to 0$

            \item \label{item:light-tailed-1b}$\mu^{\Le(\varepsilon\Ri)}\Le(1-z^{\varepsilon}\Ri)\bpi^{\Le(\varepsilon\Ri)}_0\boldsymbol{R}^{(\varepsilon)}\Le(-z^\varepsilon \boldsymbol{F}^{(\varepsilon)}+\mu^{\Le(\varepsilon\Ri)} \boldsymbol{I}\Ri)\boldsymbol{R}^{(\varepsilon)}\Le(-z^\varepsilon \boldsymbol{F}^{(\varepsilon)}+\mu^{(\varepsilon)} \boldsymbol{I}\Ri)e^T\to 0$
            
            \item \label{item:light-tailed-1c}$\bpi^{\Le(\varepsilon\Ri)}_0\boldsymbol{R}^{(\varepsilon)}\Le(-z^\varepsilon \boldsymbol{F}^{(\varepsilon)}+\mu^{(\varepsilon)} \boldsymbol{I}\Ri)e^T\to 0$ 

            \item \label{item:light-tailed-1d}$\Le(1-z^{-\varepsilon}\Ri)\bpi^{\Le(\varepsilon\Ri)}\Le(\mu^{\Le(\varepsilon\Ri)} \boldsymbol{I} - z^{\varepsilon}  \boldsymbol{F}^{(\varepsilon)}\Ri)\boldsymbol{R}^{(\varepsilon)}e^T\to 0$

            \item \label{item:light-tailed-1e} $\Le(1-z^{-\varepsilon}\Ri)\Q\Le(z^\varepsilon;\varepsilon\Ri)\Le(\Le(\mu^{\Le(\varepsilon\Ri)}I-z \boldsymbol{F}^{(\varepsilon)}\Ri)\boldsymbol{R}^{(\varepsilon)}\Ri)^2e^T\to 0$

            \item \label{item:light-tailed-1f}$\mu^{\Le(\varepsilon\Ri)}\Le(1-z^{-\varepsilon}\Ri)\piidle^{\Le(\varepsilon\Ri)}\boldsymbol{R}^{(\varepsilon)}\Le(\mu^{\Le(\varepsilon\Ri)}I-z^{\varepsilon} \boldsymbol{F}^{(\varepsilon)}\Ri)\boldsymbol{R}^{(\varepsilon)}e^T\to 0$

            \item \label{item:light-tailed-1g}$\mu^{\Le(\varepsilon\Ri)}\piidle^{\Le(\varepsilon\Ri)}\boldsymbol{R}^{(\varepsilon)}e^T\to 0$
        \end{enumerate}

        \item \label{def_to0i}For all $z\in\Le(0,1\Ri)$ and all $i\in \mathbb{N}_0$ as $\varepsilon\to 0$:
        \begin{enumerate}
            \item \label{item:light-tailed-2a}$[\Le(1-z^{-\varepsilon}\Ri)\boldsymbol{\pi}^{\Le(\varepsilon\Ri)}\Le(\mu^{\Le(\varepsilon\Ri)}I-z \boldsymbol{F}^{(\varepsilon)}\Ri)\boldsymbol{R}^{(\varepsilon)}]_i\to 0$

            \item \label{item:light-tailed-2b}$[\Q\Le(z^\varepsilon;\varepsilon\Ri) \Le(\Le(-z^\varepsilon \boldsymbol{F}^{(\varepsilon)}+\mu^{\Le(\varepsilon\Ri)} \boldsymbol{I}\Ri)R\Ri)^2]_i \to 0$

            \item \label{item:light-tailed-2c}$[\mu^{\Le(\varepsilon\Ri)}\Le(1-z^{-\varepsilon}\Ri)\bpi^{\Le(\varepsilon\Ri)}_0\boldsymbol{R}^{(\varepsilon)}\Le(-z^\varepsilon \boldsymbol{F}^{(\varepsilon)} + \mu^{\Le(\varepsilon\Ri)}I\Ri)\boldsymbol{R}^{(\varepsilon)}]_i\to 0$

            \item \label{item:light-tailed-2d}$[\mu^{\Le(\varepsilon\Ri)}\bpi^{\Le(\varepsilon\Ri)}_0\boldsymbol{R}^{(\varepsilon)}]_i\to 0$
        \end{enumerate}
    \end{enumerate}
\end{definition}
Each of these terms appear in at least one of $U(z^\varepsilon)$, $V(z^\varepsilon)$, $\overline{\boldsymbol{Y}}(z^\varepsilon)$, or $\overline{Y}_i(z^\varepsilon)$. With all of these terms being negligible, we can conclude (as we do in \cref{Thm1}) that we have our desired heavy traffic results.
\begin{reptheorem}{Thm1}
    If a family of IMAP queues with a heavy-traffic convergent phase process has a light-tailed phase process, then $\varepsilon Q^{(\varepsilon)}$ is asymptotically exponential with mean $M^{(0)}$ (where $M^{(0)}$ is defined in \cref{def:HT_cond}) and asymptotically independent from $Y^{(\varepsilon)}$ in the $\varepsilon\to 0$ limit.

    Equivalently,
    \begin{align}
        \lim_{\varepsilon\to 0} \overline{q}(z^\varepsilon;\varepsilon) &= \frac{1}{1+M^{(0)}\ln z}\label{eq:Thm1-res-a}\\
        \lim_{\varepsilon\to 0} \hat{q}(z^\varepsilon;i,\varepsilon) &= \frac{\pi^{(0)}_i}{1+M^{(0)}\ln z} \label{eq:Thm1-res-b}
    \end{align}
\end{reptheorem}
\begin{proof}
    We now shift our focus to the scaled system $\varepsilon Q^{(\varepsilon)}$ rather than the original system $Q^{(\varepsilon)}$. Notice that 
    \begin{align*}
        \E[z^{\varepsilon Q^{(\varepsilon)}}] &= \E[(z^\varepsilon)^{Q^{(\varepsilon)}}]= \overline{q}(z^\varepsilon;\varepsilon)
    \end{align*}
    and
    \begin{align*}
        \E\Le[z^{\varepsilon Q^{(\varepsilon)}}\indic\{Y^{(\varepsilon)} = i\}\Ri] &= \E\Le[(z^\varepsilon)^{Q^{(\varepsilon)}}\indic\{Y^{(\varepsilon)} = i\}\Ri]= \hat{q}(z^\varepsilon;i, \varepsilon).
    \end{align*}
    Therefore, our analysis of $\varepsilon Q^{(\varepsilon)}$ can rely on our existing near-independence and near-exponential equations \cref{near-ind-eq,near-exp-eq} for $Q^{(\varepsilon)}$.
    
    We break our result into three parts.
    In part 1, we show that $[\varepsilon Q^{(\varepsilon)} | Y^{(\varepsilon)} = 0]$ is asymptotically exponential or  
    \begin{align}
        \lim_{\varepsilon\to 0} \hat{q}(z^\varepsilon;0,\varepsilon) &= \frac{\pi^{(0)}_0}{1+M^{(0)}\ln z}. \label{eq:Thm-1-0res}
    \end{align}
    Next, in part 2, we show this for any $i\in \mathbb{N}_0,$ $[\varepsilon Q^{(\varepsilon)} | Y^{(\varepsilon)} = i]$ is asymptotically exponential or  
    \begin{align}
        \lim_{\varepsilon\to 0} \hat{q}(z^\varepsilon;i,\varepsilon) &= \frac{\pi^{(0)}_i}{1+M^{(0)}\ln z}.\tag{\ref{eq:Thm1-res-b}}
    \end{align}
    Finally, in part 3,  we show that $\varepsilon Q^{(\varepsilon)}$ (absent any conditioning on $Y^{(\varepsilon)}$) is asymptotically exponential or
    \begin{align*}
        \lim_{\varepsilon\to 0} \overline{q}(z^\varepsilon;\varepsilon) &= \frac{1}{1+M^{(0)}\ln z}. \tag{\ref{eq:Thm1-res-a}}
    \end{align*}
    \subsubsection{Part 1: Asymptotic Exponentiality Conditioned on State 0}

    First, we show that $[\varepsilon Q^{(\varepsilon)}|Y^{(0)}=0]$ is asymptotically exponential by proving \cref{eq:Thm-1-0res}.
    
    By \cref{ind_exp_lemma}, for any member of the parameterized family we have
    \begin{align*}
        \Q\Le(z;\varepsilon\Ri) &= \Qi\Le(z;0,\varepsilon\Ri)\frac{\bpi^{\Le(\varepsilon\Ri)}}{\pi_0^{\Le(\varepsilon\Ri)}} +\Qi\Le(z;0,\varepsilon\Ri)\frac{\bpi^{\Le(\varepsilon\Ri)}}{\pi_0^{\Le(\varepsilon\Ri)}}\Le( \boldsymbol{F}^{(\varepsilon)}\Le(1-z\Ri)+\mu^{\Le(\varepsilon\Ri)}\Le(1-\frac{1}{z}\Ri)\boldsymbol{I}\Ri)\boldsymbol{R}^{(\varepsilon)} + \overline{\boldsymbol{Y}}\Le(z;\varepsilon\Ri)\\
        \Qi\Le(z;0,\varepsilon\Ri) &= \pi_0^{\Le(\varepsilon\Ri)}\frac{1+U\Le(z;\varepsilon\Ri)}{1+M^{(0)}\ln \Le(z\Ri)/\varepsilon +V\Le(z;\varepsilon\Ri)}
    \end{align*}

    We begin by showing $U\Le(z^\varepsilon;\varepsilon\Ri) \to 0$. We can consider $U\Le(z^\varepsilon;\varepsilon\Ri)$ in two parts:
    \begin{align*}
        U\Le(z^\varepsilon;\varepsilon\Ri)&=\Le(-\frac{1-z^{-\varepsilon}}{\varepsilon\mu^{\Le(\varepsilon\Ri)}}\Ri)\cdot\Le(\overline{\boldsymbol{Y}}\Le(z^\varepsilon;\varepsilon\Ri)\Le(-z^\varepsilon \boldsymbol{F}^{(\varepsilon)}+\mu^{(\varepsilon)} \boldsymbol{I}\Ri)e^T\Ri)
    \end{align*}
    
    We analyze the two multiplicands separately. Because $-\frac{1-z^{-\varepsilon}}{\varepsilon\mu^{\Le(\varepsilon\Ri)}}\to -\frac{\ln\Le(z\Ri)}{\mu^{\Le(\varepsilon\Ri)}}$ and $|-\frac{\ln\Le(z\Ri)}{\mu^{\Le(\varepsilon\Ri)}}|<\infty$ for any $z>0$, if $\overline{\boldsymbol{Y}}\Le(z^\varepsilon;\varepsilon\Ri)\Le(-z^\varepsilon \boldsymbol{F}^{(\varepsilon)}+\mu^{(\varepsilon)} \boldsymbol{I}\Ri)e^T\to 0$ then $U\Le(z^\varepsilon;\varepsilon\Ri)\to 0$ as well.

    Expanding $\overline{\boldsymbol{Y}}\Le(z^\varepsilon;\varepsilon\Ri)\Le(-z^\varepsilon \boldsymbol{F}^{(\varepsilon)}+\mu^{(\varepsilon)} \boldsymbol{I}\Ri)e^T$, we have 3 distinct terms:
    \begin{align}
        &\overline{\boldsymbol{Y}}\Le(z^\varepsilon;\varepsilon\Ri)\Le(-z^\varepsilon \boldsymbol{F}^{(\varepsilon)}+\mu^{(\varepsilon)} \boldsymbol{I}\Ri)e^T\nonumber\\
        &=\Le(1-z^{-\varepsilon}\Ri)\Q\Le(z^\varepsilon;\varepsilon\Ri) \Le(\Le(-z^\varepsilon \boldsymbol{F}^{(\varepsilon)}+\mu^{\Le(\varepsilon\Ri)} \boldsymbol{I}\Ri)R\Ri)^2\Le(-z^\varepsilon \boldsymbol{F}^{(\varepsilon)}+\mu^{(\varepsilon)} \boldsymbol{I}\Ri)e^T\label{eq:Y_FI_e_expdd_1}\\
        &-\mu^{\Le(\varepsilon\Ri)}\Le(1-z^{\varepsilon}\Ri)\bpi^{\Le(\varepsilon\Ri)}_0\boldsymbol{R}^{(\varepsilon)}\Le(-z^\varepsilon \boldsymbol{F}^{(\varepsilon)}+\mu^{\Le(\varepsilon\Ri)} \boldsymbol{I}\Ri)\boldsymbol{R}^{(\varepsilon)}\Le(-z^\varepsilon \boldsymbol{F}^{(\varepsilon)}+\mu^{(\varepsilon)} \boldsymbol{I}\Ri)e^T \label{eq:Y_FI_e_expdd_2}\\
        &-\mu^{\Le(\varepsilon\Ri)}\bpi^{\Le(\varepsilon\Ri)}_0\boldsymbol{R}^{(\varepsilon)}\Le(-z^\varepsilon \boldsymbol{F}^{(\varepsilon)}+\mu^{(\varepsilon)} \boldsymbol{I}\Ri)e^T \label{eq:Y_FI_expdd_3}
    \end{align}

    Because the family of IMAP queues has a light-tailed phase process, Properties~\ref{item:light-tailed-1a}, \ref{item:light-tailed-1b}, and \ref{item:light-tailed-1c} of \cref{def:light-tailed} guarantee that each of \cref{eq:Y_FI_e_expdd_1,eq:Y_FI_e_expdd_2,eq:Y_FI_expdd_3}, respectively, are negligible as $\varepsilon\to 0$.

    Thus, $U(z^\varepsilon;\varepsilon) \to 0$ as $\varepsilon\to 0$.

    Next, let us show $V\Le(z^\varepsilon;\varepsilon\Ri)\to 0$. It is useful to reformulate $V\Le(z^\varepsilon; \varepsilon\Ri)$ by adding and subtracting $\varepsilon^{-1}\Le(1-z^{-\varepsilon}\Ri)$ and then carefully rearranging terms.
    \begin{align}
        &V\Le(z^\varepsilon; \varepsilon\Ri) = \varepsilon^{-1}\Le(1-z^{-\varepsilon}\Ri)\Le(-z^\varepsilon\frac{\overline{\lambda}^{\Le(\varepsilon\Ri)}}{\mu^{\Le(\varepsilon\Ri)}}\Ri)\nonumber\\
        &+ \varepsilon^{-1}\Le(1-z^{-\varepsilon}\Ri)\Le(\Le(\mu^{\Le(\varepsilon\Ri)}\Ri)^{-1}\bpi^{\Le(\varepsilon\Ri)}\Le(-z^\varepsilon \boldsymbol{F}^{(\varepsilon)}+\mu^{(\varepsilon)} \boldsymbol{I}^{\Le(\varepsilon\Ri)}\Ri)\boldsymbol{R}^{(\varepsilon)}\Le(-z^\varepsilon \boldsymbol{F}^{(\varepsilon)}+\mu^{(\varepsilon)} \boldsymbol{I}\Ri)e^T\Ri)\nonumber \\
        &- M^{(0)}\ln\Le(z\Ri)\nonumber\\
        &=\varepsilon^{-1}\Le(1-z^{-\varepsilon}\Ri) \Le(1-1-z^{\varepsilon}\frac{\lambdabar^{\Le(\varepsilon\Ri)}}{\mu^{\Le(\varepsilon\Ri)}}\Ri)\nonumber \\
        &+ \varepsilon^{-1}\Le(1-z^{-\varepsilon}\Ri)\Le(-1+\Le(\mu^{\Le(\varepsilon\Ri)}\Ri)^{-1}\bpi^{\Le(\varepsilon\Ri)}\Le(-z^\varepsilon \boldsymbol{F}^{(\varepsilon)}+\mu^{\Le(\varepsilon\Ri)} \boldsymbol{I}\Ri)\boldsymbol{R}^{(\varepsilon)}\Le(-z^\varepsilon \boldsymbol{F}^{(\varepsilon)}+\mu^{\Le(\varepsilon\Ri)} \boldsymbol{I}\Ri)e^T\Ri) \nonumber \\\
        &- M^{(0)}\ln\Le(z\Ri)\nonumber\\
        &=\varepsilon^{-1}\Le(1-z^{-\varepsilon}\Ri) \Le(1-z^{\varepsilon}\frac{\lambdabar^{\Le(\varepsilon\Ri)}}{\mu^{\Le(\varepsilon\Ri)}}\Ri) \label{eq:V-expdd-1}\\
        &+ \varepsilon^{-1}\Le(1-z^{-\varepsilon}\Ri)\Le(-1+\Le(\mu^{\Le(\varepsilon\Ri)}\Ri)^{-1}\bpi^{\Le(\varepsilon\Ri)}\Le(-z^\varepsilon \boldsymbol{F}^{(\varepsilon)}+\mu^{\Le(\varepsilon\Ri)} \boldsymbol{I}\Ri)\boldsymbol{R}^{(\varepsilon)}\Le(-z^\varepsilon \boldsymbol{F}^{(\varepsilon)}+\mu^{\Le(\varepsilon\Ri)} \boldsymbol{I}\Ri)e^T\Ri) \label{eq:V-expdd-2}\\
        &- M^{(0)}\ln\Le(z\Ri) \label{eq:V-expdd-3}
    \end{align}
    Let us consider \cref{eq:V-expdd-1} separately from the other terms of $V\Le(z^\varepsilon;\varepsilon\Ri)$, \cref{eq:V-expdd-2,eq:V-expdd-3}. 

    Beginning with \cref{eq:V-expdd-1}, we note that $\varepsilon^{-1}\Le(1-z^{-\varepsilon}\Ri)\to-\ln z$ (where $|-\ln z| < \infty$ for all $z\in (0,1)$) as $\varepsilon\to 0$. Next, using the fact that $\frac{\lambdabar^{(\varepsilon)}}{\mu^{(\varepsilon)}} = (1-\varepsilon)$ we can see that $\Le(1-z^{\varepsilon}\frac{\lambdabar^{\Le(\varepsilon\Ri)}}{\mu^{\Le(\varepsilon\Ri)}}\Ri) = \Le(1-z^{\varepsilon}(1-\varepsilon)\Ri)\to 0$ as $\varepsilon\to 0$. Thus, their product, \cref{eq:V-expdd-1}, is negligible as $\varepsilon\to 0$.

    Next, we move onto \cref{eq:V-expdd-2,eq:V-expdd-3}. We again note that $\varepsilon^{-1}\Le(1-z^{-\varepsilon}\Ri)\to-\ln z$. By property \ref{M_conv} of the heavy-traffic family conditions (\cref{def:HT_cond}), we can see that
    \[-1+\Le(\mu^{\Le(\varepsilon\Ri)}\Ri)^{-1}\bpi^{\Le(\varepsilon\Ri)}\Le(-z^\varepsilon \boldsymbol{F}^{(\varepsilon)}+\mu^{\Le(\varepsilon\Ri)} \boldsymbol{I}\Ri)\boldsymbol{R}^{(\varepsilon)}\Le(-z^\varepsilon \boldsymbol{F}^{(\varepsilon)}+\mu^{\Le(\varepsilon\Ri)} \boldsymbol{I}\Ri)e^T\to -M^{(0)}\]

    Thus, \cref{eq:V-expdd-2} approaches $M^{(0)}\ln z$ in the $\varepsilon\to 0$ limit and \cref{eq:V-expdd-2,eq:V-expdd-3} cancel in this limit.

    We therefore conclude that $V\Le(z^\varepsilon;\varepsilon\Ri)\to 0$.

    Having shown that $U\Le(z^\varepsilon;\varepsilon\Ri), V\Le(z^\varepsilon;\varepsilon\Ri)\to 0$, we have our desired result: \begin{equation}
        \Qi\Le(z^\varepsilon;0,\varepsilon\Ri) = \pi_0^{\Le(\varepsilon\Ri)} \frac{1+U\Le(z^\varepsilon;\varepsilon\Ri)}{1+M^{(0)}\ln z - V\Le(z^\varepsilon;\varepsilon\Ri)} \to \frac{\pi_0^{\Le(0\Ri)}}{1+M^{(0)}\ln z} \tag{\ref{eq:Thm-1-0res}}
    \end{equation}

    We have thus shown that the scaled queue length process is asymptotically exponential when conditioned on $Y^{\Le(\varepsilon\Ri)}=0$.

    \subsubsection{Part 2: Asymptotic Exponentiality Conditioned on State $i$}

    Next, we show that the scaled queue length process is asymptotically exponential when conditioned on $Y^{(\varepsilon)}=i$ any $i\in \mathbb{N}_0$.
    
    Using \cref{near-ind-eq}, for any $i\in \mathbb{N}_0$, we can express the $i^{\text{th}}$ partial $z$-transform as
     \begin{align}
         \Qi\Le(z^\varepsilon;i,\varepsilon\Ri) &= \Qi\Le(z^\varepsilon;0,\varepsilon\Ri)\frac{\pi^{\Le(\varepsilon\Ri)}_i}{\pi_0^{\Le(\varepsilon\Ri)}}
         \label{eq:Qii_expdd_1}\\
         &+\Le[\Qi\Le(z^\varepsilon;0,\varepsilon\Ri)\frac{\bpi^{\Le(\varepsilon\Ri)}}{\pi_0^{\Le(\varepsilon\Ri)}}\Le( \boldsymbol{F}^{(\varepsilon)}\Le(1-z^\varepsilon\Ri)+\mu^{\Le(\varepsilon\Ri)}\Le(1-z^{-\varepsilon}\Ri)\boldsymbol{I}\Ri)\boldsymbol{R}^{(\varepsilon)}\Ri]_i
         \label{eq:Qii_expdd_2}\\
         &+ \overline{Y}^i\Le(z^\varepsilon;\varepsilon\Ri)
         \label{eq:Qii_expdd_3)}.
     \end{align}

     If we can show that \cref{eq:Qii_expdd_2,eq:Qii_expdd_3)} in our expression for $\Qi\Le(z^\varepsilon;i,\varepsilon\Ri)$ are asymptotically negligible then we have our desired asymptotic independence result. Let us focus on the terms separately.
     We begin with \cref{eq:Qii_expdd_2}
     \begin{align*}
         \Le[\Qi\Le(z^\varepsilon;0,\varepsilon\Ri)\frac{\bpi^{\Le(\varepsilon\Ri)}}{\pi_0^{\Le(\varepsilon\Ri)}}\Le( \boldsymbol{F}^{(\varepsilon)}\Le(1-z^\varepsilon\Ri)+\mu^{\Le(\varepsilon\Ri)}\Le(1-z^{-\varepsilon}\Ri)\boldsymbol{I}\Ri)\boldsymbol{R}^{(\varepsilon)}\Ri]_i \\
         \quad= \frac{\Qi\Le(z^\varepsilon;0,\varepsilon\Ri)}{\pi_0^{\Le(\varepsilon\Ri)}}\cdot\Le[\Le(1-z^{-\varepsilon}\Ri)\bpi^{\Le(\varepsilon\Ri)}\Le(\mu^{\Le(\varepsilon\Ri)}I-z^\varepsilon \boldsymbol{F}^{(\varepsilon)}\Ri)\boldsymbol{R}^{(\varepsilon)}\Ri]_i
     \end{align*}

     Because $\frac{\Qi\Le(z^\varepsilon;0,\varepsilon\Ri)}{\pi_0^{\Le(\varepsilon\Ri)}}\to \frac{1}{1-M^{(0)}\ln z}$ and, by property \ref{item:light-tailed-2a} of our light-tailed phase process conditions, $\Le[\Le(1-z^{-\varepsilon}\Ri)\bpi^{\Le(\varepsilon\Ri)}\Le(\mu^{\Le(\varepsilon\Ri)}I-z^\varepsilon \boldsymbol{F}^{(\varepsilon)}\Ri)\boldsymbol{R}^{(\varepsilon)}\Ri]_i\to 0$, the second summand approaches $0$.

     Moving onto the third summand, we can expand $\overline{Y}^i\Le(z^\varepsilon;\varepsilon\Ri)$ to
     \begin{align}
        \overline{Y}^i\Le(z^\varepsilon;\varepsilon\Ri) 
        &=\Le[\Le(1-z^{-\varepsilon}\Ri)\Q\Le(z^\varepsilon;\varepsilon\Ri) \Le(\Le(-z^{\varepsilon} \boldsymbol{F}^{(\varepsilon)}+\mu^{\Le(\varepsilon\Ri)} \boldsymbol{I}\Ri)\boldsymbol{R}^{(\varepsilon)}\Ri)^2\Ri]_i \\
        &\quad- \Le[\mu^{\Le(\varepsilon\Ri)}\Le(1-z^{-\varepsilon}\Ri)\piidle^{\Le(\varepsilon\Ri)}\boldsymbol{R}^{(\varepsilon)}\Le(-z^{\varepsilon} \boldsymbol{F}^{(\varepsilon)}+\mu^{\Le(\varepsilon\Ri)} \boldsymbol{I}\Ri)R\Ri]_i \\
        &\quad - \Le[\mu^{\Le(\varepsilon\Ri)}\piidle^{\Le(\varepsilon\Ri)}\boldsymbol{R}^{(\varepsilon)}\Ri]_i
     \end{align}

     The light-tailed phase process conditions, specifically properties \ref{item:light-tailed-2a}, \ref{item:light-tailed-2b}, and \ref{item:light-tailed-2c}, provide that each of the three parts above go to $0$ as $\varepsilon\to 0$ and therefore $\overline{Y}^i\Le(z^\varepsilon;\varepsilon\Ri) \to 0$ as well.

     With these two terms shown to be negligible, we have our desired asymptotic independence result:
     For each $i$,
     \begin{equation} \label{eq:thm1-res-a}
     \lim_{\varepsilon\to 0}\Qi\Le(z^\varepsilon;i,\varepsilon\Ri) = \frac{\pi_i^{\Le(0\Ri)}}{1+M^{(0)}\ln z} \tag{\ref{eq:Thm1-res-a}}
     \end{equation}
     
    This expression exactly states that $\varepsilon Q^{\Le(\varepsilon\Ri)}$ and $Y^{\Le(\varepsilon\Ri)}$ are asymptotically independent and ${[\varepsilon Q^{\Le(\varepsilon\Ri)}|Y^{(0)}=i]}$ is asymptotically exponential with mean $M^{(0)}$.

    \subsubsection{Part 3: Asymptotic Exponentiality Without Conditioning}

    Because we have infinitely many modulation states $i$, our prior results do not immediately prove that $\varepsilon Q^{(\varepsilon)}$ is asymptotically exponential. In particular, if $\hat{q}(z^\varepsilon;i,\varepsilon)$ converges slower and slower for larger $i$, $\varepsilon Q^{(\varepsilon)}$ could fail to converge to the exponential distribution we'd otherwise expect it to. Thus, we show explicitly that the convergence to the exponential still occurs absent any conditioning.

    We begin by expressing $\overline{q}(z^\varepsilon;\varepsilon)$ in terms of $\hat{q}(z^\varepsilon;i,\varepsilon)$ and applying our near-independence equation (\cref{near-ind-eq}).

    \begin{align}
        \overline{q}(z^\varepsilon;\varepsilon) &= \sum_{i\in \mathbb{N}_0} \hat{q}(z^\varepsilon;\varepsilon, i)= \hat{q}(z^\varepsilon;\varepsilon)e^T\nonumber\\
        &=\Qi\Le(z^\varepsilon;0,\varepsilon\Ri)\frac{\bpi^{\Le(\varepsilon\Ri)}}{\pi_0^{\Le(\varepsilon\Ri)}}e^T \label{eq:532-1}\\
        &+\Qi\Le(z^\varepsilon;0,\varepsilon\Ri)\frac{\bpi^{\Le(\varepsilon\Ri)}}{\pi_0^{\Le(\varepsilon\Ri)}}\Le( \boldsymbol{F}^{(\varepsilon)}\Le(1-z^\varepsilon\Ri)+\mu^{\Le(\varepsilon\Ri)}\Le(1-z^{-\varepsilon}\Ri)\boldsymbol{I}\Ri)\boldsymbol{R}^{(\varepsilon)}e^T \label{eq:532-2}\\
        &+ \overline{\boldsymbol{Y}}\Le(z^\varepsilon;\varepsilon\Ri)e^T \label{eq:532-3}
    \end{align}

    Let us begin with \cref{eq:532-2}.
    \[\frac{\Qi\Le(z^\varepsilon;0,\varepsilon\Ri)}{\pi_0^{\Le(\varepsilon\Ri)}}\cdot\bpi^{\Le(\varepsilon\Ri)}\Le( \boldsymbol{F}^{(\varepsilon)}\Le(1-z^\varepsilon\Ri)+\mu^{\Le(\varepsilon\Ri)}\Le(1-z^{-\varepsilon}\Ri)\boldsymbol{I}\Ri)\boldsymbol{R}^{(\varepsilon)}e^T\]
    
    We can notice that $\frac{\Qi\Le(z^\varepsilon;0,\varepsilon\Ri)}{\pi_0^{\Le(\varepsilon\Ri)}} \to \frac{1}{1-M^{(0)}\ln z}$ (with $|\frac{1}{1-M^{(0)}\ln z}|<\infty$) by \cref{eq:thm1-res-a} and that the Property \ref{item:light-tailed-1d} of the light-tailed phase process (\cref{def:light-tailed}) conditions guarantee 
    \begin{align}
        {\bpi^{\Le(\varepsilon\Ri)}\Le( \boldsymbol{F}^{(\varepsilon)}\Le(1-z^\varepsilon\Ri)+\mu^{\Le(\varepsilon\Ri)}\Le(1-z^{-\varepsilon}\Ri)\boldsymbol{I}\Ri)\boldsymbol{R}^{(\varepsilon)}e^T\to 0}. 
    \end{align}
    Together, these imply that this second term is asymptotically negligible.

    Focusing on the \cref{eq:532-3},
    \begin{align}
        \overline{\boldsymbol{Y}}\Le(z^\varepsilon;\varepsilon\Ri)e^T &=\Le(1-z^{-\varepsilon}\Ri)\Q\Le(z^\varepsilon;\varepsilon\Ri) \Le(\Le(-z^{\varepsilon} \boldsymbol{F}^{(\varepsilon)}+\mu^{\Le(\varepsilon\Ri)} \boldsymbol{I}\Ri)\boldsymbol{R}^{(\varepsilon)}\Ri)^2e^T \label{eq:Ybar-e-1}\\
        &\quad- \mu^{\Le(\varepsilon\Ri)}\Le(1-z^{-\varepsilon}\Ri)\piidle^{\Le(\varepsilon\Ri)}\boldsymbol{R}^{(\varepsilon)}\Le(-z^{\varepsilon} \boldsymbol{F}^{(\varepsilon)}+\mu^{\Le(\varepsilon\Ri)} \boldsymbol{I}\Ri)\boldsymbol{R}^{(\varepsilon)}e^T \label{eq:Ybar-e-2}\\
        &\quad - \mu^{\Le(\varepsilon\Ri)}\piidle^{\Le(\varepsilon\Ri)}\boldsymbol{R}^{(\varepsilon)}e^T \label{eq:Ybar-e-3}
     \end{align}
     Each of \cref{eq:Ybar-e-1,eq:Ybar-e-2,eq:Ybar-e-3} approach to $0$ as $\varepsilon\to 0$ by properties \ref{item:light-tailed-1e}, \ref{item:light-tailed-1f}, and \ref{item:light-tailed-1g} of \cref{def:light-tailed}. Thus, \cref{eq:532-2}, like \cref{eq:532-3}, is asymptotically negligible and we need only consider the asymptotic behavior of \cref{eq:532-1}.

     In particular, because $\bpi^{(\varepsilon)} e^T =1$ for all $\varepsilon$,
     \begin{align*}
         \Qi\Le(z^\varepsilon;0,\varepsilon\Ri)\frac{\bpi^{\Le(\varepsilon\Ri)}}{\pi_0^{\Le(\varepsilon\Ri)}}e^T &= \frac{\Qi\Le(z^\varepsilon;0,\varepsilon\Ri)}{\pi_0^{\Le(\varepsilon\Ri)}}
     \end{align*}
     and therefore
     \begin{align}
         \lim_{\varepsilon\to 0} \overline{q}\Le(z^\varepsilon;\varepsilon\Ri) &= \lim_{\varepsilon\to 0} \frac{\Qi\Le(z^\varepsilon;0,\varepsilon\Ri)}{\pi_0^{\Le(\varepsilon\Ri)}}= \frac{1}{1+M^{(0)}\ln z} \tag{\ref{eq:Thm1-res-a}}
     \end{align}
     With \cref{eq:thm1-res-a} proven, we now have completed our proof for \cref{Thm1}: that $\varepsilon Q$ is asymptotically exponentially distributed and is asymptotically independent of $Y$.
\end{proof}

\subsection{Birth-Death Modulation Process Fundamental Bounds}
Before we explicitly consider the asymptotic behavior of any birth-death families, we establish a few fundamental bounds on the elements in $\boldsymbol{R}^{(\varepsilon)}$ and $\bpi^{\Le(\varepsilon\Ri)}$ in \Cref{FundBDBds}. These form the backbone of \cref{DerBDBds} and, through it, \cref{Thm2}.

    \begin{lemma}
        \label{FundBDBds}
        A birth-death family that meets our heavy traffic conditions and is uniformly stable and rate-bounded has the following bounds for some $C_1, C_2, C_3$:
        \begin{enumerate}[(a)]
        \item $\pi^{\Le(\varepsilon\Ri)}_n \leq C_1a^n$ for all $n\in \mathbb{N}, \varepsilon\in \mathcal{E}$ \label{eq:birth-bd-1}
        \item $r^{\Le(\varepsilon\Ri)}_{i,j} \leq C_2 \tilde{b}^{\Le(j-i\Ri)^+}$ where $\tilde{b}:=\frac{b}{1-b}<1$ for all $i,j\in \mathbb{N}_0, \varepsilon\in \mathcal{E}$ \label{eq:birth-bd-2}
        \item $\Le|[\boldsymbol{R}^{(\varepsilon)}e^T]_i\Ri| \leq C_3 i$ for all $i\in \mathbb{N}_0, \varepsilon\in \mathcal{E}$ \label{eq:birth-bd-3}
        \end{enumerate}
        where $a,b$ are defined in the uniform stability definition (\cref{def:unifstab}).
    \end{lemma}
    \begin{proof}
    Let us prove the three bounds in order.
    \begin{enumerate}[(a)]
        \item $\pi^{\Le(\varepsilon\Ri)}_n \leq C_1a^n$.
        \label{pi_bd}

        The bounded rate conditions (\cref{def:rtbdd}) guarantee that there is some $a\in (0,1)$ and $N$ such that for all $n\geq N$, $\frac{\alpha^{\Le(\varepsilon\Ri)}_{n, n+1}}{\alpha^{\Le(\varepsilon\Ri)}_{n+1,n}}\leq a$ and therefore ${\pi^{\Le(\varepsilon\Ri)}_n = \pi_N^{\Le(\varepsilon\Ri)} \prod_{m=N}^n \frac{\alpha^{\Le(\varepsilon\Ri)}_{n, n+1}}{\alpha^{\Le(\varepsilon\Ri)}_{n+1,n}} \leq \pi^{\Le(\varepsilon\Ri)}_N a^{n-N}}$. Letting $C_1= a^{-N}$, this bound trivially holds for all for all $n\leq N$ and we have our desired bound for all $n$.

        This bound also holds for $\pi^{(\varepsilon)}_{n,0}$ and $\Qi(z;n)$ as $\pi^{(\varepsilon)}_{n,0}\leq \pi^{(\varepsilon)}_n$ and $\Qi(z;n)\leq \pi^{(\varepsilon)}_n$.

        \item $r^{\Le(\varepsilon\Ri)}_{i,j} \leq C_2 \tilde{b}^{\Le(j-i\Ri)^+}$.
        \label{rijbd}

        It is useful to consider this bound in three cases: when $i= j$, when $i\geq j$, and when $i<j$.
        
        To simplify notation we let $\p(x_1\to x_2\to \dots \to x_n/y_1,y_2, \dots, y_m)$ be the probability that the phase process begins in state $x_1$ and transitions (over any number of steps) to state $x_2$ \dots and then finally to state $x_n$ without entering states $y_1, y_2,\dots, y_m$.
        \begin{enumerate}

            \item Case 1: $i= j$.
            
            Starting in state $i$ there is some probability, $p_i$, that the process reaches state $0$ before returning to state $i$. The number of visits to state $i$ is distributed as a shifted geometric variable, $K_{i,i}\sim\text{Geom}\Le(p_i\Ri)$ (where $K_{i,i} \in \{1, 2, \dots\}$). Thus if we can uniformly lower bound $p_i$ for all $i$ and $\varepsilon$ then we can uniformly upper bound $r^{(\varepsilon)}_{i,i}$.

            Let us start with the case when $i\leq N$. In this case we can naively bound the probability of moving from state $i$ to state $0$ without returning to $i$ by the probability of transitioning downward $i$ times in a row: ${p_{i\downarrow0}:=\prod_{m=1}^i \frac{\alpha_{m,m-1}}{\alpha_{m,m-1}+\alpha_{m,m+1}}}$.

            When $i > N$ we can break the process of moving from $i$ to $0$ into two components: moving from $i$ to $N$ without returning to $i$ and then moving from $N$ to $0$ without returning to $i$. We already have a lower bound on the probability of transitioning from $N$ to $0$ without returning to state $i$, ${p_{N\downarrow 0} =\prod_{m=1}^N \frac{\alpha^{(\varepsilon)}_{m,m-1}}{\alpha_{m,m-1}^{(\varepsilon)}+\alpha_{m,m+1}^{(\varepsilon)}}}$, so we only need to bound the probability $\p(i\to N/i)$.

            To move from $i$ to $N$ the process has to first transition downward from $i$ to $i-1$, which occurs with probability at least $1-b$ by our uniform stability conditions. From there we can think of the process as akin to a gambler's ruin which either reaches state $i$ or state $N$ first.

            Consider a biased gambler's ruin on $\{N, N+1, \dots, i\}$ with ``success" probability $b$ starting in state $i-1$. The probability the gambler runs out of money (i.e., reaches state $N$ before state $i$) is $\frac{1-\Le(\frac{b}{1-b}\Ri)^1}{1-\Le(\frac{b}{1-b}\Ri)^{i-N}}$. We can uniformly upper bound this probability for any $i$ by $1-\frac{b}{1-b}>0$.

            Because the probability of transitioning upward for any state $i\geq N$ is bounded above by $b$, a coupling argument shows that the probability of the process reaching state $i$ before state $N$ from state $i-1$ has the same upper bound of $1-\frac{b}{1-b}$.

            Letting \[\underline{p}:=\Le(1-b\Ri)\Le(1-\frac{b}{1-b}\Ri)\prod_{m=1}^N \Le(\frac{\alpha_{m,m-1}}{\alpha_{m,m-1}+\alpha_{m,m+1}}\Ri)\] we have for any state $i$, ${p_i \geq \underline{p}>0}$
            and therefore 
            \[r_{i,i} \leq T_{\text{max}} \frac{1-\underline{p}}{\underline{p}}\]
            where $T_{\text{max}}$ is the maximum expected time in state for any state as defined in \cref{def:rtbdd}.
            \item Case 2: $i>j$.

            Next we consider the case when $i>j$. Because the phase process is Markovian, if the process enters state $j$ after starting in state $i$ then the expected time in state $j$ before reaching state $0$ is the same as if the process had started in state $j$. With probability $1$, the phase process started in state $i$ eventually reaches state $0$ so, as a birth-death process, it must pass through state $j$ at least once on the way to $0$. Therefore,
            \[r_{i,j}=r_{j,j}\leq T_{max} \frac{1-\underline{p}}{\underline{p}}.\]

            \item Case 3: $i<j$.

            Lastly we consider the case when $i<j$. In this case, the phase process might reach $0$ without ever going through state $j$. If we can bound the probability that the process started in $i$ visits $j$ before $0$ then we can bound $r_{i,j}$.

            Consider the probability in two cases: from state $i$ to state $j$ the process either transitions through state $N$ or not:
            \[\p(i\to j/0) = \p(i\to j/0,N) + \p(i\to N\to j/0)\]

            We can analyze these two possibilities separately. The probability the process moves from from $i$ to $j$ without reaching state $N$ is bounded by $(\frac{b}{1-b})^{j-i}$ by coupling with a biased gambler's ruin.

            Next consider the probability of touching state $N$ but not state $0$ on the way to state $j$ from state $i$. We can immediately upper bound this by the probability of simply moving from state $N$ to state $j$ from state $N$.
            \[\p(i\to N\to j/0) \leq \p(N\to j/0)\]

            We can also notice that the probability that the system reaches $j$ before $0$ after starting in $N$ with some number of returns is the same as the probability that, given this is the last time the system is in $N$, the system reaches state $j$ before $0$. Thus
            \[\p(N\to j/0) = \frac{\p(N\to j/0,N)}{\p(N\to j/0, N) + \p(N\to 0/j,N)} \leq \frac{\p(N\to j/0,N)}{\p(N\to 0/j,N)}\]

            We thus have a bound in terms of two easier to analyze quantities: $\p(N\to j/0,N)$ and $\p(N\to 0/j,N)$. We have an existing lower bound on $\p(N\to 0/j,N)$ in $\underline{p} = \prod_{i=1}^{N} \frac{\alpha_{i, i-1}}{\alpha_{i,i-1}+\alpha_{i, i+1}}$. We can also use our gambler's ruin technique again to upper bound $\p(N\to j/0,N)$. Once the process moves from state $N$ to $N+1$ (with probability no more than $b$), we can couple the process with a gambler's ruin on $\{N, N+1, N+2, \dots, j\}$. The probability the gambler's ruin reaches state $j$ can be upper bounded by $(\frac{b}{1-b})^{j-N-1}$.

            Taken together we have a bound on the probability of reaching state $j$,
            \begin{align*}
                \p(i\to j/0) &\leq (\frac{b}{1-b})^{j-i} + \frac{1}{\underline{p}} (\frac{b}{1-b})^{j-N-1},
            \end{align*}
            and can find our final bound
            \begin{align*}
                r_{i,j}^{(\varepsilon)} &\leq T_{max} \frac{1-\underline{p}}{\underline{p}}\Le(\Le(\frac{b}{1-b}\Ri)^{j-i} + \frac{1}{\underline{p}} \Le(\frac{b}{1-b}\Ri)^{j-N-1}\Ri)\\
                &= T_{max} \frac{1-\underline{p}}{\underline{p}}\Le(1 + \frac{1}{\underline{p}} \Le(\frac{b}{1-b}\Ri)^{i-N-1}\Ri) \Le(\frac{b}{1-b}\Ri)^{j-i}\\
                &\leq T_{max} \frac{1-\underline{p}}{\underline{p}}\Le(1 + \frac{1}{\underline{p}} \Le(\frac{b}{1-b}\Ri)^{-N-1}\Ri) \Le(\frac{b}{1-b}\Ri)^{j-i}
            \end{align*}

        \end{enumerate}

        Taking these three cases together we have our desired bound that for any $i,j$,
        \begin{align*}
            r_{i,j}^{(\varepsilon)} \leq C_2 (\frac{p}{1-p})^{(j-i)^+},
        \end{align*}
        where $C_2 :=T_{max} \frac{1-\underline{p}}{\underline{p}}\Le(1 + \frac{1}{\underline{p}} \Le(\frac{b}{1-b}\Ri)^{-N-1}\Ri) \geq T_{max} \Le(\frac{b}{1-b}\Ri)$.
        
        \item $\Le|[\boldsymbol{R}^{(\varepsilon)}e^T]_i\Ri| \leq C_3 i$.
        \label{ReTbd}
        \begin{align*}
            \Le|[\boldsymbol{R}^{(\varepsilon)}e^T]_i\Ri| &= \Le|\sum_{j=0}^\infty r_{i,j}^{\Le(\varepsilon\Ri)}\Ri|
            \leq \Le|\sum_{j=0}^{i-1} r_{i,j}^{\Le(\varepsilon\Ri)}\Ri| + \Le|\sum_{j=i}^\infty r_{i,j}^{\Le(\varepsilon\Ri)}\Ri|
        \end{align*}
            Using property \ref{eq:birth-bd-2} we can bound $r_{i,j}^{(\varepsilon)}$ to find that
        \begin{align*}
            \Le|[\boldsymbol{R}^{(\varepsilon)}e^T]_i\Ri| &\leq C_2 i + C_2 \sum_{j=i}^\infty a ^{j-i}
            = C_2i + C_2 \frac{1}{1-a}
            \leq C_3 i,
        \end{align*}
        where $C_3:=C_2\Le(1 + \frac{1}{1-a}\Ri)$.

    \end{enumerate}
    We thus have our desired bounds on $\pi_n^{(\varepsilon)}$, $r_{i,j}^{(\varepsilon)}$, and $\Le|[\boldsymbol{R}^{(\varepsilon)}e^T]_i\Ri|$.
    \end{proof}

\subsection{Birth-Death and Light-Tailed Modulation Processes}
\label{BDandLTMods}
With the bounds in \cref{FundBDBds} proven, we are ready to prove in \cref{DerBDBds} that the uniform stability and rate bounded conditions imply the light-tailed phase process conditions when we restrict our attention to birth-death families meeting our heavy traffic conditions.
\begin{lemma}
    \label{DerBDBds}
    A birth-death heavy-traffic family which is uniformly stable and rate bounded satisfies the light-tailed phase process conditions.
\end{lemma}
We provide a full proof in Appendix~\ref{DerDBBds_proof}. We first prove extensions of the inequalities in \cref{FundBDBds} to the quantities $\bpi^{(\varepsilon)} \Le(\mu^{(\varepsilon)} \boldsymbol{I}-z^\varepsilon \boldsymbol{F}^{(\varepsilon)}\Ri)$ and $\boldsymbol{R}^{(\varepsilon)} \Le(\mu^{(\varepsilon)} \boldsymbol{I}-z^\varepsilon \boldsymbol{F}^{(\varepsilon)}\Ri)$. In particular we show that just like we have the bounds
\begin{align*}
            \pi^{\Le(\varepsilon\Ri)}_n &\leq C_1a^n\\
        r^{\Le(\varepsilon\Ri)}_{i,j} &\leq C_2 \tilde{b}^{\Le(j-i\Ri)^+}\\
        \Le|[\boldsymbol{R}^{(\varepsilon)}e^T]_i\Ri| &\leq C_3 i
\end{align*}
for appropriate constants $C_1, C_2, C_3$ defined in \cref{FundBDBds}, we will also have bounds
\begin{align*}
    \Le|[\pi^{\Le(\varepsilon\Ri)}\Le(\mu^{\Le(\varepsilon\Ri)}\boldsymbol{I}-z^\varepsilon \boldsymbol{F}^{(\varepsilon)}\Ri)]_j\Ri| &\leq \tilde{C}_1a^j\\
    \Le|[\boldsymbol{R}^{(\varepsilon)}\Le(\mu^{\Le(\varepsilon\Ri)}\boldsymbol{I}-z^\varepsilon \boldsymbol{F}^{(\varepsilon)}\Ri)]_{i,j}\Ri| &\leq \tilde{C}_2\tilde{b}^{\Le(j-i\Ri)^+}\\
    \Le|[\boldsymbol{R}^{(\varepsilon)}\Le(\mu^{\Le(\varepsilon\Ri)}\boldsymbol{I}-z^\varepsilon \boldsymbol{F}^{(\varepsilon)}\Ri)e^T]_{i}\Ri| &\leq \tilde{C}_3i
\end{align*}
for appropriate $\tilde{C}_1, \tilde{C}_2, \tilde{C}_3$ defined in Appendix~\ref{DerDBBds_proof}.

We then use the fundamental bounds and our new derived bounds to show that each of the necessary quantities in the light-tailed phase process definition are asymptotically negligible.

\subsection{Asymptotic Independence and Exponentiality for Birth-death Modulation Processes}
With \cref{DerBDBds} proven, we can directly apply \cref{Thm1} to our birth-death heavy-traffic families.

\begin{reptheorem}{Thm2}
\label{Full_Thm2}
    In a family of IMAP queues with birth-death phase processes that meets our heavy traffic conditions, if the family is uniformly stable and rate bounded then it is also light-tailed. Consequently $\varepsilon Q^{(\varepsilon)}$ is asymptotically exponential with mean $M^{(0)}$ (where $M^{(0)}$ is defined in  \cref{Rdefn}) and asymptotically independent from $Y^{(\varepsilon)}$ in the $\varepsilon\to 0$ limit.

    Equivalently,
    \begin{align}
        \lim_{\varepsilon\to 0} \overline{q}(z;\varepsilon) &= \frac{1}{1+M^{(0)}\ln z}\tag{\ref{eq:Thm1-res-a}}\\
        \lim_{\varepsilon\to 0} \hat{q}(z;i,\varepsilon) &= \frac{\pi^{(0)}_i}{1+M^{(0)}\ln z}\tag{\ref{eq:Thm1-res-b}}
    \end{align}
\end{reptheorem}
\begin{proof}
    \cref{DerBDBds} proves that any birth-death convergent heavy-traffic family that is uniformly stable and rate bounded satisfies the light-tailed phase process conditions. \cref{Thm1} provides any heavy-traffic family which satisfies the light-tailed phase process conditions has the desired asymptotic independence and asymptotic exponential distribution. Thus, the uniform stability and rate bounded conditions provide our desired heavy-traffic results.
\end{proof}


\subsection{Hybrid-Abandonment Systems}
Finally, we are able to return to our original motivating system, a tandem queue with hybrid-abandonment. Consider a sequence of hybrid-abandonment systems with fixed service rates $\mu_1, \mu_2$ and abandonment rate $\theta$, but changing arrival rate $\lambda$. In particular, we consider the sequence of systems as $\lambda$ increases toward some $\lambda^{(0)}$ such that as $\lambda\to \lambda^{(0)}$, $\varepsilon\to 0$.
\begin{lemma}
\label{Thm3Sup}
    A family of hybrid-abandonment queues with changing $\lambda\in (\lambda_{\text{min}}, \lambda^{(0)})$, where $\varepsilon\to 0$ as $\lambda\to \lambda^{(0)}$, is a family of IMAP queues that meets our heavy-traffic condition and is birth-death, rate bounded, and uniformly stable.
\end{lemma}
\begin{proof}
    We prove this lemma in two parts. In \cref{sec:571}, we show that the family of hybrid-abandonment queues meets our heavy-traffic conditions. In \cref{sec:572}, we show that it is also birth-death, rate bounded, and uniformly stable.
    \subsubsection{Part 1: Heavy-Traffic Conditions}
    \label{sec:571}
    We first prove that the sequence of hybrid-abandonment queues is a family of IMAP queues that meets our heavy traffic conditions.

    Because $\mu_1$, $\mu_2$, and $\theta$ are fixed, $\boldsymbol{F}^{(\varepsilon)}$ and $\mu^{(\varepsilon)}$ are fixed. Thus checking our heavy traffic conditions simplifies to verifying that $\boldsymbol{R}^{(\varepsilon)}$ and $\bpi^{(\varepsilon)}$ converge and that $M(z;\varepsilon)\to M^{(0)}$ for all $z\in(0,1)$.
    
    We begin by establishing bounds on $\pi_n^{(0)}-\pi^{(\varepsilon)}_n$ and $r^{(0)}_{i,j} - r^{(\varepsilon)}_{i,j}$ necessary to prove the heavy-traffic conditions.
    \begin{enumerate}
        \item Bounding $|\pi_n^{(0)}-\pi^{(\varepsilon)}_n|$.
        
        We first show that
        \begin{equation}
            |\pi_n^{(0)}-\pi^{(\varepsilon)}_n| \leq \hat{C}_1 \varepsilon e^{-n}. \label{eq:pinmpi0bd}
        \end{equation}

        We begin by establishing a relationship between $\varepsilon$ and $\pi_0^{(\varepsilon)}$, the long-run probability the second queue and first queue are empty, respectively. Notice that in any stable system the arrival rate at which jobs arrive to and leave the second queue are equal, i.e.,
        \begin{align*}
            \mu_1(1-\pi_0^{(\varepsilon)}) =\mu_2 (1-\varepsilon)
            \implies
            1-\varepsilon = \Le(1-\pi^{(\varepsilon)}_0\Ri)\mu_1/\mu_2
            \implies
            \pi_0^{(\varepsilon)} = 1- \frac{\mu_2}{\mu_1}(1-\varepsilon).
        \end{align*}

        Next, because the phase process is a birth-death process, for any $n,\varepsilon$
        \[\pi^{(\varepsilon)}_n = \pi_{0}^{(\varepsilon)} \prod_{i=1}^n \frac{(\lambda^{(\varepsilon)})}{(\mu_1 + (i-1)\theta)}.\]

        Thus 
        \begin{align}\nonumber
            |\pi^{(0)}_n - \pi^{(\varepsilon)}_n | &= \Le|\pi^{(0)}_0 \frac{(\lambda^{(0)})^n}{\prod_{i=1}^n \Le(\mu_1+\Le(i-1\Ri)\theta\Ri)} -\pi^{(\varepsilon)}_0 \frac{(\lambda^{(\varepsilon)})^n}{\prod_{i=1}^n \Le(\mu_1+\Le(i-1\Ri)\theta\Ri)}\Ri| \\ \nonumber
            &= \Le|\pi^{(0)}_0 \frac{(\lambda^{(0)})^n}{\prod_{i=1}^n \Le(\mu_1+\Le(i-1\Ri)\theta\Ri)} -\pi^{(\varepsilon)}_0 \frac{(\lambda^{(\varepsilon)})^n + (\lambda^{(0)})^n - (\lambda^{(0)})^n}{\prod_{i=1}^n \Le(\mu_1+\Le(i-1\Ri)\theta\Ri)}\Ri| \\ \nonumber
            &= \Le|(\pi^{(0)}_0-\pi_0^{(\varepsilon)}) \frac{(\lambda^{(0)})^n}{\prod_{i=1}^n \Le(\mu_1+\Le(i-1\Ri)\theta\Ri)} -\pi^{(\varepsilon)}_0 \frac{(\lambda^{(\varepsilon)})^n - (\lambda^{(0)})^n}{\prod_{i=1}^n \Le(\mu_1+\Le(i-1\Ri)\theta\Ri)}\Ri| \\
            &\leq \Le|(\pi^{(0)}_0-\pi_0^{(\varepsilon)}) \frac{(\lambda^{(0)})^n}{\prod_{i=1}^n \Le(\mu_1+\Le(i-1\Ri)\theta\Ri)}\Ri| \label{eq:pi0mpieps-1}\\
            &+ \Le|\pi^{(\varepsilon)}_0 \frac{(\lambda^{(\varepsilon)})^n - (\lambda^{(0)})^n}{\prod_{i=1}^n \Le(\mu_1+\Le(i-1\Ri)\theta\Ri)}\Ri|\label{eq:pi0mpieps-2}
        \end{align}

        We can consider \cref{eq:pi0mpieps-1} and \cref{eq:pi0mpieps-2} separately. Let's begin with \cref{eq:pi0mpieps-1}.
        \begin{align*}
            \Le|(\pi^{(0)}_0-\pi_0^{(\varepsilon)}) \frac{(\lambda^{(0)})^n}{\prod_{i=1}^n \Le(\mu_1+\Le(i-1\Ri)\theta\Ri)}\Ri| &= \Le|-\varepsilon \frac{\mu_2}{\mu_1} \frac{(\lambda^{(0)})^n}{\prod_{i=1}^n \Le(\mu_1+\Le(i-1\Ri)\theta\Ri)}\Ri|\\
            &=\varepsilon \frac{\mu_2}{\mu_1} \frac{(\lambda^{(0)})^n}{\prod_{i=1}^n \Le(\mu_1+\Le(i-1\Ri)\theta\Ri)}
        \end{align*}
        Let $\hat{n} = 1+\frac{e\lambda^{(0)}}{\theta}$. For all $n\geq \hat{n}$, $\frac{\lambda^{(0)}}{\Le(\mu_1+\Le(n-1\Ri)\theta\Ri)} \leq e^{-1}$. Letting $A:=\frac{\mu_2}{\mu_1} \sup_{n\leq \hat{n}} e^n\frac{(\lambda^{(0)})^n}{\prod_{i=1}^n \Le(\mu_1+\Le(i-1\Ri)\theta\Ri)}$, we have that
        \begin{align*}
            \Le|(\pi^{(0)}_0-\pi_0^{(\varepsilon)}) \frac{(\lambda^{(0)})^n}{\prod_{i=1}^n \Le(\mu_1+\Le(i-1\Ri)\theta\Ri)}\Ri| &= \varepsilon \frac{\mu_2}{\mu_1} \frac{(\lambda^{(0)})^n}{\prod_{i=1}^n \Le(\mu_1+\Le(i-1\Ri)\theta\Ri)}
            \leq A\varepsilon e^{-n}
        \end{align*}

        Next we turn to \cref{eq:pi0mpieps-2}. By factoring $(\lambda^{(0)}-\lambda^{(\varepsilon)})$ from $((\lambda^{(0)})^n-(\lambda^{(\varepsilon)})^n)$ we can perform a similar analysis,
        \begin{align*}
            \Le|\pi^{(\varepsilon)}_0 \frac{(\lambda^{(\varepsilon)})^n - (\lambda^{(0)})^n}{\prod_{i=1}^n \Le(\mu_1+\Le(i-1\Ri)\theta\Ri)}\Ri| &= \pi^{(\varepsilon)}_0 \frac{(\lambda^{(0)})^n-(\lambda^{(\varepsilon)})^n}{\prod_{i=1}^n \Le(\mu_1+\Le(i-1\Ri)\theta\Ri)}\\
            &= \pi^{(\varepsilon)}_0 \frac{(\lambda^{(0)}-\lambda^{(\varepsilon)})\sum_{k=0}^{n-1}(\lambda^{(0)})^k(\lambda^{(\varepsilon)})^{n-1-k}}{\prod_{i=1}^n \Le(\mu_1+\Le(i-1\Ri)\theta\Ri)}\\
            &\leq \frac{(\lambda^{(0)}-\lambda^{(\varepsilon)})n(\lambda^{(0)})^{n-1}}{\prod_{i=1}^n \Le(\mu_1+\Le(i-1\Ri)\theta\Ri)}
        \end{align*}

        Because $\frac{(n+1)(\lambda^{(0)})^{n}}{\prod_{i=1}^{n+1} \Le(\mu_1+\Le(i-1\Ri)\theta\Ri)}/\frac{n(\lambda^{(0)})^{n-1}}{\prod_{i=1}^n \Le(\mu_1+\Le(i-1\Ri)\theta\Ri)} = \Le(\frac{n+1}{n}\Ri)\frac{\lambda^{(0)}}{\mu_1 + n\theta}\to 0$, there exists a $\hat{m}$ such that for all $n\geq \hat{m}$, \[\frac{(n+1)(\lambda^{(0)})^{n}}{\prod_{i=1}^{n+1} \Le(\mu_1+\Le(i-1\Ri)\theta\Ri)}/\frac{n(\lambda^{(0)})^{n-1}}{\prod_{i=1}^n \Le(\mu_1+\Le(i-1\Ri)\theta\Ri)} \leq e^{-1}.\] Letting $B:= \frac{\mu_2}{\mu_1} \sup_{n\leq \hat{m}} e^n\frac{n(\lambda^{(0)})^{n-1}}{\prod_{i=1}^n \Le(\mu_1+\Le(i-1\Ri)\theta\Ri)}$ we have
        \begin{align*}
            \Le|\pi^{(\varepsilon)}_0 \frac{(\lambda^{(\varepsilon)})^n - (\lambda^{(0)})^n}{\prod_{i=1}^n \Le(\mu_1+\Le(i-1\Ri)\theta\Ri)}\Ri| 
            &\leq \frac{(\lambda^{(0)}-\lambda^{(\varepsilon)})n(\lambda^{(0)})^{n-1}}{\prod_{i=1}^n \Le(\mu_1+\Le(i-1\Ri)\theta\Ri)}\leq B\varepsilon e^{-n}
        \end{align*}

        Using our bounds on \cref{eq:pi0mpieps-1} and \cref{eq:pi0mpieps-2}, we now have our desired bound,
                \begin{align*}
            \Le|\pi^{(\varepsilon)}_n - \pi^{(0)}_n\Ri| &\leq \Le|(\pi^{(0)}_0-\pi_0^{(\varepsilon)}) \frac{(\lambda^{(0)})^n}{\prod_{i=1}^n \Le(\mu_1+\Le(i-1\Ri)\theta\Ri)}\Ri| + \Le|\pi^{(\varepsilon)}_0 \frac{(\lambda^{(\varepsilon)})^n - (\lambda^{(0)})^n}{\prod_{i=1}^n \Le(\mu_1+\Le(i-1\Ri)\theta\Ri)}\Ri|\\
            &\leq A\varepsilon e^{-n} + B\varepsilon e^{-n}= \hat{C}_1 \varepsilon e^{-n}
        \end{align*}
        where $\hat{C}_1 := A+B$.

        \item Bounds on $|r_{i,j}^{(0)}-r^{(\varepsilon)}_{i,j}|$.
        
        With our bounds on $|\pi_n^{(\varepsilon)}-\pi_n^{(0)}|$ proven, we next show that
        \begin{equation}
            |r_{i,j}^{(0)}-r^{(\varepsilon)}_{i,j}| \leq \hat{C}_2(\lambda^{(0)} - \lambda^{(\varepsilon)}). \label{eq:rij0-rijeps-bd}
        \end{equation}

        Let $T_{i,j}$ be the random amount of time that the phase process spends in state $j$ between starting in state $i$ and reaching state $0$. Note that $r_{i,j} = \E[T_{i,j}]$.
        
        Consider a coupled pair of systems: one, $Y^{(0)}$, with arrival rate $\lambda^{(0)}$ and one, $Y^{(\varepsilon)}$, with rate $\lambda^{(\varepsilon)}$ for some $\varepsilon>0$. 

        When $(Y^{(0)}, Y^{(\varepsilon)}) = (n,n)$, the pair jointly transition upward at rate $\lambda^{(\varepsilon)}$. The pair jointly transition downward at rate $\mu_1 + (n-1)\theta$. With rate $\lambda^{(0)} - \lambda^{(\varepsilon)}$, $Y^{(0)}$ transitions upward and $Y^{(\varepsilon)}$ doesn't transition. When $Y^{(0)} \ne Y^{(\varepsilon)}$, the two evolve independently. This coupling ensures that $Y^{(0)} \geq Y^{(\varepsilon)}$ almost surely.

        Consider two possibilities: if the systems do not diverge before $Y^{(0)}$ reaches state $0$ then the two have moved in lockstep the entire duration and have spent exactly the same time in state $j$. If the two processes do diverge before this point then the two processes may spend different times in state $j$, thus
        \begin{align*}
            r_{i,j}^{(0)} - r^{(\varepsilon)}_{i,j} = \E[T^{(0)}_{i,j} -T^{(\varepsilon)}_{i,j}]
            &= \E[T^{(0)}_{i,j} -T^{(\varepsilon)}_{i,j}\mid Y^{(0)}, Y^{(\varepsilon)}\text{ diverge before } Y^{(0)}=0] \\&\qquad\cdot \p(Y^{(0)}, Y^{(\varepsilon)}\text{ diverge before } Y^{(0)}=0)
        \end{align*}

        We bound both of these terms separately. First we consider the conditional mean. Let $I$ be the (random) index where the first divergence occurs. At this divergence, the two processes move from $(I,I)$ to $(I,I+1)$. Thus
        \begin{align*}
            |\E[T^{(0)}_{i,j} -T^{(\varepsilon)}_{i,j}\mid Y^{(0)}, Y^{(\varepsilon)}\text{ diverge before } Y^{(0)}=0]| &=| \E[r^{(0)}_{I+1,j}-r^{(\varepsilon)}_{I,j}]|\\
            &\leq |r^{(0)}_{j,j}| + |r^{(\varepsilon)}_{j,j}|
            \leq 2C_2,
        \end{align*}
        where $C_2$ is defined in \cref{FundBDBds}.

        Until a divergence occurs, the divergences occur at rate $\lambda^{(0)}-\lambda^{(\varepsilon)}$. The probability that a divergence occurs before time $t$ is $\p(\text{Pois}((\lambda^{(0)}-\lambda^{(\varepsilon)})t)>0) = 1-\exp(-(\lambda^{(0)}-\lambda^{(\varepsilon)})t)$. Letting $\tau_i$ be the time when $Y^{(0)}$ first reaches $0$,
        \begin{align*}
            \mathbb{P}(Y^{(0)}, Y^{(\varepsilon)}\text{ diverge before } Y^{(0)}=0) & = \E[1-\exp(-(\lambda^{(0)}-\lambda^{(\varepsilon)})\tau_i)]\\
            &\leq (\lambda^{(0)}-\lambda^{(\varepsilon)})\E[\tau_i].
        \end{align*}

        Noting that $\E[\tau_i]$ is the same as $\boldsymbol{R}^{(\varepsilon)}e^T$ in \cref{FundBDBds}, we have that $\E[\tau_i]\leq C_3 i$ where $C_3$ is defined in \cref{FundBDBds}. Thus
        \begin{align*}
            &\mathbb{P}(Y^{(0)}, Y^{(\varepsilon)}\text{ diverge before } Y^{(0)}=0) \leq C_3(\lambda^{(0)}-\lambda^{(\varepsilon)})i \\
            \implies
            &|r_{i,j}^{(0)} - r^{(\varepsilon)}_{i,j}| \leq (2C_2)\cdot(C_3(\lambda^{(0)}-\lambda^{(\varepsilon)})i)=\hat{C}_3(\lambda^{(0)}-\lambda^{(\varepsilon)})i
        \end{align*}
        for $\hat{C}_3 = 2C_2C_3$.

    \end{enumerate}

    With these bounds, we are now able to show that the hybrid-abandonment systems meet our heavy-traffic condition (\cref{def:HT_cond}). We verify the conditions here. We begin by showing that each of $\bpi$, $\mathbf{R}$, $\mathbf{F}$, and $\mu_2$ have the desired convergence properties and then show that $M(z^\varepsilon;\varepsilon)\to M^{(0)}$
    \begin{enumerate}
        \item As $\varepsilon\to 0$ or, equivalently, $\lambda^{(\varepsilon)} \to \lambda^{(0)}$,
        \begin{enumerate}
            \item For any $n$, $|\pi_n^{(\varepsilon)} - \pi_n^{(0)}|\leq \hat{C}_1\varepsilon e^{-n}$ and thus $\pi^{(\varepsilon)}_n\to \pi_{n}^{(0)}$ where $\bpi^{(0)}$ is a valid distribution.

            \item For any $i,j$, $|r_{i,j}^{(0)}-r_{i,j}^{(\varepsilon)}| \leq \hat{C}_2 (\lambda^{(0)}-\lambda^{(\varepsilon)})$ so $r_{i,j}^{(0)}\to r_{i,j}^{(\varepsilon)}$.

            \item $\boldsymbol{F}^{(\varepsilon)}$ is constant across all $\varepsilon$ so trivially $\boldsymbol{F}^{(\varepsilon)}$ is convergent as $\varepsilon\to 0$.

            \item $\mu_2^{(\varepsilon)}$ is constant across all $\varepsilon$ so trivially $\mu_2^{(\varepsilon)}$ is convergent as $\varepsilon\to 0$.
        \end{enumerate}

        \item As $\varepsilon\to 0$,
        \begin{align*}
        &1-(\mu^{(\varepsilon)}_2)^{-1} \bpi^{(\varepsilon)}(\mu^{(\varepsilon)}\boldsymbol{I}-z^\varepsilon \boldsymbol{F}^{(\varepsilon)})\boldsymbol{R}^{(\varepsilon)}(\mu^{(\varepsilon)}\boldsymbol{I}-z^\varepsilon \boldsymbol{F}^{(\varepsilon)})e^T\\ &\qquad\to 1-(\mu_2^{(0)})^{-1} \bpi^{(0)}(\mu^{(0)}\boldsymbol{I}-\boldsymbol{F}^{(0)})\boldsymbol{R}^{(0)}(\mu^{(0)}\boldsymbol{I}-\boldsymbol{F}^{(0)})e^T
        \end{align*}

            We can decompose $\bpi^{(\varepsilon)}(\boldsymbol{I}-z^\varepsilon \boldsymbol{F}^{(\varepsilon)})\boldsymbol{R}^{(\varepsilon)}(\boldsymbol{I}-z^\varepsilon \boldsymbol{F}^{(\varepsilon)})e^T$ to isolate the effect of $z^\varepsilon$, $\boldsymbol{R}^{(\varepsilon)}$, and $\bpi^{(\varepsilon)}$. Because $\boldsymbol{F}^{(\varepsilon)}$ does not depend on $\varepsilon$ in this case, we omit the dependence on $\varepsilon$. Similarly we replace $\mu^{(\varepsilon)}$ by $\mu_2$ for clarity because it does not depend on $\varepsilon$.
            \begin{align}
                &\bpi^{(\varepsilon)}(\mu_2\boldsymbol{I}-z^\varepsilon \boldsymbol{F})\boldsymbol{R}^{(\varepsilon)}(\mu_2\boldsymbol{I}-z^\varepsilon \boldsymbol{F})e^T \\&= \bpi^{(0)}(\mu_2\boldsymbol{I}-\boldsymbol{F})\boldsymbol{R}^{(0)}(\mu_2\boldsymbol{I}-\boldsymbol{F})e^T\nonumber\\
                &+(\bpi^{(\varepsilon)}-\bpi^{(0)})(\mu_2\boldsymbol{I}-\boldsymbol{F})\boldsymbol{R}^{(0)}(\mu_2\boldsymbol{I}-\boldsymbol{F})e^T\label{eq:M0_decomp_2}\\
                &+\bpi^{(\varepsilon)}(1-z^\varepsilon)\boldsymbol{F}\boldsymbol{R}^{(0)}(\mu_2\boldsymbol{I}-\boldsymbol{F})e^T\label{eq:M0_decomp_3}\\
                &+\bpi^{(\varepsilon)}(\mu_2\boldsymbol{I}-z^{\varepsilon}\boldsymbol{F})(\boldsymbol{R}^{(\varepsilon)}-\boldsymbol{R}^{(0)})(\boldsymbol{I}-\boldsymbol{F})e^T\label{eq:M0_decomp_4}\\
                &+\bpi^{(\varepsilon)}(\mu_2\boldsymbol{I}-z^\varepsilon \boldsymbol{F})\boldsymbol{R}^{(\varepsilon)}(1-z^\varepsilon) \boldsymbol{F}e^T \label{eq:M0_decomp_5}
            \end{align}

            We show \cref{eq:M0_decomp_2,eq:M0_decomp_3,eq:M0_decomp_4,eq:M0_decomp_5} are each negligible in this limit. We extensively use the bounds derived in \cref{FundBDBds} and \cref{DerBDBds}.
            \begin{enumerate}
                \item \cref{eq:M0_decomp_2} is asymptotically negligible.

                First let us bound $(\bpi^{(\varepsilon)}-\bpi^{(0)})(\mu_2\boldsymbol{I}-\boldsymbol{F})$. In particular, this has an analogous bound to its underlying quantities, $(\bpi^{(\varepsilon)}-\bpi^{(0)})$ (\cref{eq:pinmpi0bd}).
                \begin{align*}
                    \Le|[(\bpi^{(\varepsilon)}-\bpi^{(0)})(\mu^{(0)}\boldsymbol{I}-\boldsymbol{F})]_{i}\Ri| &= \Le|\sum_{k} ((\pi^{(\varepsilon)}_i-\pi^{(0)}_i)[\mu_2\boldsymbol{I}-\boldsymbol{F}]_{i,k})\Ri|\\
                    &=\Le|(\pi^{(\varepsilon)}_i -\pi^{(0)}_i)\mu_2 - (\pi^{(\varepsilon)}_{i-1}-\pi^{(0)}_{i-1})\beta_{i,i-1}\indic\{i>0\}\Ri|\\
                    &\leq \mu_2 \hat{C}_1 \varepsilon e^{-i} + \mu_1 \hat{C}_1\varepsilon e^{-(i-1)}\\
                    &=\hat{C}_1(\mu_2+e\mu_1)\varepsilon e^{-i}
                \end{align*}

                Thus using the bound we just derived and the bound on $\boldsymbol{R}(\mu^{(\varepsilon)}\boldsymbol{I}-\boldsymbol{F})e^T$ introduced in \cref{BDandLTMods},
                \begin{align*}
                    &\Le|(\bpi^{(\varepsilon)}-\bpi^{(0)})(\mu^{(0)}\boldsymbol{I}-\boldsymbol{F})\boldsymbol{R}^{(0)}(\mu_2\boldsymbol{I}-\boldsymbol{F})e^T\Ri|\\ &= \Le|\sum_{i}[(\bpi^{(\varepsilon)}-\bpi^{(0)})(\mu^{(0)}\boldsymbol{I}-\boldsymbol{F})]_{i}\ [\boldsymbol{R}(\mu_2\boldsymbol{I}-\boldsymbol{F})e^T]_i\Ri|\\
                    &\leq \sum_{i} ((\hat{C}_1(\mu_2+e\mu_1)\varepsilon e^{-i})\ (\tilde{C}_3i))=\varepsilon \hat{C}._1\tilde{C}_3(\mu_2+e\mu_1) \sum_i ie^{-i}
                \end{align*}
                Thus $(\bpi^{(\varepsilon)}-\bpi^{(0)})(\mu^{(0)}\boldsymbol{I}-\boldsymbol{F})\boldsymbol{R}^{(0)}(\mu_2\boldsymbol{I}-\boldsymbol{F})e^T\to \varepsilon$ as $\varepsilon\to 0$.

                \item \cref{eq:M0_decomp_3} is asymptotically negligible.

                First, we note that
                \begin{align*}
                    [\bpi^{(\varepsilon)}\boldsymbol{F}]_j = \sum_{i}\pi^{(\varepsilon)}_i \boldsymbol{F}_{i,j}
                    = \pi^{(\varepsilon)}_{i+1} \beta_{i+1, i}
                    &\leq C_1 \beta_{i+1,i} e^{-(i+1)}
                    = C_1 e^{-1}\beta_{i+1,i} e^{-i}.
                \end{align*}
                Using our newly derived bound we can see that
                \begin{align*}
                    &|\bpi^{(\varepsilon)}(1-z^\varepsilon)\boldsymbol{F}\boldsymbol{R}^{(0)}(\boldsymbol{I}-\boldsymbol{F})e^T|  = (1-z^\varepsilon) \Le|\sum_{i}
                    \Le([\bpi^{(\varepsilon)}\boldsymbol{F}]_i\ [\boldsymbol{R}^{(0)}(\mu_2 \boldsymbol{I} -\boldsymbol{F})e^T]_i\Ri) \Ri|\\
                    &\leq (1-z^\varepsilon)\Le|\sum_{i} \Le((C_1 e^{-1}\beta_{i+1,i} e^{-i})\ (\tilde{C}_3 i)\Ri)\Ri|
                    =(1-z^\varepsilon)\tilde{C}_3C_1 e^{-1}\beta_{i+1,i}\sum_{i}  e^{-i} i
                \end{align*}
                Thus as $\varepsilon\to0$, $(1-z^{\varepsilon})\to 0$ and $\bpi^{(\varepsilon)}(1-z^\varepsilon)\boldsymbol{F}\boldsymbol{R}^{(0)}(\mu_2\boldsymbol{I}-\boldsymbol{F})e^T\to 0$

                \item \cref{eq:M0_decomp_4} is asymptotically negligible.

                \begin{align*}
                    &|\bpi^{(\varepsilon)}(\mu_2\boldsymbol{I}-z^{\varepsilon}\boldsymbol{F})(\boldsymbol{R}^{(\varepsilon)}-\boldsymbol{R}^{(0)})(\mu_2\boldsymbol{I}-\boldsymbol{F})e^T|\\&= |[\bpi^{(\varepsilon)}(\mu_2\boldsymbol{I}-z^{\varepsilon}\boldsymbol{F})](\boldsymbol{R}^{(\varepsilon)}-\boldsymbol{R}^{(0)})[(\mu_2\boldsymbol{I}-\boldsymbol{F})e^T]|\\
                    &\leq |[\bpi^{(\varepsilon)}(\mu_2\boldsymbol{I}-z^{\varepsilon}\boldsymbol{F})](\boldsymbol{R}^{(\varepsilon)}-\boldsymbol{R}^{(0)})[(\mu_2+\mu_1)e^T]|\\
                    &\leq (\mu_1+\mu_2)|[\bpi^{(\varepsilon)}(\mu_2\boldsymbol{I}-z^{\varepsilon}\boldsymbol{F})](\boldsymbol{R}^{(\varepsilon)}-\boldsymbol{R}^{(0)})e^T|\\
                    &=(\mu_1+\mu_2) \sum_{i,j}([\bpi^{(\varepsilon)}(\mu_2\boldsymbol{I}-z^{\varepsilon}\boldsymbol{F})]_i\ (r_{i,j}^{(\varepsilon)}-r_{i,j}^{(0)}))\\
                    &\leq (\mu_1+\mu_2) \sum_{i}(\tilde{C}_1 e^{-i} (\sum_j (r_{i,j}^{(\varepsilon)}-r_{i,j}^{(0)})))
                \end{align*}

                We use two different bounds for $r^{(\varepsilon)}_{i,j} - r^{(0)}_{i,j}$. In particular we use our coupling bound, \mbox{$|r^{(\varepsilon)}_{i,j} - r^{(0)}_{i,j}| \leq \hat{C}_2(\lambda^{(0)}-\lambda^{(\varepsilon)})i$}, \eqref{eq:rij0-rijeps-bd}, for ``small" $j$ and our bounds on $|r_{i,j}^{(\varepsilon}|$, $|r_{i,j}^{(\varepsilon)}| \leq C_2 e^{-(j-i)^+}$ (property \ref{eq:birth-bd-2}), for ``large" $j$. We change our bound at $j=i+\ln ((\lambda^{(0)}-\lambda^{(\varepsilon)})^{-1})$. Thus, 
                \begin{align*}
                    \left|\sum_{j}(r^{(\varepsilon)}_{i,j} - r^{(0)}_{i,j})\right| &= \left|\sum_{j=0}^{\lfloor i-\ln(\lambda^{(0)}-\lambda^{(\varepsilon)})\rfloor}(r^{(\varepsilon)}_{i,j} - r^{(0)}_{i,j}) + \sum_{j=1+\lfloor i-\ln(\lambda^{(0)}-\lambda^{(\varepsilon)})\rfloor}^{\infty}(r^{(\varepsilon)}_{i,j} - r^{(0)}_{i,j}) \right|\\
                    &\leq |\sum_{j=0}^{\lfloor i-\ln(\lambda^{(0)}-\lambda^{(\varepsilon)})\rfloor}(\hat{C}_2(\lambda^{(0)}-\lambda^{(\varepsilon)})i)| \\
                    &+ |\sum_{j=1+\lfloor i-\ln(\lambda^{(0)}-\lambda^{(\varepsilon)})\rfloor}^{\infty}(C_2e^{-(j-i)^+})|\\
                    &\leq (i-\ln(\lambda^{(0)}-\lambda^{(\varepsilon)}))(\hat{C}_2(\lambda^{(0)}-\lambda^{(\varepsilon)})i)\\ 
                    &+ C_2(\lambda^{(0)}-\lambda^{(\varepsilon)}) \sum_{n=0}^\infty e^{-n}\\
                    &=\hat{C}_2(\lambda^{(0)}-\lambda^{(\varepsilon)})i^2 - \hat{C}_2\ln(\lambda^{(0)}-\lambda^{(\varepsilon)})(\lambda^{(0)}-\lambda^{(\varepsilon)})i \\&+ \frac{C_2}{1-e^{-1}} (\lambda^{(0)}-\lambda^{(\varepsilon)})
                \end{align*}
                With this bound on  $|\sum_{j}(r^{(\varepsilon)}_{i,j} - r^{(0)}_{i,j})|$ we can analyze \cref{eq:M0_decomp_4} directly.
                \begin{align*}
                    &\Le|\bpi^{(\varepsilon)}(\mu_2\boldsymbol{I}-z^{\varepsilon}\boldsymbol{F})(\boldsymbol{R}^{(\varepsilon)}-\boldsymbol{R}^{(0)})(\mu_2\boldsymbol{I}-\boldsymbol{F})e^T\Ri| \\
                    &\leq (\mu_1+\mu_2) \sum_{i}( \tilde{C}_1e^{-i} \sum_{j}(r_{i,j}^{(0)}-r_{i,j}^{(\varepsilon)}))\\
                    &\leq  (\mu_1+\mu_2)\sum_{i} (\tilde{C}_1e^{-i}(\hat{C}_2(\lambda^{(0)}-\lambda^{(\varepsilon)})i^2 - \hat{C}_2\ln(\lambda^{(0)}-\lambda^{(\varepsilon)})(\lambda^{(0)}-\lambda^{(\varepsilon)})i ))\\
                    &+  (\mu_1+\mu_2)(\lambda^{(0)}-\lambda^{(\varepsilon)}) \sum_i(\frac{C_2}{1-e^{-1}}e^{-i})\\
                    &= (\mu_1+\mu_2)(\lambda^{(0)}-\lambda^{(\varepsilon)}) \sum_{i}(C_1\hat{C}_2 i^2e^{-i)} \\&-  (\mu_1+\mu_2)(\lambda^{(0)}-\lambda^{(\varepsilon)})\ln (\lambda^{(0)}-\lambda^{(\varepsilon)}) \sum_i(C_1\hat{C}_2 ie^{-i}) \\&+  (\mu_1+\mu_2)(\lambda^{(0)}-\lambda^{(\varepsilon)}) \sum_i(\frac{C_2}{1-e^{-1}}e^{-i})
                \end{align*}
                Thus because $(\lambda^{(0)}-\lambda^{(\varepsilon)})\to 0$ and $-(\lambda^{(0)}-\lambda^{(\varepsilon)})\ln((\lambda^{(0)}-\lambda^{(\varepsilon)}))\to 0$ as $\varepsilon\to 0$,\\ \mbox{
$\bpi^{(\varepsilon)}(\mu_2\boldsymbol{I}-z^{\varepsilon}\boldsymbol{F})(\boldsymbol{R}^{(\varepsilon)}-\boldsymbol{R}^{(0)})(\mu_2\boldsymbol{I}-\boldsymbol{F})e^T\to 0$}.                

                \item \cref{eq:M0_decomp_5} is asymptotically negligible.

                \begin{align*}
                    &\Le|\bpi^{(\varepsilon)}(\boldsymbol{I}-z^\varepsilon \boldsymbol{F})\boldsymbol{R}^{(\varepsilon)}(1-z^\varepsilon) \boldsymbol{F}e^T\Ri| = (1-z^\varepsilon)\Le|\bpi^{(\varepsilon)}(\boldsymbol{I}-z^\varepsilon \boldsymbol{F})\boldsymbol{R}^{(\varepsilon)} \boldsymbol{F}e^T\Ri| \\
                    &=(1-z^\varepsilon)\Le|[\bpi^{(\varepsilon)}(\boldsymbol{I}-z^\varepsilon \boldsymbol{F})]\ \boldsymbol{R}^{(\varepsilon)} \ [\boldsymbol{F}e^T]\Ri|
                    \leq(1-z^\varepsilon)\Le|[\bpi^{(\varepsilon)}(\boldsymbol{I}-z^\varepsilon \boldsymbol{F})]\ \boldsymbol{R}^{(\varepsilon)} \ [\mu_1e^T]\Ri|\\
                    &= \mu_1(1-z^\varepsilon)\Le|[\bpi^{(\varepsilon)}(\boldsymbol{I}-z^\varepsilon \boldsymbol{F})]\ \boldsymbol[{R}^{(\varepsilon)}e^T]\Ri|
                    \leq \mu_1(1-z^\varepsilon) \Le|\sum_i ((\tilde{C}_1e^{-i}) (C_3i))\Ri|\\
                    &=(1-z^\varepsilon) \mu_1 \tilde{C}_1C_3 \sum_i ie^{-i}
                \end{align*}

                Thus as $\varepsilon\to 0$, $(1-z^\varepsilon)\to 0$ and $\bpi^{(\varepsilon)}(\boldsymbol{I}-z^\varepsilon \boldsymbol{F})\boldsymbol{R}^{(\varepsilon)}(1-z^\varepsilon) \boldsymbol{F}e^T\to 0$.

            \end{enumerate}

        \item For all $\varepsilon$, the phase process is an $M/M/1+M$ queue with $\lambda>0$ and $\theta>0$. Thus the phase process is ergodic for all $\varepsilon$.
    \end{enumerate}

    The hybrid-abandonment system therefore meets our heavy traffic conditions (\cref{def:HT_cond}). Having show this we need only show that the system also meets the birth-death, rate bounded, and uniformly stable conditions.

\subsubsection{Part 2: Birth-Death, Rate Bounded, and Uniformly Stable}
\label{sec:572}
    
    The phase process, the first queue, is a birth-death process for all $\varepsilon$, so it immediately satisfies our birth-death condition (\cref{def:bd}). Thus, we need only show that the family is uniformly stable (\cref{def:unifstab}) and rate-bounded (\cref{def:rtbdd}).
    With $a=b=e^{-1}$ and $N=\lceil \frac{\theta e}{\lambda_{\text{max}}}\rceil +1$,
    \begin{enumerate}
            \item $\sup_{n\geq N,\varepsilon} \frac{\alpha_{n,n+1}}{\alpha_{n,n-1}} = \sup_{n\geq N,\varepsilon} \frac{\lambda}{\mu_11\{n>0\} + \Le(n-1\Ri)\theta1\{n>0\}} \leq e^{-1}$.

        \item $\sup_{n\geq N,\varepsilon} \frac{\alpha_{n,n+1}}{\alpha_{n+1, n}} = \sup_{n\geq N,\varepsilon} \frac{\lambda^{\Le(\varepsilon\Ri)}}{\mu_1+n\theta} \leq e^{-1}$.
        \end{enumerate}
        Thus, the family satisfies all the conditions of \cref{def:unifstab} and is uniformly stable.

        Finally, let us show that the family is rate-bounded, with $T_{\max} = (\lambda_{\text{min}})^{-1}$ and $\hat{\lambda} = \mu_1$.
    \begin{enumerate}
        \item $\sup_{n,\varepsilon}\Le\{ \Le(\sum_{m\ne n} \alpha_{n,m}\Ri)^{-1}\Ri\} = \sup_n \Le\{\Le(\lambda^{\Le(\varepsilon\Ri)}+\mu_11\{n>0\}+\Le(n-1\Ri)\theta1\{n>0\}\Ri)^{-1}\Ri\} = (\lambda_{\text{min}})^{-1} <~\infty$.

        \item $\sup_{n,\varepsilon}\Le(\lambda^{\Le(\varepsilon\Ri)}_n + \beta^{\Le(\varepsilon\Ri)}_{n,n+1}+\beta^{\Le(\varepsilon\Ri)}_{n,n-1}\Ri) \leq \mu_1 $.

        \item $\mu^{(\varepsilon)}=\mu_2$ for all $\varepsilon$.

    \end{enumerate}
    The family therefore meets all the requirements of \cref{def:rtbdd} and is rate-bounded.

    Thus, the sequence of hybrid-abandonment queues is a birth-death family of IMAP queue that meets our heavy traffic conditions and is uniformly stable and rate bounded.

\end{proof}

\subsection{Asymptotic Independence and Exponentiality for Hybrid Abandonment Systems}
Finally, we prove that the same asymptotic results hold for hybrid-abandonment systems as for the more general heavy-traffic families. This further narrows the variety of systems we consider but allows us to provide the most straightforward guarantees.
\begin{reptheorem}{Thm3}
\label{Full_Thm3}
     In a tandem queue with hybrid abandonment, as the arrival rate rises to a critical threshold and the second queue's idleness approaches $0$, $\varepsilon Q^{(\varepsilon)}$ is asymptotically exponential with mean $M^{(0)}$ (where $M^{(0)}$ is defined in  \cref{def:HT_cond}) and asymptotically independent from $Y^{(\varepsilon)}$.

    Equivalently,
    \begin{align}
        \lim_{\varepsilon\to 0} \overline{q}(z;\varepsilon) &= \frac{1}{1+M^{(0)}\ln z}\tag{\ref{eq:thm1-res-a}}\\
        \lim_{\varepsilon\to 0} \hat{q}(z;i,\varepsilon) &= \frac{\pi^{(0)}_i}{1+M^{(0)}\ln z}\tag{\ref{eq:Thm1-res-b}}
    \end{align}
\end{reptheorem}
\begin{proof}
    \cref{Thm3Sup} shows that the hybrid-abandonment systems are birth-death convergent heavy traffic families that are uniformly stable and rate bounded. Thus, by \cref{Thm2}, the hybrid-abandonment systems have asymptotically exponentially scaled second queue lengths and asymptotically independent first and second queue lengths. 
\end{proof}

\section{Numerical Study}
\label{Sec6}

\cref{Thm3} proves that in the load $\rho\to 1$ limit, the scaled second queue length $\varepsilon \E[{Q_2}]$ has mean $M^{(0)}$, is exponentially distributed, and is independent of the first queue length. We now use simulation to show that our results are approximately true even outside of the asymptotic regime and are a useful tool in managerial decision-making.

We investigate four questions in this section. We first investigate the accuracy of the approximation $\varepsilon \E[Q_2]\approx M^{(0)}$ or, equivalently, that $\E[Q_2] \approx \frac{M^{(0)}}{1-\rho}$ when $\rho$ is large. Next, we investigate whether $Q_1$ and $Q_2$ are meaningfully correlated in this high $\rho$ environment. We then investigate whether $\varepsilon Q_2$ is approximately distributed as an exponential random variable with mean $M^{(0)}$. Finally, we investigate the usefulness of our heavy-traffic approximation in a common operations management task: selecting the optimal capacity in a queueing system.

To answer each question, we directly simulate the hybrid-abandonment tandem system and observe its behavior. In answering the first three of our four questions, we consider $\mu_1\in \{1.1, 1.2, 1.5, 2, 4\}$, $\rho \in \{.8, .9, .95, .99\}$, and $\theta \in \{.1, .3, .5\}$. We consider each possible combination of $\mu_1$, $\theta_1$, and $\rho$ and fix $\mu_2=1$ resulting in 60 distinct parameter sets. 

Finally, when studying the approximation's use in optimal capacity selection, we additionally consider a holding cost parameter $c\in \{.5, .75, 1, 1.5, 2\}$ and let $\mu_2$ be the decision variable in a simple optimization model.

Full details on the simulation and a link to the data are available in Appendix~\ref{app:sim}.



\subsection{Question 1: Accuracy of Heavy-Traffic Approximation}
\label{sec:q1}

\cref{Thm3} states that as $\rho \to 1$, the scaled second queue length $(1-\rho) Q_2$ converges to $M^{(0)}$. This gives rise to our heavy traffic approximation for $\E[Q_2]$, $\E[Q_2] \approx \frac{M^{(0)}}{1-\rho}$. 
In each of the 60 parameter sets described above, we simulate our hybrid-abandonment tandem system to find and report an estimate of $\E[Q_2]$ and the standard deviation of our estimator. We then compare our simulation estimate of $\E[Q_2]$ to our heavy-traffic estimate, $\frac{M^{(0)}}{1-\rho}$ and report the difference, $E[Q_2] - \frac{M^{(0)}}{1-\rho}$, and relative difference, ${(\E[Q_2] - \frac{M^{(0)}}{1-\rho})/\E[Q_2]}$. 

To maintain readability, in \cref{tbl:mean_conv} we aggregate our results based on load $\rho$. Specifically, rather than reporting the results when $(\mu_1,\rho,\theta) = (1.1, .8, .1)$, $(\mu_1,\rho,\theta) = (1.2, .8, .3)$, $(\mu_1,\rho,\theta) = (1.5, .8, .5)$, etc. separately, we average our four reported statistics over all of these parameter sets and report this final average. We choose to collapse along $\rho$ values as our results are far more dependent on load $\rho$ than on $\mu_1$ or $\theta$.
\begin{inline-table}{Difference between empirical mean $\E[Q_2]$ and approximate mean $\frac{M^{(0)}}{1-\rho}$.}
\begin{center}
    \begin{tabular}{|c||c|c| c | c|} \hline \label{tbl:mean_conv} 
         $\rho$ & Average $\E[Q_2]$ & Average SD of $\E[Q_2]$ & Average Gap & Average Rel Gap \\
         \hline\hline
         $0.80$ & $3.66$ &  $0.001$ & $-0.758$ & $-17.0 \%$\\
         \hline
         $0.90$ & $8.08$ & $0.003$ & $-0.754$ & $-8.5\%$\\
         \hline
         $0.95$ & $16.9$ & $0.014$ & $-0.765$ & $-4.3\%$ \\
         \hline
         $0.99$ & $87.6$ &  $0.330$ & $-0.745$ & $-0.9\%$ \\
         \hline
    \end{tabular}
\end{center}
\end{inline-table}


As $\rho\to 1$, we can notice two related trends. First, the average gap between the empirical mean and the approximate mean is roughly constant across all values of $\rho$ at about $-0.75$ jobs. Thus, as the empirical mean rises with $\rho$, the relative gap shrinks from $-17\%$ to under $-1\%$.

Based on the results presented in \cref{tbl:mean_conv}, we can conclude that when $\rho \geq 0.8$ our heavy-traffic approximation is consistently within 1 customer of the true mean queue length value. Even though the absolute difference between the approximation and the true mean queue length does not vanish as $\rho$ increases from $0.8$ toward $1$, the relative difference between our approximation and the true mean queue length shrinks from around 17\% toward 0\% as $\rho$ increases toward 1.

\begin{observation}
    The heavy-traffic approximation is very accurate for loads above $0.8$. The absolute difference between the approximation and the true value is roughly constant across $\rho$ values while the relative difference between the two shrinks toward $0$ as $\rho$ increases toward $1$.
\end{observation}
\subsection{Question 2: Validation of Non-correlation}
\label{sec:q2}
\cref{Thm3} states that as $\rho\to 1$, the first and second queue lengths are asymptotically independent. To show that this is approximately true even outside of our asymptotic regime, we simulate the hybrid-abandonment tandem system to empirically estimate the correlation between $Q_1$ and $Q_2$. 
In \cref{table:corr}, we present the empirical Pearson correlation coefficient between $Q_1$ and $Q_2$ across various parameter sets. As in \cref{sec:q1}, we again aggregate our results along $\rho$ values to maintain readability. We report the average correlation coefficient among the 15 parameter sets at each $\rho$ value as well as the maximum and minimum correlation coefficient at each $\rho$ value.
\begin{inline-table}{Correlation between $Q_1$ and $Q_2$.}
\begin{center} \label{table:corr}
    \begin{tabular}{|c||c|c|c|c|c}
    \hline
     $\rho$ & Average Correlation & Max Correlation & Min Correlation \\
     \hline\hline$0.8$& $-0.042$ & $-0.001$ & $-0.088$ \\
     \hline$0.9$& $-0.032$ & $-0.001$ & $-0.078$ \\
     \hline$0.95$& $-0.020$ & $-0.001$ &  $-0.057$\\
     \hline$0.99$& $-0.005$&  $0.000$ & $-0.017$ \\
     \hline
\end{tabular}
\end{center}
\end{inline-table}

Two trends in the data are immediately noticeable: first, across all our parameter sets, the correlation is extremely weak. Although there is no universally agreed benchmark for a ``weak" correlation, a common standard is that of \citet{cohen1998} which defines any correlation coefficient with magnitude less than $0.1$ as very weak or negligible. The highest magnitude of a correlation coefficient we observe in any parameter set is $0.088$ as seen in \cref{table:corr}, thus the correlation is negligible in each of the $60$ parameter sets studied. 

Second, the magnitude of the correlation decreases as $\rho$ increases toward $1$. Both the maximum magnitude of the correlation and the average magnitude of the correlation decrease as $\rho$ increases. Therefore, we can treat the two queue lengths as uncorrelated even outside of our asymptotic regime.

\begin{observation}
    At load above $0.8$, the two queue lengths are not meaningfully correlated. The small amount of correlation that does exist between the two decreases in magnitude as the load increases.
\end{observation}

\subsection{Question 3: Validation of Exponential Distribution}
\label{sec:q3}
\cref{Thm3} states that as $\rho \to 1$, the scaled queue length $(1-\rho) Q_2$ is asymptotically distributed as an exponential with mean $M^{(0)}$. 
To compare the empirical distribution of $(1-\rho)Q_2$ outside of this limit with an exponential with mean $M^{(0)}$, we compute the Kolmogorov-Smirnov statistic or KS statistic, a form of ``distance" between two distributions. The KS statistic between two random variables, $X$ and $Y$ is given by
\begin{align*}
    KS(X,Y) = \sup_{t} |F_X(t)-F_Y(t)|.
\end{align*}
Given \cref{Thm3}, we would expect that as we consider $\rho$ closer and closer to $1$ we should see the KS gap between $(1-\rho)Q_2$ and the exponential with mean $M^{(0)}$ should decrease toward $0$. 

In \cref{tbl:KSexp} we present the average and maximum KS statistics between the scaled empirical queue length and the exponential distribution with mean $M^{(0)}$ and the average standard deviation of the KS statistic of our different parameter sets aggregated by their $\rho$ value.

\begin{inline-table}{KS gap between $(1-\rho)Q_2$ and Exponential with mean $M^{(0)}$}.
\begin{center}
    \begin{tabular}{|c||c|c|c|} \hline
    \label{tbl:KSexp}
        $\rho$ & Average KS Stat & Max KS Stat & Avg KS Stat Std\\
        \hline\hline
         $0.8$ & $0.2000$ & $0.2001$ & $0.0003$ \\
         \hline
         $0.9$ & $0.1000$ & $0.1003$ & $0.0006$\\ 
         \hline
         $0.95$ & $0.0500$ & $0.0503$ &  $0.0008$\\
         \hline $0.99$ & $0.0129$ & $0.0222$ & $0.0033$\\
         \hline
    \end{tabular}
\end{center}
\end{inline-table}

We can observe from \cref{tbl:KSexp} that the average and max KS statistics are approximately equal to $1-\rho$. A gap of $1-\rho$ is the minimum gap that we can achieve simply due to the mismatch between the discrete empirical distribution and the continuous exponential distribution it is converging to. In any stable system with load $\rho$, the system is idle $(1-\rho)$ of the time, so $F_{(1-\rho)Q_2}(0) =F_{Q_2}(0) = 1-\rho$. In contrast, the exponential distribution, a positive, continuous variable, has no mass at or below $0$, so $F(0) = 0$. Thus, the KS gap between $(1-\rho)Q_2$ and the exponential with mean $M^{(0)}$ is bounded below by $1-\rho$.

To circumvent this floor on the KS gap we also compare our empirical distribution to a discrete random variable. A natural candidate is a geometric random variable with our desired mean. This would coincide with distribution of the queue length of a simple $M/M/1$ that has the same asymptotic behavior as our $Q_2$.

In \cref{tbl:KSgeom} we present the average and maximum KS statistics between the empirical queue length and the geometric distribution with mean $\frac{M^{(0)}}{1-\rho}$ and the average standard deviation of the KS statistic of our different parameter sets aggregated by their $\rho$ value.
\begin{inline-table}{KS gap between $Q_2$ and Geometric with mean $\frac{M^{(0)}}{1-\rho}$}.
\begin{center}
    \begin{tabular}{|c||c|c|c|c|}\hline
    \label{tbl:KSgeom}
        $\rho$ & Average KS Stat & Max KS Stat & Avg KS Stat Std\\
        \hline\hline
         $0.8$ & $0.0608$  & $0.0810$ & $0.0003$\\
         \hline
         $0.9$ & $0.0314$ & $0.0473$ & $0.0006$\\
         \hline
         $0.95$ & $0.0163$ & $0.0277$ & $0.0008$\\
         \hline 
         $0.99$ & $0.0097$ & $0.0200$ & $0.0033$ \\
         \hline
    \end{tabular}
\end{center}
\end{inline-table}

Comparing the KS distance from $(1-\rho)Q_2$ to the exponential with mean $M^{(0)}$ and the geometric with mean $\frac{M^{(0)}}{1-\rho}$, we see that the geometric with mean $\frac{M^{(0)}}{1-\rho}$ is a much better approximation. At all loads studied, both the average and maximum KS statistic observed are lower in the geometric case than the exponential case. This improvement is far larger at smaller values of $\rho$ but even at the maximum value considered ($\rho=.99$), the geometric variable is a small but non-negligible improvement over the exponential in fitting the observed data.

Based on the data in \cref{tbl:KSexp} and \cref{tbl:KSgeom}, we can see that the exponential with mean $M^{(0)}$ is a good fit for the scaled queue length $(1-\rho) Q_2$ only when $\rho$ is fairly close to 1. The geometric with mean $\frac{M^{(0)}}{1-\rho}$ is a much better fit for $Q_2$ at all loads studied.

\begin{observation}
    The scaled distribution $\varepsilon Q_2$ is well approximated by an exponential, and an even better fit is obtained from approximating $Q_2$ by a geometric distribution. The accuracy of both approximations increases significantly as the load of the system rises.
\end{observation}

\subsection{Question 4: Application to Capacity Management}
\label{sec:q4}
Capacity planning is a classic problem in operations management (see \citet{kleinrock1975queueing,holt1960planning}). A firm must decide how much capacity (such as machines in a factory, servers in a data center, tables in a restaurant, etc.) to buy before it observes its demand. Capacity is costly, so firms have an incentive to buy as little as necessary. However, buying too little means either not satisfying all demand or otherwise delivering significantly degraded service. Firms therefore have to carefully consider how much capacity they purchase.
Consider a simple static capacity planning problem: A hybrid-abandonment tandem system serves customers and the manager must pick the optimal capacity $\mu_2$ for the second stage server. The operator has to consider both the cost of the capacity and the cost of waiting customers' goodwill lost during the wait. We assume that capacity costs $1$ per unit of capacity per unit of time and that goodwill cost accrued due to waiting at a rate of $c$ per customer in the second queue per unit of time. We thus have an optimization problem of the form
\begin{align*}
    \min_{\mu_2} \ C(\mu_2) = \min_{\mu_2} \ (c\E[Q_2] + \mu_2).
\end{align*}

We define $\mu_2^*$ to be the optimizer of this problem.

Because we do not have an expression for a tractable exact $\E[Q_2]$ in terms of system primitives, any attempt to minimize the cost from first principles of optimization would require repeated simulation to estimate $E[Q_2]$ under different parameter sets. We can, however, easily compute our heavy traffic approximation for $\E[Q_2]$, $\frac{M^{(0)}}{1-\rho}$. This creates an ``approximation problem",
\begin{align*}
    \min_{\mu_2} \ \hat{C}(\mu_2) = \min_{\mu_2} \ (c \frac{M^{(0)}}{1-\rho} + \mu_2)= \min_{\mu_2} \ (c \frac{M^{(0)}}{1-\lambdabar/
    \mu_2
    } + \mu_2),
\end{align*}
where $\lambdabar=(1-\pi_0)\mu_1$ is the average arrival rate to the second queue. This approximation problem does have a simple analytical solution: $\hat{\mu}^*_2 = \lambdabar + \sqrt{\frac{\lambdabar M^{(0)}}{c}}$. 

We define the ``excess cost" $\Delta = C(\hat{\mu}_2^* ) - C(\mu_2^*)$ as the additional cost of running the system at $\mu_2=\hat{\mu}^*_2$ rather than $\mu_2=\mu_2^*$. Similarly, we define the ``relative excess cost" as the relative increase in cost under the same scenario, $\Delta/C(\mu^*_2)$.
\begin{inline-table}{Optimal and Approximate $\mu_2$ Values and Costs}
\begin{center}
\label{table:cost}
    \begin{tabular}{|c||c|c|c||c|c||c|c|}
    \hline
         $c$ &Avg $\mu_2^*$ & Avg $\rho^*$ & \hspace{-8pt} Avg $\hat{\mu}_2^*$ \hspace{-8pt}~& \hspace{-8pt} Avg $C(\mu_2^*)$ \hspace{-8pt} & \hspace{-8pt} Avg $\hat{C}(\hat{\mu}_2^*)$ \hspace{-8pt} & Max $\Delta$ & \hspace{-8pt} Max $\Delta/C(\mu^*_2)$ \hspace{-8pt}~\\
         \hline\hline
         $0.5$ & $0.530$ & $0.403$ & $0.550$ &$0.8417$ & $1.369$ & $0.0067$ & $0.55\%$\\\
         $0.75$ & $0.598$ & $0.359$ & $0.623$ & $0.9820$ & $1.762$ &  $0.0076$ & $0.58\%$\\
         $1$ & $0.658$ & $0.324$ & $0.684$ & $1.101$ & $2.132$ & $0.0090$ & $0.60\%$ \\
         $1.5$ & $0.756$ & $0.283$ & $0.787$ & $1.302$ & $2.832$ & $0.0080$ & $0.60\%$\\
         $2$ & $0.838$ & $0.254$ & $0.874$ & $1.472$ & $3.501$ & $0.0110$ & $0.60\%$ \\\hline
    \end{tabular}
\end{center}
\end{inline-table}

\cref{table:cost} demonstrates two somewhat at-odds observations. Observing the Avg $C(\mu_2^*)$ and Avg $\hat{C}(\hat{\mu}_2^*)$ columns, we can see that the optimal cost in the approximation problem is consistently much larger than the optimal cost in the true problem. At different $c$ values, the optimal approximation cost avg $\hat{C}(\hat{\mu}_2^*)$ is anywhere from $1.5$ to $2$ times larger than the optimal true cost avg $C(\mu_2^*)$. However, even though the optimal values in these two problems diverge starkly, the absolute excess cost Max $\Delta$ remains extremely small, below $0.0110$. Across all of the parameter sets we consider, the relative excess cost Max $\Delta/C(\mu_2^*)$ is below $0.6\%$ as well. This fits well with the results seen in \cref{tbl:mean_conv} which showed that the difference between the empirical mean queue length and the approximate mean queue length is roughly constant for different values of $\rho$. As a result, the approximation problem seems very similar to a shifted version of the original problem.



\begin{observation}
    The heavy-traffic approximation for the mean queue length is a useful tool for capacity management even when optimal load $\rho$ is comparatively small. Using the heavy-traffic approximation rather than the true mean queue length results in a similar optimal capacity level and extremely low excess cost, both in absolute and in relative terms.

\end{observation}
    The approximation does seem to consistently overestimate the cost of running the system, however, meaning that it may not be broadly applicable for forecasting either average queue length or long run average running cost. 

\section{Conclusion}

We present the first heavy-traffic analysis of the hybrid-abandonment tandem queue. We characterize the scaled second queue length by finding its asymptotic mean and showing it is asymptotically exponentially distributed and asymptotically independent of the first queue length. We do this by analyzing countably infinite Markovian arrival process (IMAP) queues, a class that includes our original hybrid-abandonment system. We provide sufficient conditions for the scaled queue length of an IMAP queue to be asymptotically exponential and asymptotically independent of its phase process and verify these conditions hold in the case of a hybrid-abandonment queue. Further, we show empirically that our heavy-traffic results are useful approximations beyond heavy traffic.

A promising area of future research is similar analysis of queues with infinite Markovian service processes (IMSPs), either in lieu of or in addition to IMAP arrivals. We believe IMSPs will be slightly harder to analyze due to unused service occurring when the queue is empty. Another promising extension is to provide an analogous bound on the visitation matrix $\boldsymbol{R}$ and the steady-state distribution of the phase process $\bpi$ for phase processes beyond birth-death processes. Of particular interest would be bounds for queueing networks to extend the heavy traffic analysis from two-stage systems to more complex queueing networks.

\section*{Acknowledgments}

This material is based upon work supported by the National Science Foundation under Grant Number CMMI-2514091. We are thankful for the monetary support provided by the NSF and to the staff who made it possible.
\section*{Disclaimer}
Any opinions, findings, and conclusions or recommendations expressed in this material are those of the author(s) and do not necessarily reflect the views of the National Science Foundation.
\section*{Code and Data Disclosure}
The code and generated data used in the numerical experiments \cref{Sec6} are available at the author's GitHub repository \href{https://github.com/CJCurtis/HybridAbandonmentSimulator}{here}.

\begin{appendices}

\section{Notation}

In \cref{tbl:notation} we present a full list of notation used in this paper.
    
\begin{table}[]
\caption{Notation Used}
\label{tbl:notation}

    \begin{tabular}{|c |l |c|}

    \hline

    Symbol & Meaning & Formula\\
    \hline
    $Q$ & Queue Length Distribution &\\
    $Y$ & Phase Process Distribution &\\
    \hline
    $\alpha_{i,j}$ & Phase process transition rate from state $i$&\\ & to state $j$ & \\
    $\boldsymbol{A}$ & Generator Matrix of $Y$ & $[\alpha_{i,j}]_{i,j\in\mathbb{N}_0}$\\
    $\lambda_i$ & In-state arrival Rate in Modulation State $i$ & \\
    $p_{i,j}$ & Probability of a between-state arrival due   & \\
    & to phase process transition from state $i$ &\\
    & to state $j$ & \\
    $\beta_{i,j}$ & Rate of queue arrivals caused by $i$ to $j$  & $\alpha_{i,j}p_{i,j}$\\
    &transitions in phase process&\\
    $\lambdabar$ & Long run average arrival rate & $\E[\lambda_Y+\sum_{j\ne Y} \beta_{Y,j}]$\\
    $\mu$ & Queue Service Rate & \\ \hline

    $\pi_{i,j}$& Steady State probability of $\Le(i,j\Ri)$  & $\p\Le(Y=i, Q=j\Ri)$ \\
    $\pi_i$ & Steady State Probability of $i$ & $\sum_{j\in \mathbb{N}_0} \pi_{i,j}$\\
    $\boldsymbol{\pi}$ & Steady State Distribution vector of $Y$ & $\{\pi_i\}_{i\in \mathbb{N}_0}$\\
    $\piidle$&Idleness in modulation state $i$ & $P(Y=i, Q=0)$\\
        $\boldsymbol{\pi}_{\overline{0}}$ & Steady state distribution $\boldsymbol{\pi}$ with the  & $\{\pi_i\}_{i\in \mathbb{N}_0/0}$\\
        &$0$ element removed&\\
    $\varepsilon$ & Idle Capacity & $\p\Le(Q=0\Ri)$\\
        \hline

    $R$ & Near-inverse matrix & See \cref{sec:near-inverse}\\
    $e^T$ & Ones vector & $\Le(1,1,\dots\Ri)^T$\\
    $\boldsymbol{F}$ & Arrival Rates Matrix & $[ F]_{i,j}=\begin{cases}
        \lambda_i, \quad \ \  i=j\\
        \beta_{i,j}, \quad i\ne j
    \end{cases}$\\
    $\boldsymbol{I}$ & Identity Matrix & $\text{diag}\Le(1,1,\dots\Ri)$\\
        \hline

    $\hat{q}\Le(z;i\Ri)$ & Partial $z$-transform with $Y=i$ & $\E[z^Q\indic{Y=i}]$\\
    $\hat{\boldsymbol{{q}}}\Le(z\Ri)$ & $z$-transform vector & $\{\hat{q}\Le(z;i\Ri)\}_{i\in S}$\\
    $\overline{q}(z)$ & $z$-transform of $Q$ & $\E[z^Q]$\\
        \hline

    $\cdot^{\Le(\varepsilon\Ri)}$ & Quantity in the system with idleness $\varepsilon$  & \\
    &in a parameterized family&\\
    $f\Le(z;\varepsilon\Ri)$&Function $f$ evaluated at $z$ where &\\& underlying terms (such as $\boldsymbol{F}$ or $\bpi$) &\\&
    are from the $\varepsilon$ system&\\

    \hline
    $\p(a_1\to a_2\to a_3/b_1, b_2)$ & Probability the phase process, &\\& starting in $a_1$, passes through $a_2$ then $a_3$ &\\
    &(possibly with repetition)&\\ &without entering $b_1$ or $b_2$ &\\
    \hline

    \end{tabular}

    \end{table}

\section{Proofs}
\subsection{\cref{R_exst_lemma} Proof}
\label{R_exst_proof}
\begin{proof}
    The element $[\boldsymbol{R}]_{i,j}$, which we call $r_{i,j}$ for notational convenience, is the expected time the process spends in state $j$ between starting in state $i$ and reaching the $0$ state for the first time. We can represent this as
    \begin{align*}
        [\boldsymbol{R}]_{i,j}=r_{i,j} = \E[\sum^{N_{i,j}}_{k=1} T_k],
    \end{align*}
    where $N_{i,j}$ is the number of visits the process makes to state $j$ between starting in state $i$ and reaching state $0$ and $T_k$ is the length of time spent in state $j$ during the $k^{\text{th}}$ visit to the state.

    Because $T_1, T_2, \dots, T_k, \dots$ is a sequence of i.i.d. exponential random variables with mean $\frac{1}{-\alpha_{j,j}}<\infty$, if $\E[N_{i,j}]<\infty$ then $r_{i,j} = \E[N_{i,j}]\E[T_0] < \infty$ by Wald's lemma. Thus, if we can show that $\E[N_{i,j}]<\infty$ for all $i$ and $j$, then $\boldsymbol{R}$ is well-posed.

    First, let us show that $\E[N_{i,i}]<\infty$ for all $i\in \mathbb{N}_0$ and then show the same for when $i\ne j$. Because the chain is positive recurrent, with probability $1$ the chain to return to state $i$ infinitely many times after starting there. Because the chain is irreducible, there is some positive probability $p_i$ that chain reaches state $0$ before returning to state $i$. Thus, $N_{i,i}\sim \text{Geom}\Le(p_i\Ri)$ and $\E[N_{i,i}] = \frac{1-p_i}{p_i}<\infty$.

    Next consider when $i\ne j$. Consider two possibilities: starting in state $i$ the chain either reaches $0$ before $j$ or vice versa. If the chain reaches $0$ before $j$, then, by definition, $N_{i,j}=0$. Conversely if the chain reaches $j$ before $0$ then the expected number of returns to $j$ is exactly the same as if the process simply began in state $j$. Therefore 
        \[\E[N_{i,j}] = P\Le(\text{Reach $j$ before $0$}\Ri)\E[N_{j,j}]\leq \E[N_{j,j}]<\infty.\]

    Thus, $\E[N_{i,j}]<\infty$ for all $i,j$ and $r_{i,j} = \frac{\E[N_{i,j}]}{-\alpha_{j,j}} < \infty$. The $\boldsymbol{R}$ matrix is well posed.

    Next, let us show that our proposed $\boldsymbol{R}$ has the near-inverse property. Notice that the quantity $N_{i,j}$ is an element of the fundamental matrix of an absorbing DTMC, particularly the embedded DTMC of our CTMC after making $0$ an absorbing state. Let $\boldsymbol{H}_{\overline{0},\overline{0}}$ and  be the one step transition matrix of the transient states in this absorbing embedded DTMC, i.e., all state except $0$.
    \[\boldsymbol{N}\boldsymbol{I}+\boldsymbol{H}_{\overline{0},\overline{0}}+\Le(\boldsymbol{H}_{\overline{0},\overline{0}}\Ri)^2+\dots = \boldsymbol{I}+\boldsymbol{NH}_{\overline{0},\overline{0}} = \boldsymbol{I}+\boldsymbol{H}_{\overline{0},\overline{0}}\boldsymbol{N}\]
    or, after rearranging,
    \[\boldsymbol{I}=\boldsymbol{N}\Le(\boldsymbol{I}-\boldsymbol{H}_{\overline{0},\overline{0}}\Ri)=\Le(\boldsymbol{I}-\boldsymbol{H}_{\overline{0},\overline{0}}\Ri)\boldsymbol{N}.\] This is precisely $\boldsymbol{N}=\Le(\boldsymbol{I}-\boldsymbol{H}_{\overline{0},\overline{0}}\Ri)^{-1}$ and therefore
    \[-\Le(\boldsymbol{W}_{\overline{0},\overline{0}}\Ri)^{-1}\boldsymbol{N} = \Le(\boldsymbol{W}_{\overline{0},\overline{0}}\Ri)^{-1}\Le(\boldsymbol{H}_{\overline{0},\overline{0}}-\boldsymbol{I}\Ri) = \Le(\boldsymbol{A}_{\overline{0}, \overline{0}}\Ri)^{-1}.\]
    where $\boldsymbol{W}_{\overline{0},\overline{0}}$ is the holding rates matrix for all states except $0$ and $\boldsymbol{A}_{\overline{0}, \overline{0}} = [\alpha_{i,j}]_{i,j\in \mathbb{N}_0/0}$.

    Thus, \begin{align*}
        r_{i,j}= \begin{cases}
            \Le[(\boldsymbol{A}_{\overline{0}, \overline{0}})^{-1}\Ri]_{i,j}, \qquad & i,j\in \mathbb{N}_0/0\\
            0, & \text{else}
        \end{cases}
    \end{align*}

    The steady-state balance equation provides that \[\pi_0 \Le(\alpha_{0, 1}, \alpha_{0, 2}, \dots\Ri) + \Le(\pi_1, \pi_2, \dots\Ri) \boldsymbol{A}_{\overline{0}, \overline{0}} = 0\] or, equivalently, \[\Le(\alpha_{0,1}, \alpha_{0,2}, \dots\Ri) \Le(\boldsymbol{A}_{\overline{0}, \overline{0}}\Ri)^{-1} = -\frac{1}{\pi_0}\Le(\pi_1, \pi_2, \dots\Ri).\]

    Thus, we can now construct our matrix $\boldsymbol{AR}$. For $i,j\in \mathbb{N}_0/0$, $[\boldsymbol{A}_{\overline{0}, \overline{0}}\boldsymbol{R}]_{i,j} = \indic\{i=j\}$. When $i=0$, $[\boldsymbol{A}_{\overline{0}, \overline{0}}\boldsymbol{R}]_{i,j} = 0$ because $[\boldsymbol{R}]_{0,j} = 0 $ for all $j\in \mathbb{N}_0$. Lastly, when $i\in \mathbb{N}_0/0$, $[\boldsymbol{A}_{\overline{0}, \overline{0}}\boldsymbol{R}]_{i,0} = \frac{\pi_i}{\pi_0}$.
    We have our desired result that 
    \begin{align*}
    \boldsymbol{AR}
 = \begin{bmatrix}
        0 & -\frac{1}{\pi_0} \bpi_{\overline{0}}\\
        0& \boldsymbol{I}
    \end{bmatrix}.
    \end{align*}
    
\end{proof}

\subsection{\cref{ind_exp_lemma} Proof}
\label{ind_exp_proof}
\begin{proof}
    We begin with the balance equation for the IMAP queue. For any state $\Le(y,q\Ri)$,
    \begin{align*}
         \sum_{i\ne y} \alpha_{i,y} \Le(p_{i, y} \pi_{i, q-1}\indic\{q>0\} + \Le(1-p_{i,y}\Ri)\pi_{i,q}\Ri) + \lambda_{y} \pi_{y,q-1}\indic\{q>0\} + \mu\pi_{y,q+1} \\= \Le(\sum_{i\ne y}\alpha_{y,i} + \lambda_{y}+\mu\indic\{q>0\}\Ri) \pi_{y,q}
    \end{align*}

    Multiplying both sides of the equation by $z^q$ and then summing over all values of $q$ we can convert the above balance equation into a partial $z$-transform balance equation.
    
    \begin{align*}
        \sum_{q=0}^\infty(\sum_{i\ne y} &\alpha_{i,y} (p_{i, y} \pi_{i, q-1}\indic\{q>0\} + (1-p_{i,y})\pi_{i,q}) + \lambda_{y} \pi_{y,q-1}\indic\{q>0\} + \mu\pi_{y,q+1})z^q) =\\
        & \sum_{q=0}^\infty(\sum_{i\ne y} \alpha_{i,y} (p_{i, y} \pi_{i, q-1}\indic\{q>0\} + (1-p_{i,y})\pi_{i,q}))z^q\\
        +&\sum_{q=0}^\infty(\lambda_{y} \pi_{y,q-1}\indic\{q>0\}z^q) + \sum_{q=0}^\infty(\mu\pi_{y,q+1}z^q)=\\
        &\sum_{n=0}^\infty(\sum_{i\ne y} \alpha_{i,y} (p_{i, y} \pi_{i, n})z^{n+1}) + \sum_{q=0}^\infty(\sum_{i\ne y} \alpha_{i,y}(1-p_{i,y})\pi_{i,q})z^q \\
        +&\sum_{n=0}^\infty(\lambda_y \pi_{y,n}z^{n+1}) + \sum_{n=0}^\infty(\mu \pi_{y,n})z^{n-1} - \mu \pi_{y, 0} z^{-1}=\\
        &z\sum_{i\ne y} \alpha_{i,y}p_{i, y}\Qi(z;i) + \sum_{i\ne y} \alpha_{i,y}(1-p_{i, y})\Qi(z;i) + z\Qi(z;y) + z^{-1}\mu(\Qi(z;y)-\pi_{y,0})\\
        &= (\lambda_y z + \frac{\mu}{z})\Qi(z;y) + \sum_{i\ne y}\alpha_{i,y}(p_{i,y}z + (1-p_{i,y}))\Qi(z;i) - \frac{\mu}{z} \pi_{y,0}
    \end{align*}
    and
    \begin{align*}
        \sum_{q=0}^\infty \Le( \Le(\sum_{i\ne y}\alpha_{y,i} + \lambda_{y}+\mu\indic\{q>0\}\Ri) \pi_{y,q}z^q\Ri)&= \sum_{q=0}^\infty \Le( \Le(\sum_{i\ne y}\alpha_{y,i} + \lambda_{y}+\mu\Ri) \pi_{y,q}z^q\Ri) - \mu\pi_{y,0}\\
        &=\Le(\sum_{i\ne y}\alpha_{y,i} + \lambda_{y}+\mu\Ri)\sum_{q=0}^\infty \Le(\pi_{y,q}z^q\Ri) - \mu\pi_{y,0}\\
        &= \Le(\sum_{i\ne y}\alpha_{y,i} + \lambda_{y}+\mu\Ri) \Qi(z;y) - \mu\pi_{y,0} 
    \end{align*}

    Re-equating the two we have
    \begin{align*}
        \Le(\lambda_y z + \frac{\mu}{z}\Ri)\Qi\Le(z;y\Ri) + \sum_{i\ne y}\alpha_{i,y}\Le(p_{i,y}z + \Le(1-p_{i,y}\Ri)\Ri)\Qi\Le(z;i\Ri) - \frac{\mu}{z} \pi_{y,0} \\= \Le(\sum_{i\ne y}\alpha_{y,i} + \lambda_{y}+\mu\Ri) \Qi\Le(z;y\Ri) - \mu\pi_{y,0} 
    \end{align*}
    We can simplify this by adding \\$\alpha_{y,y}\Le(p_{y,y}z + \Le(1-p_{y,y}\Ri)\Ri)\Qi\Le(z;y\Ri) = \alpha_{y,y}\Qi\Le(z;y\Ri)$ and $\sum_{i}\alpha_{i,y}\Le(p_{i,y}\Le(1-z\Ri)\Ri)\Qi\Le(z;i\Ri)$ to both sides to find
    \begin{align*}
        \Le(\lambda_y z + \frac{\mu}{z}\Ri)\Qi\Le(z;y\Ri) + \sum_{i}\alpha_{i,y}\Le(p_{i,y} + \Le(1-p_{i,y}\Ri)\Ri)\Qi\Le(z;i\Ri) - \frac{\mu}{z} \pi_{y,0} \\= \Le(\lambda_{y}+\mu\Ri) \Qi\Le(z;y\Ri) + \sum_{i}\alpha_{i,y}\Le(p_{i,y}\Le(1-z\Ri)\Ri)\Qi\Le(z;i\Ri)- \mu\pi_{y,0} 
    \end{align*}

    Rearranging terms we have 
    \begin{align*}
         &\sum_{i}\alpha_{i,y}\Qi\Le(z;i\Ri)   \\
         &=\Le(\lambda_{y}+\mu\Ri) \Qi\Le(z;y\Ri) - \Le(\lambda_y z + \frac{\mu}{z}\Ri)\Qi\Le(z;y\Ri) + \sum_{i}\alpha_{i,y}\Le(p_{i,y}\Le(1-z\Ri)\Ri)\Qi\Le(z;i\Ri)- \mu\pi_{y,0} +\frac{\mu}{z} \pi_{y,0}\\
         &=\lambda_y\Le(1-z\Ri)\Qi\Le(z;y\Ri) + \mu\Le(1-\frac{1}{z}\Ri)\Qi\Le(z;y\Ri) + \sum_{i}\alpha_{i,y}\Le(p_{i,y}\Le(1-z\Ri)\Ri)\Qi\Le(z;i\Ri)- \mu\Le(1-\frac{1}{z}\Ri)\pi_{y,0}
    \end{align*}

    We can collect each of these equations into a single matrix-form balance equation
    \begin{align*}
        \Q\Le(z\Ri)\boldsymbol{A} = \Q\Le(z\Ri)\Le( \boldsymbol{F}\Le(1-z\Ri)+\mu\Le(1-\frac{1}{z}\Ri)\boldsymbol{I}\Ri) - \mu\Le(1-\frac{1}{z}\Ri)\piidle
    \end{align*}

    Right multiplying both sides by $\boldsymbol{R}$ we have the following relationship
    \begin{align}
    \label{eqn:qh0}
        \Q\Le(z\Ri) = \Qi\Le(z;0\Ri) \frac{\bpi}{\pi_0} +\Q\Le(z\Ri)\Le( \boldsymbol{F}\Le(1-z\Ri)+\mu\Le(1-\frac{1}{z}\Ri)\boldsymbol{I}\Ri)\boldsymbol{R} - \mu\Le(1-\frac{1}{z}\Ri)\piidle\boldsymbol{R}
    \end{align}
    \cref{eqn:qh0} provides us $\Q(z)$ in terms of $\Qi(z;0)$ and $\Q(z)$ itself. By substituting our equation for $\Q(z)$ into itself we can yield an equation that places more emphasis on $\Qi(z;0)$
    This iterated equation is our near-independence equation
    \begin{align*}
        &\Q\Le(z\Ri) = \Qi\Le(z;0\Ri)\frac{\bpi}{\pi_0}\Le(I + \Le( \boldsymbol{F}\Le(1-z\Ri)+\mu\Le(1-\frac{1}{z}\Ri)\boldsymbol{I}\Ri)\boldsymbol{R}\Ri) \\&+ \Q\Le(z\Ri) \Le(\Le( \boldsymbol{F}\Le(1-z\Ri)+\mu\Le(1-\frac{1}{z}\Ri)\boldsymbol{I}\Ri)\boldsymbol{R}\Ri)^2 - \mu\Le(1-\frac{1}{z}\Ri)\piidle\boldsymbol{R}\Le( \boldsymbol{F}\Le(1-z\Ri)+\mu\Le(1-\frac{1}{z}\Ri)\boldsymbol{I}\Ri)\boldsymbol{R} \\&- \mu\Le(1-\frac{1}{z}\Ri)\piidle\boldsymbol{R}
    \end{align*}
    or, after collecting terms,
    \begin{align*}
        \Q\Le(z\Ri) &= \Qi\Le(z;0\Ri)\frac{\bpi}{\pi_0} +\Qi\Le(z;0\Ri)\frac{\bpi}{\pi_0}\Le( \boldsymbol{F}\Le(1-z\Ri)+\mu\Le(1-\frac{1}{z}\Ri)\boldsymbol{I}\Ri)\boldsymbol{R} + \Le(1-\frac{1}{z}\Ri)\overline{\boldsymbol{Y}}\Le(z\Ri)
    \end{align*}

    If we instead right multiply our matrix-form balance equation by $e^T= \Le(1, 1, \dots, \Ri)^T$ we have a flow balance equation
    \begin{align*}
                0 = \Q\Le(z\Ri)\Le( \boldsymbol{F}\Le(1-z\Ri)+\mu\Le(1-\frac{1}{z}\Ri)\boldsymbol{I}\Ri)e^T - \mu\Le(1-\frac{1}{z}\Ri)\varepsilon
    \end{align*}
    where $\piidle e^T = \varepsilon$.

    Substituting our near-independence equation for $\Q\Le(z\Ri)$ into our flow balance equation we have
    \begin{align*}
        &\Le(\Qi\Le(z;0\Ri)\frac{\bpi}{\pi_0} +\Qi\Le(z;0\Ri)\frac{\bpi}{\pi_0}\Le( \boldsymbol{F}\Le(1-z\Ri)+\mu\Le(1-\frac{1}{z}\Ri)\boldsymbol{I}\Ri)\boldsymbol{R} + \Le(1-\frac{1}{z}\Ri)\overline{\boldsymbol{Y}}\Le(z\Ri)\Ri)\\
        &\cdot \Le( \boldsymbol{F}\Le(1-z\Ri)+\mu\Le(1-\frac{1}{z}\Ri)\boldsymbol{I}\Ri)e^T - \mu\Le(1-\frac{1}{z}\Ri)\varepsilon = 0
    \end{align*}
    Isolating $\Qi\Le(z;0\Ri)$ we find
    \begin{align*}
        \Qi\Le(z;0\Ri) &=\pi_0\frac{\mu\Le(1-\frac{1}{z}\Ri)\varepsilon - \Le(1-\frac{1}{z}\Ri)^2\overline{\boldsymbol{Y}}\Le(z\Ri)\Le(-z \boldsymbol{F} +\mu \boldsymbol{I}\Ri)e^T}{\Le(1-\frac{1}{z}\Ri)\Le(\mu\varepsilon+\overline{\lambda}\Le(1-z\Ri)\Ri)+\Le(1-\frac{1}{z}\Ri)^2\bpi\Le(-z \boldsymbol{F}+\mu \boldsymbol{I}\Ri)\boldsymbol{R}\Le(-z \boldsymbol{F}+\mu \boldsymbol{I}\Ri)e^T}\\
        &=\pi_0\frac{1-\mu^{-1}\varepsilon^{-1}\Le(1-\frac{1}{z}\Ri)\overline{\boldsymbol{Y}}\Le(z\Ri)\Le(-z \boldsymbol{F}+\mu \boldsymbol{I}\Ri)e^T}{1+ \varepsilon^{-1}\frac{\overline{\lambda}}{\mu}\Le(1-z\Ri)+\mu^{-1}\varepsilon^{-1}\Le(1-\frac{1}{z}\Ri)\bpi\Le(-z \boldsymbol{F}+\mu \boldsymbol{I}\Ri)\boldsymbol{R}\Le(-z \boldsymbol{F}+\mu \boldsymbol{I}\Ri)e^T}
    \end{align*}
    or, after condensing terms,
    \begin{align*}
        \Qi\Le(z;0\Ri) &= \pi_0\frac{1+U\Le(z\Ri)}{1+M\ln\Le(z\Ri) /\varepsilon +V\Le(z\Ri)}
    \end{align*}
\end{proof}

\subsection{\cref{DerBDBds} Proof}
\label{DerDBBds_proof}
\begin{proof}
    As seen in \cref{Thm1}, $\boldsymbol{R}$ and $\bpi$ often appear right multiplied by the matrix $\Le(\mu^{\Le(\varepsilon\Ri)}\boldsymbol{I} - z^\varepsilon \boldsymbol{F}^{(\varepsilon)}\Ri)$. We can use the bounds obtained in \cref{FundBDBds} above to derive a number of related bounds on these related quantities.
    \begin{enumerate}
    
        \item $\Le|[\pi^{\Le(\varepsilon\Ri)}\Le(\mu^{\Le(\varepsilon\Ri)}\boldsymbol{I}-z^\varepsilon \boldsymbol{F}^{(\varepsilon)}\Ri)]_j\Ri| \leq \tilde{C}_1a^j$ for some $\tilde{C}_1>0$.

        \begin{align*}
            &\Le|[\pi^{\Le(\varepsilon\Ri)}\Le(\mu^{\Le(\varepsilon\Ri)}\boldsymbol{I}-z^\varepsilon \boldsymbol{F}^{(\varepsilon)}\Ri)]_j\Ri| = \Le|\sum_{m=0}^\infty \pi_m[\mu^{\Le(\varepsilon\Ri)}\boldsymbol{I}-z^\varepsilon  \boldsymbol{F}^{(\varepsilon)}]_{m,j}\Ri|\\
            &=\Le|\pi_{j-1} [\mu^{\Le(\varepsilon\Ri)}\boldsymbol{I}-z^\varepsilon  \boldsymbol{F}^{(\varepsilon)}]_{j-1,j} + \pi_{j} [\mu^{\Le(\varepsilon\Ri)}\boldsymbol{I}-z^\varepsilon  \boldsymbol{F}^{(\varepsilon)}]_{j,j} + \pi_{j+1} [\mu^{\Le(\varepsilon\Ri)}\boldsymbol{I}-z^\varepsilon  \boldsymbol{F}^{(\varepsilon)}]_{j+1,j}\Ri|\\
            &= \Le|\pi_{j-1} [\alpha_{j-1,j}p_{j-1, j}] + \pi_{j} [\mu^{\Le(\varepsilon\Ri)}-z^\varepsilon\lambda_j] + \pi_{j+1} [\alpha_{j+1,j}p_{j+1, j}]\Ri|\\
            &\leq \pi_{j-1} \Le|\alpha_{j-1,j}p_{j-1, j}\Ri| + \pi_{j} \Le|\mu^{\Le(\varepsilon\Ri)}-z^\varepsilon\lambda_j\Ri| + \pi_{j+1} \Le|\alpha_{j+1,j}p_{j+1, j}\Ri|\\
            &\leq \pi_{j-1} \hat{\lambda} + \pi_{j} \Le(\hat{\mu} + \hat{\lambda}\Ri) + \pi_{j+1} \hat{\lambda}\\
            &\leq C_1a^{j-1}\hat{\lambda} + C_1a^{j} \Le(\hat{\mu} + \hat{\lambda}\Ri) + C_1a^{j+1} \hat{\lambda}\\
            &\leq \Le(C_1a^{-1}\hat{\lambda} + C_1 \Le(\hat{\mu} + \hat{\lambda}\Ri) + C_1a \hat{\lambda}\Ri)a^j\\
            &=\tilde{C}_1a^j
        \end{align*}
        for $\tilde{C}_1 = C_1a^{-1}\hat{\lambda} + C_1 \Le(\hat{\mu} + \hat{\lambda}\Ri) + C_1a \hat{\lambda}$
        \item $\Le|[\boldsymbol{R}^{(\varepsilon)}\Le(\mu^{\Le(\varepsilon\Ri)}\boldsymbol{I}-z^\varepsilon \boldsymbol{F}^{(\varepsilon)}\Ri)]_{i,j}\Ri| \leq \tilde{C}_2\tilde{b}^{\Le(j-i\Ri)^+}$ for some $\tilde{C}_2>0$.
        \begin{align*}
            &\Le|[\boldsymbol{R}^{(\varepsilon)}\Le(\mu^{\Le(\varepsilon\Ri)}\boldsymbol{I}-z^\varepsilon \boldsymbol{F}^{(\varepsilon)}\Ri)]_{i,j}\Ri|  = \Le|\sum_{m=0}^\infty r^{\Le(\varepsilon\Ri)}_{i,m}[\mu^{\Le(\varepsilon\Ri)}\boldsymbol{I}-z^\varepsilon \boldsymbol{F}^{(\varepsilon)}]_{m,j}\Ri|\\
            &=\Le| r^{\Le(\varepsilon\Ri)}_{i, j-1}\Le(-z^\varepsilon\alpha_{j-1,j}p_{j-1,j}\Ri) + r^{\Le(\varepsilon\Ri)}_{i, j}\Le(\mu^{\Le(\varepsilon\Ri)}-z^\varepsilon\lambda_j\Ri) + r^{\Le(\varepsilon\Ri)}_{i, j+1}\Le(-z^\varepsilon\alpha_{j+1,j}p_{j+1,j}\Ri)\Ri|\\
            &\leq  r^{\Le(\varepsilon\Ri)}_{i, j-1}\Le|-z^\varepsilon\alpha_{j-1,j}p_{j-1,j}\Ri| + r^{\Le(\varepsilon\Ri)}_{i, j}\Le|\mu^{\Le(\varepsilon\Ri)}-z^\varepsilon\lambda_j\Ri| + r^{\Le(\varepsilon\Ri)}_{i, j+1}\Le|-z^\varepsilon\alpha_{j+1,j}p_{j+1,j}\Ri|\\
            &\leq C_2 \tilde{b}^{\Le(j-1-i\Ri)^+} \hat{\lambda} + C_2\tilde{b}^{\Le(j-i\Ri)^+}\Le(\hat{\lambda}+\hat{\mu}\Ri) + C_2 \tilde{b}^{\Le(j+1-i\Ri)^+}\hat{\lambda}\leq \tilde{C}_2 \tilde{b}^{\Le(j-i\Ri)^+}
        \end{align*}
        for $\tilde{C}_2 = C_2 \tilde{b}^{-1} \hat{\lambda} + C_2\Le(\hat{\lambda}+\hat{\mu}\Ri) + C_2 \tilde{b}^{-1}\hat{\lambda}$.

        \item $\Le|[\boldsymbol{R}^{(\varepsilon)}\Le(\mu^{\Le(\varepsilon\Ri)}\boldsymbol{I}-z^\varepsilon \boldsymbol{F}^{(\varepsilon)}\Ri)e^T]_{i}\Ri| \leq \tilde{C}_3i$ for some $\tilde{C}_3>0$.

        \begin{align*}
            &\Le|[\boldsymbol{R}^{(\varepsilon)}\Le(\mu^{\Le(\varepsilon\Ri)}\boldsymbol{I}-z^\varepsilon \boldsymbol{F}^{(\varepsilon)}\Ri)e^T]_{i}\Ri| = \Le|\sum_{j=0}^\infty[\boldsymbol{R}^{(\varepsilon)}\Le(\mu^{\Le(\varepsilon\Ri)}\boldsymbol{I}-z^\varepsilon \boldsymbol{F}^{(\varepsilon)}\Ri)]_{i,j}\Ri|\\
            \leq &\Le|\sum_{j=0}^{i-1}[\boldsymbol{R}^{(\varepsilon)}\Le(\mu^{\Le(\varepsilon\Ri)}\boldsymbol{I}-z^\varepsilon \boldsymbol{F}^{(\varepsilon)}\Ri)]_{i,j}\Ri| 
            + \Le|\sum_{j=i}^\infty[\boldsymbol{R}^{(\varepsilon)}\Le(\mu^{\Le(\varepsilon\Ri)}\boldsymbol{I}-z^\varepsilon \boldsymbol{F}^{(\varepsilon)}\Ri)]_{i,j}\Ri|\\
            \leq &\tilde{C}_2i + \tilde{C}_2\sum_{j=i}^\infty \tilde{b}^{j-i}
            \leq \tilde{C}_2i + \tilde{C}_2 \frac{1}{1-\tilde{b}}
            \leq \tilde{C}_3i,
        \end{align*}
        for $\tilde{C}_2 = 1+ \frac{1}{1-\tilde{b}}$.

    \end{enumerate}

    With these extended bounds in place, we're finally ready to show that the birth-death uniform stability conditions imply our original uniform stability conditions and, in turn, imply our desired heavy-traffic results.

    \begin{enumerate}
        \item For all $z\in\Le(0,1\Ri)$, as $\varepsilon\to 0$:
        \begin{enumerate}
            \item $\Le(1-z^{-\varepsilon}\Ri)\Q\Le(z^\varepsilon;\varepsilon\Ri) \Le(\Le(-z^\varepsilon \boldsymbol{F}^{(\varepsilon)}+\mu^{\Le(\varepsilon\Ri)} \boldsymbol{I}\Ri)R\Ri)^2\Le(-z^\varepsilon \boldsymbol{F}^{(\varepsilon)}+\mu \boldsymbol{I}\Ri)e^T\to 0$.

            Because $\Le(1-z^{-\varepsilon}\Ri)\to 0$ as $\varepsilon\to 0$, if 
            \begin{align*}
                \Le|\Q\Le(z^\varepsilon;\varepsilon\Ri) \Le(\Le(-z^\varepsilon \boldsymbol{F}^{(\varepsilon)}+\mu^{\Le(\varepsilon\Ri)} \boldsymbol{I}\Ri)\boldsymbol{R}^{(\varepsilon)}\Ri)^2\Le(-z^\varepsilon \boldsymbol{F}^{(\varepsilon)}+\mu \boldsymbol{I}\Ri)e^T\Ri|
            \end{align*} is uniformly bounded for all $z,\varepsilon$ then the convergence of $\Le(1-z^{-\varepsilon}\Ri)$ to $0$ implies the product converges to $0$ as well.
            Utilizing the bounds we derived above we can see that
            \begin{align*}
                &\Le|\Q\Le(z^\varepsilon;\varepsilon\Ri) \Le(\Le(-z^\varepsilon \boldsymbol{F}^{(\varepsilon)}+\mu^{\Le(\varepsilon\Ri)} \boldsymbol{I}\Ri)\boldsymbol{R}^{(\varepsilon)}\Ri)^2\Le(-z^\varepsilon \boldsymbol{F}^{(\varepsilon)}+\mu \boldsymbol{I}\Ri)e^T\Ri| \\
                &=\Big|\sum_{i,j}[\Q\Le(z^\varepsilon;\varepsilon\Ri)\Le(-z^\varepsilon \boldsymbol{F}^{(\varepsilon)}+\mu^{\Le(\varepsilon\Ri)} \boldsymbol{I}\Ri)]_i[\boldsymbol{R}^{(\varepsilon)}\Le(-z^\varepsilon \boldsymbol{F}^{(\varepsilon)}+\mu^{\Le(\varepsilon\Ri)} \boldsymbol{I}\Ri)]_{i,j}\\ &\qquad\qquad \cdot[\boldsymbol{R}^{(\varepsilon)}\Le(-z^\varepsilon \boldsymbol{F}^{(\varepsilon)}+\mu^{\Le(\varepsilon\Ri)} \boldsymbol{I}\Ri)e^T]_{j}\Big|\\
                &\leq \Le|\sum_{i,j}\Le(\tilde{C}_1 a^i\Ri)\Le(\tilde{C}_2\tilde{b}^{\Le(j-i\Ri)^+}\Ri)\Le(\tilde{C}_3j\Ri)\Ri|
                =\tilde{C}_1\tilde{C}_2\tilde{C}_3\Le|\sum_{i=0}^\infty\Le(a^i\Le(\sum_{j=0}^{i-1}j+\sum_{j=i}^\infty j\tilde{b}^{j-i}\Ri)\Ri)\Ri|\\
                &=\tilde{C}_1\tilde{C}_2\tilde{C}_3\Le|\sum_{i=0}^\infty\Le(a^i\Le(\sum_{j=0}^{i-1}j+\sum_{j=i}^\infty \Le(j-i\Ri)\tilde{b}^{j-i}+\sum_{j=i}^\infty i \tilde{b}^{j-i}\Ri)\Ri)\Ri|\\
                &=\tilde{C}_1\tilde{C}_2\tilde{C}_3\Le|\sum_{i=0}^\infty\Le(a^i\Le(\sum_{j=0}^{i-1}j+\sum_{k=0}^\infty k\tilde{b}^k+i\sum_{k=0}^\infty  \tilde{b}^{k}\Ri)\Ri)\Ri|\\
                &\leq \tilde{C}_1\tilde{C}_2\tilde{C}_3\Le|\sum_{i=0}^\infty\Le(a^i\Le(\frac{\Le(i-1\Ri)\Le(i-2\Ri)}{2}+\frac{\tilde{b}}{\Le(1-\tilde{b}\Ri)^2}+i\frac{1}{1-\tilde{b}}\Ri)\Ri)\Ri|\\
                &= \tilde{C}_1\tilde{C}_2\tilde{C}_3\Le|\sum_{i=0}^\infty\Le(a^iO\Le(i^2\Ri)\Ri)\Ri| < \infty
            \end{align*}

            \item $\mu^{\Le(\varepsilon\Ri)}\Le(1-z^{\varepsilon}\Ri)\bpi^{\Le(\varepsilon\Ri)}_0\boldsymbol{R}^{(\varepsilon)}\Le(-z^\varepsilon \boldsymbol{F}^{(\varepsilon)}+\mu^{\Le(\varepsilon\Ri)} \boldsymbol{I}\Ri)\boldsymbol{R}^{(\varepsilon)}\Le(-z^\varepsilon \boldsymbol{F}^{(\varepsilon)}+\mu \boldsymbol{I}\Ri)e^T\to 0$.
            
            Because $\mu^{\Le(\varepsilon\Ri)}\Le(1-z^{\varepsilon}\Ri)\to 0$ as $\varepsilon\to 0$, if
            \begin{align*}
                \Le|\bpi^{\Le(\varepsilon\Ri)}_0\boldsymbol{R}^{(\varepsilon)}\Le(-z^\varepsilon \boldsymbol{F}^{(\varepsilon)}+\mu^{\Le(\varepsilon\Ri)} \boldsymbol{I}\Ri)\boldsymbol{R}^{(\varepsilon)}\Le(-z^\varepsilon \boldsymbol{F}^{(\varepsilon)}+\mu^{(\varepsilon)} \boldsymbol{I}\Ri)e^T\Ri|
            \end{align*} is uniformly bounded for all $z,\varepsilon$ then the convergence of $\mu^{\Le(\varepsilon\Ri)}\Le(1-z^{\varepsilon}\Ri)$ to $0$ implies the product converges to $0$ as well.
            Utilizing the bounds we derived above we can see that
            \begin{align*}
                \big| \bpi^{\Le(\varepsilon\Ri)}_0&\boldsymbol{R}^{(\varepsilon)}\Le(-z^\varepsilon \boldsymbol{F}^{(\varepsilon)}+\mu^{\Le(\varepsilon\Ri)} \boldsymbol{I}\Ri)\boldsymbol{R}^{(\varepsilon)}\Le(-z^\varepsilon \boldsymbol{F}^{(\varepsilon)}+\mu^{(\varepsilon)} \boldsymbol{I}\Ri)e^T \big| \\
                &=\Le|\sum_{i,j}\pi^{\Le(\varepsilon\Ri)}_{i,0}[\boldsymbol{R}^{(\varepsilon)}\Le(-z^\varepsilon \boldsymbol{F}^{(\varepsilon)}+\mu^{\Le(\varepsilon\Ri)} \boldsymbol{I}\Ri)]_{i,j}[\boldsymbol{R}^{(\varepsilon)}\Le(-z^\varepsilon \boldsymbol{F}^{(\varepsilon)}+\mu^{\Le(\varepsilon\Ri)} \boldsymbol{I}\Ri)e^T]_{j}\Ri|\\
                &\leq \Le|\sum_{i,j}\Le(C_1 a^i\Ri)\Le(\tilde{C}_2\tilde{b}^{\Le(j-i\Ri)^+}\Ri)\Le(\tilde{C}_3j\Ri)\Ri|
                =\tilde{C}_1\tilde{C}_2\tilde{C}_3\Le|\sum_{i=0}^\infty\Le(a^i\Le(\sum_{j=0}^{i-1}j+\sum_{j=i}^\infty j\tilde{b}^{j-i}\Ri)\Ri)\Ri|\\
                &=C_1\tilde{C}_2\tilde{C}_3\Le|\sum_{i=0}^\infty\Le(a^i\Le(\sum_{j=0}^{i-1}j+\sum_{j=i}^\infty \Le(j-i\Ri)\tilde{b}^{j-i}+\sum_{j=i}^\infty i \tilde{b}^{j-i}\Ri)\Ri)\Ri|\\
                &=C_1\tilde{C}_2\tilde{C}_3\Le|\sum_{i=0}^\infty\Le(a^i\Le(\sum_{j=0}^{i-1}j+\sum_{k=0}^\infty k\tilde{b}^k+i\sum_{k=0}^\infty  \tilde{b}^{k}\Ri)\Ri)\Ri|\\
                &\leq \tilde{C}_1\tilde{C}_2\tilde{C}_3\Le|\sum_{i=0}^\infty\Le(a^i\Le(\frac{\Le(i-1\Ri)\Le(i-2\Ri)}{2}+\frac{\tilde{p}}{\Le(1-\tilde{p}\Ri)^2}+i\frac{1}{1-\tilde{p}}\Ri)\Ri)\Ri|\\
                &= C_1\tilde{C}_2\tilde{C}_3\Le|\sum_{i=0}^\infty\Le(a^iO\Le(i^2\Ri)\Ri)\Ri| < \infty
            \end{align*}
            
            \item $\bpi^{\Le(\varepsilon\Ri)}_0\boldsymbol{R}^{(\varepsilon)}\Le(-z^\varepsilon \boldsymbol{F}^{(\varepsilon)}+\mu^{(\varepsilon)} \boldsymbol{I}\Ri)e^T\to 0$. 

            \begin{align*}
                &\Le|\bpi^{\Le(\varepsilon\Ri)}_0\boldsymbol{R}^{(\varepsilon)}\Le(-z^\varepsilon \boldsymbol{F}^{(\varepsilon)}+\mu^{(\varepsilon)} \boldsymbol{I}\Ri)e^T\Ri| 
                = \Le|\sum_{i=0}^\infty\pi_{i,0}^{\Le(\varepsilon\Ri)}[\boldsymbol{R}^{(\varepsilon)}\Le(-z^\varepsilon \boldsymbol{F}^{(\varepsilon)}+\mu^{(\varepsilon)} \boldsymbol{I}\Ri)e^T]_i\Ri|\\
                \leq &\Le|\sum_{i=0}^\infty\pi_{i,0}^{\Le(\varepsilon\Ri)}\tilde{C}_3i\Ri|
                = \Le|\sum_{i=0}^{\lfloor\log \frac{1}{\varepsilon}\rfloor}\pi_{i,0}^{\Le(\varepsilon\Ri)}\tilde{C}_3i + \sum_{i=\lfloor\log \frac{1}{\varepsilon}\rfloor+1}^\infty\pi_{i,0}^{\Le(\varepsilon\Ri)}\tilde{C}_3i\Ri|\\
                \leq &\Le|\sum_{i=0}^{\lfloor\log \frac{1}{\varepsilon}\rfloor}\pi_{i,0}^{\Le(\varepsilon\Ri)}\tilde{C}_3\log \frac{1}{\varepsilon} + \sum_{i=\lfloor\log \frac{1}{\varepsilon}\rfloor+1}^\infty\pi_{i,0}^{\Le(\varepsilon\Ri)}\tilde{C}_3i\Ri|
                \leq \tilde{C}_3\Le|\varepsilon\log \frac{1}{\varepsilon}\Ri| + \tilde{C}_3 \Le|\sum_{i=\lfloor\log \frac{1}{\varepsilon}\rfloor+1}^\infty\pi_{i,0}^{\Le(\varepsilon\Ri)}i\Ri|\\
                \leq &\tilde{C}_3\Le|\varepsilon\log \frac{1}{\varepsilon}\Ri| + \tilde{C}_3 \Le|\sum_{i=\lfloor\log \frac{1}{\varepsilon}\rfloor+1}^\infty C_1 a^ii \Ri|
            \end{align*}
            Because $\varepsilon\log \frac{1}{\varepsilon}\to 0$ as $\varepsilon\to0$, we only need to establish that the second term also goes to $0$. Notice that $\sum_{i=0}^\infty ia^i$ is a convergent sum because $0<a<1$. Thus as $\lfloor\log \frac{1}{\varepsilon}\rfloor+1$ goes to $\infty$ as $\varepsilon\to 0$, the ``tail sum" $\sum_{i=\lfloor\log \frac{1}{\varepsilon}\rfloor+1}^\infty ia^i$ must also go to $0$. Therefore both terms approach $0$ and we have our desired result.

            \item $\Le(1-z^{-\varepsilon}\Ri)\bpi^{\Le(\varepsilon\Ri)}\Le(\mu^{\Le(\varepsilon\Ri)} \boldsymbol{I} - z^{\varepsilon}  \boldsymbol{F}^{(\varepsilon)}\Ri)\boldsymbol{R}^{(\varepsilon)}e^T\to 0$

            Because $\Le(1-z^{-\varepsilon}\Ri)\to 0$ as $\varepsilon\to 0$, if $\Le|\bpi^{\Le(\varepsilon\Ri)}\Le(\mu^{\Le(\varepsilon\Ri)} \boldsymbol{I} - z^{\varepsilon}  \boldsymbol{F}^{(\varepsilon)}\Ri)\boldsymbol{R}^{(\varepsilon)}e^T\Ri|$ is uniformly bounded for all $z,\varepsilon$ then their product also goes to $0$.

            \begin{align*}
                &\Le|\bpi^{\Le(\varepsilon\Ri)}\Le(\mu^{\Le(\varepsilon\Ri)} \boldsymbol{I} - z^{\varepsilon}  \boldsymbol{F}^{(\varepsilon)}\Ri)\boldsymbol{R}^{(\varepsilon)}e^T\Ri| = \Big|\sum_{i=0}^\infty [\bpi^{\Le(\varepsilon\Ri)}\Le(\mu^{\Le(\varepsilon\Ri)} \boldsymbol{I} - z^{\varepsilon}  \boldsymbol{F}^{(\varepsilon)}\Ri)]_i[\boldsymbol{R}^{(\varepsilon)}e^T]_i\Big|\\
                \leq & \Big|\sum_{i=0}^\infty\Le(\tilde{C}_1 a^i \cdot C_3i\Ri)\Big|
                =\tilde{C}_1C_3\big|\frac{a}{\Le(1-a\Ri)^2}\big| < \infty
            \end{align*}

            \item $\Le(1-z^{-\varepsilon}\Ri)\Q\Le(z^\varepsilon;\varepsilon\Ri)\Le(\Le(\mu^{\Le(\varepsilon\Ri)}\boldsymbol{I}-z^{\varepsilon} \boldsymbol{F}^{(\varepsilon)}\Ri)\boldsymbol{R}^{(\varepsilon)}\Ri)^2e^T\to 0$.
            
            Because $\Le(1-z^{-\varepsilon}\Ri)\to 0$ as $\varepsilon\to 0$, if 
            \begin{align*}
                \Le|\Q\Le(z^\varepsilon;\varepsilon\Ri) \Le(\Le(-z^\varepsilon \boldsymbol{F}^{(\varepsilon)}+\mu^{\Le(\varepsilon\Ri)} \boldsymbol{I}\Ri)\boldsymbol{R}^{(\varepsilon)}\Ri)^2\Le(-z^\varepsilon \boldsymbol{F}^{(\varepsilon)}+\mu^{(\varepsilon)} \boldsymbol{I}\Ri)e^T\Ri|
            \end{align*} is uniformly bounded for all $z,\varepsilon$ then the convergence of $\Le(1-z^{-\varepsilon}\Ri)$ to $0$ implies the product converges to $0$ as well.
            Utilizing the bounds we derived above we can see that
            \begin{align*}
                &\Big|\Q\Le(z^\varepsilon;\varepsilon\Ri) \Le(\Le(-z^\varepsilon \boldsymbol{F}^{(\varepsilon)}+\mu^{\Le(\varepsilon\Ri)} \boldsymbol{I}\Ri)\boldsymbol{R}^{(\varepsilon)}\Ri)^2e^T\Big| \\
                &=\Big|\sum_{i,j}[\Q\Le(z^\varepsilon;\varepsilon\Ri)\Le(-z^\varepsilon \boldsymbol{F}^{(\varepsilon)}+\mu^{\Le(\varepsilon\Ri)} \boldsymbol{I}\Ri)]_i[\boldsymbol{R}^{(\varepsilon)}\Le(-z^\varepsilon \boldsymbol{F}^{(\varepsilon)}+\mu^{\Le(\varepsilon\Ri)} \boldsymbol{I}\Ri)]_{i,j}[\boldsymbol{R}^{(\varepsilon)}e^T]_{j}\Big|\\
                &\leq \Big|\sum_{i,j}\Le(\tilde{C}_1 a^i\Ri)\Le(\tilde{C}_2\tilde{b}^{\Le(j-i\Ri)^+}\Ri)\Le(C_3j\Ri)\Big|\\
                &=\tilde{C}_1\tilde{C}_2C_3\Big|\sum_{i=0}^\infty\Le(a^i\Le(\sum_{j=0}^{i-1}j+\sum_{j=i}^\infty j\tilde{b}^{j-i}\Ri)\Ri)\Big|\\
                &=\tilde{C}_1\tilde{C}_2C_3\Big|\sum_{i=0}^\infty\Le(a^i\Le(\sum_{j=0}^{i-1}j+\sum_{j=i}^\infty \Le(j-i\Ri)\tilde{b}^{j-i}+\sum_{j=i}^\infty i \tilde{b}^{j-i}\Ri)\Ri)\Big|\\
                &=\tilde{C}_1\tilde{C}_2C_3\Big|\sum_{i=0}^\infty\Le(a^i\Le(\sum_{j=0}^{i-1}j+\sum_{k=0}^\infty k\tilde{b}^k+i\sum_{k=0}^\infty  \tilde{b}^{k}\Ri)\Ri)\Big|\\
                &\leq \tilde{C}_1\tilde{C}_2C_3\Big|\sum_{i=0}^\infty\Le(a^i\Le(\frac{\Le(i-1\Ri)\Le(i-2\Ri)}{2}+\frac{\tilde{b}}{\Le(1-\tilde{b}\Ri)^2}+i\frac{1}{1-\tilde{b}}\Ri)\Ri)\Big|\\
                &= \tilde{C}_1\tilde{C}_2C_3\Big|\sum_{i=0}^\infty\Le(a^iO\Le(i^2\Ri)\Ri)\Big| < \infty
            \end{align*}

            \item $\mu^{\Le(\varepsilon\Ri)}\Le(1-z^{\Le(\varepsilon\Ri)}\Ri)\piidle^{\Le(\varepsilon\Ri)}\boldsymbol{R}^{(\varepsilon)}\Le(\mu^{\Le(\varepsilon\Ri)}I-z^{\varepsilon} \boldsymbol{F}^{(\varepsilon)}\Ri)\boldsymbol{R}^{(\varepsilon)}e^T\to 0$.

            Because $\mu^{\Le(\varepsilon\Ri)}\Le(1-z^{\varepsilon}\Ri)\to 0$ as $\varepsilon\to 0$, if $|\piidle^{\Le(\varepsilon\Ri)}\boldsymbol{R}^{(\varepsilon)}\Le(\mu^{\Le(\varepsilon\Ri)}I-z^{\varepsilon} \boldsymbol{F}^{(\varepsilon)}\Ri)\boldsymbol{R}^{(\varepsilon)}e^T|$ is uniformly bounded for all $z,\varepsilon$ then the convergence of $\mu^{\Le(\varepsilon\Ri)}\Le(1-z^{\varepsilon}\Ri)$ to $0$ implies the product converges to $0$ as well.
            Utilizing the bounds we derived above we can see that
            \begin{align*}
                &\Le|\piidle^{\Le(\varepsilon\Ri)}\boldsymbol{R}^{(\varepsilon)}\Le(\mu^{\Le(\varepsilon\Ri)}I-z^{\varepsilon} \boldsymbol{F}^{(\varepsilon)}\Ri)\boldsymbol{R}^{(\varepsilon)}e^T\Ri| \\
                &=\Le|\sum_{i,j}\pi^{\Le(\varepsilon\Ri)}_{i,0}[\boldsymbol{R}^{(\varepsilon)}\Le(-z^\varepsilon \boldsymbol{F}^{(\varepsilon)}+\mu^{\Le(\varepsilon\Ri)} \boldsymbol{I}\Ri)]_{i,j}[\boldsymbol{R}^{(\varepsilon)}e^T]_{j}\Ri|\\
                &\leq \Le|\sum_{i,j}\Le(C_1 a^i\Ri)\Le(\tilde{C}_2p^{\Le(j-i\Ri)^+}\Ri)\Le(C_3j\Ri)\Ri|\\
                &=C_1\tilde{C}_2C_3\Le|\sum_{i=0}^\infty\Le(a^i\Le(\sum_{j=0}^{i-1}j+\sum_{j=i}^\infty jp^{j-i}\Ri)\Ri)\Ri|\\
                &=C_1\tilde{C}_2C_3\Le|\sum_{i=0}^\infty\Le(a^i\Le(\sum_{j=0}^{i-1}j+\sum_{j=i}^\infty \Le(j-i\Ri)p^{j-i}+\sum_{j=i}^\infty i p^{j-i}\Ri)\Ri)\Ri|\\
                &=C_1\tilde{C}_2\tilde{C}_3\Le|\sum_{i=0}^\infty\Le(a^i\Le(\sum_{j=0}^{i-1}j+\sum_{k=0}^\infty kp^k+i\sum_{k=0}^\infty  p^{k}\Ri)\Ri)\Ri|\\
                &\leq C_1\tilde{C}_2C_3\Le|\sum_{i=0}^\infty\Le(a^i\Le(\frac{\Le(i-1\Ri)i}{2}+\frac{p}{\Le(1-p\Ri)^2}+i\frac{1}{1-p}\Ri)\Ri)\Ri|\\
                &= C_1\tilde{C}_2C_3\Le|\sum_{i=0}^\infty\Le(a^iO\Le(i^2\Ri)\Ri)\Ri| < \infty
            \end{align*}

            \item $\mu^{\Le(\varepsilon\Ri)}\piidle^{\Le(\varepsilon\Ri)}\boldsymbol{R}^{(\varepsilon)}e^T\to 0$.

            \begin{align*}
                \Le|\mu^{\Le(\varepsilon\Ri)}\bpi^{\Le(\varepsilon\Ri)}_0\boldsymbol{R}^{(\varepsilon)}e^T\Ri| &= \Le|\mu^{\Le(\varepsilon\Ri)}\sum_{i=0}^\infty\pi_{i,0}^{\Le(\varepsilon\Ri)}[\boldsymbol{R}^{(\varepsilon)}e^T]_i\Ri|
                \leq \Le|\mu^{\Le(\varepsilon\Ri)}\sum_{i=0}^\infty\pi_{i,0}^{\Le(\varepsilon\Ri)} C_3i\Ri|\\
                &= |\mu^{\Le(\varepsilon\Ri)}| \cdot
                \Le|\sum_{i=0}^{\lfloor\log \frac{1}{\varepsilon}\rfloor}\pi_{i,0}^{\Le(\varepsilon\Ri)} C_3i + \sum_{i=\lfloor\log \frac{1}{\varepsilon}\rfloor+1}^\infty\pi_{i,0}^{\Le(\varepsilon\Ri)} C_3i\Ri|\\
                &\leq |\mu^{\Le(\varepsilon\Ri)}| \cdot \Le|\sum_{i=0}^{\lfloor\log \frac{1}{\varepsilon}\rfloor}\pi_{i,0}^{\Le(\varepsilon\Ri)} C_3\log \frac{1}{\varepsilon} + \sum_{i=\lfloor\log \frac{1}{\varepsilon}\rfloor+1}^\infty\pi_{i,0}^{\Le(\varepsilon\Ri)} C_3i\Ri|\\
                &\leq |\mu^{\Le(\varepsilon\Ri)}| \cdot C_3|\varepsilon\log \frac{1}{\varepsilon}| + |\mu^{\Le(\varepsilon\Ri)}| \cdot C_3 \Le|\sum_{i=\lfloor\log \frac{1}{\varepsilon}\rfloor+1}^\infty\pi_{i,0}^{\Le(\varepsilon\Ri)}i\Ri|\\
                &\leq |\mu^{\Le(\varepsilon\Ri)}| \cdot C_3|\varepsilon\log \frac{1}{\varepsilon}| +  |\mu^{\Le(\varepsilon\Ri)}| \cdot C_3 \Le|\sum_{i=\lfloor\log \frac{1}{\varepsilon}\rfloor+1}^\infty C_1 a^ii\Ri|
            \end{align*}
            Because $\varepsilon\log \frac{1}{\varepsilon}\to 0$ and $\mu^{\Le(\varepsilon\Ri)}\to \mu^{\Le(0\Ri)}\in \Le(0,\infty\Ri)$ as $\varepsilon\to0$, we only need to establish that the second term also goes to $0$. Notice that $\sum_{i=0}^\infty ia^i$ is a convergent sum because $0<q<1$. Thus as $\lfloor\log \frac{1}{\varepsilon}\rfloor+1$ goes to $\infty$ as $\varepsilon\to 0$, the ``tail sum" $\sum_{i=\lfloor\log \frac{1}{\varepsilon}\rfloor+1}^\infty ia^i$ must also go to $0$. Therefore both terms approach $0$ and we have our desired result.
        \end{enumerate}

        \item For all $z\in\Le(0,1\Ri)$ and all $j\in \mathbb{N}_0$ as $\varepsilon\to 0$:
        \begin{enumerate}
            \item $[\Le(1-z^{-\varepsilon}\Ri)\boldsymbol{\pi}^{\Le(\varepsilon\Ri)}\Le(\mu^{\Le(\varepsilon\Ri)}I-z \boldsymbol{F}^{(\varepsilon)}\Ri)\boldsymbol{R}^{(\varepsilon)}]_j\to 0$.
            \begin{align*}
            &\Le|[\boldsymbol{\pi}^{\Le(\varepsilon\Ri)}\Le(\mu^{\Le(\varepsilon\Ri)}I-z \boldsymbol{F}^{(\varepsilon)}\Ri)\boldsymbol{R}^{(\varepsilon)}]_j\Ri| = \Le|\sum_{i=0}^\infty [\bpi^{\Le(\varepsilon\Ri)}\Le(\mu^{\Le(\varepsilon\Ri)}I-z \boldsymbol{F}^{(\varepsilon)}\Ri)]_i r^{\Le(\varepsilon\Ri)}_{i,j}\Ri| \\
                \leq & \Le|\sum_{i=0}^\infty \tilde{C}_1 a^i C_3p^{\Le(j-i\Ri)^+}\Ri|
                \leq \Le|\sum_{i=0}^{j} \tilde{C}_1 a^i C_3p^{\Le(j-i\Ri)^+} + \sum_{i=j+1}^\infty \tilde{C}_1 a^i C_3\Ri|\\
                \leq & \Le|\sum_{i=0}^{j} \tilde{C}_1 a^i C_3p^{\Le(j-i\Ri)^+}\Ri| + \Le|\sum_{i=j+1}^\infty \tilde{C}_1 a^i C_3\Ri|
                \leq \Le|\sum_{i=0}^{j} \tilde{C}_1 a^i C_3p^{\Le(j-i\Ri)^+}\Ri| + \Le|\tilde{C}_1 C_3 \frac{q^{j+1}}{1-q}\Ri| < \infty
            \end{align*}

            Thus because $(1-z^{-\varepsilon})\to 0$ as $\varepsilon\to 0$ and $\Le|[\boldsymbol{\pi}^{\Le(\varepsilon\Ri)}\Le(\mu^{\Le(\varepsilon\Ri)}I-z^\varepsilon \boldsymbol{F}^{(\varepsilon)}\Ri)\boldsymbol{R}^{(\varepsilon)}]_j\Ri|$ is bounded uniformly across $\varepsilon, z$, their product goes to $0$ as $\varepsilon\to 0$.

            \item $[\mu^{\Le(\varepsilon\Ri)}\Le(1-z^{-\varepsilon}\Ri)\Q\Le(z^\varepsilon;\varepsilon\Ri) \Le(\Le(-z^\varepsilon \boldsymbol{F}^{(\varepsilon)}+\mu^{\Le(\varepsilon\Ri)} \boldsymbol{I}\Ri)R\Ri)^2]_j \to 0$
            \begin{align*}
                &\Le|[\Q(z^\varepsilon;\varepsilon) ((-z^\varepsilon \boldsymbol{F}^{(\varepsilon)}+\mu^{(\varepsilon)} \boldsymbol{I})R)^2]_j\Ri|\\
                &=\Le|\sum_{i,k}[\Q(z^\varepsilon;\varepsilon)(-z^\varepsilon \boldsymbol{F}^{(\varepsilon)}+\mu^{(\varepsilon)} \boldsymbol{I})]_i \cdot [R^{(\varepsilon)}(-z^\varepsilon \boldsymbol{F}+\mu^{(\varepsilon)} \boldsymbol{I})]_{i,k}\cdot r^{(\varepsilon)}_{k,j} \Ri|\\
                &=a \Le|\sum_{i,k}\tilde{C}_1 a^i \cdot \tilde{C}_3\tilde{b}^{(k-i)^+}\cdot C_3\tilde{b}^{(j-k)^+} \Ri|
                =a \Le|\sum_{i=0}^\infty\sum_{k=0}^\infty(\tilde{C}_1 a^i \cdot \tilde{C}_3\tilde{b}^{(k-i)^+}\cdot C_3\tilde{b}^{(j-k)^+}) \Ri|\\
                &=a \Le|\sum_{i=0}^\infty\tilde{C}_1 a^i(\sum_{k=0}^{i-1}(\tilde{C}_3\tilde{b}^{(k-i)^+}\cdot C_3\tilde{b}^{(j-k)^+})+\sum_{k=i}^\infty(\tilde{C}_3\tilde{b}^{(k-i)^+}\cdot C_3\tilde{b}^{(j-k)^+})) \Ri|\\
                &=a \tilde{C}_1\tilde{C}_3C_3\Le|\sum_{i=0}^\infty a^i(\sum_{k=0}^{i-1}(\tilde{b}^{(k-i)^+}\cdot \tilde{b}^{(j-k)^+})+\sum_{k=i}^\infty(\tilde{b}^{(k-i)^+}\cdot \tilde{b}^{(j-k)^+})) \Ri|\\
                &=a \tilde{C}_1\tilde{C}_3C_3\Le|\sum_{i=0}^\infty a^i(\sum_{k=0}^{i-1}(1)+\sum_{k=i}^\infty(\tilde{b}^{(k-i)^+})) \Ri|
                =a \tilde{C}_1\tilde{C}_3C_3\Le|\sum_{i=0}^\infty a^i(i+\frac{1}{1-\tilde{b}})) \Ri|\\
                &=\tilde{C}_1\tilde{C}_3C_3\Le|\frac{a}{(1-a)^2}+\frac{1}{(1-\tilde{b})(1-a)}) \Ri| < \infty
            \end{align*}

            Thus because $(1-z^{-\varepsilon})\to 0$ as $\varepsilon\to 0$ and $[\Q\Le(z^\varepsilon;\varepsilon\Ri) \Le(\Le(-z^\varepsilon \boldsymbol{F}^{(\varepsilon)}+\mu^{\Le(\varepsilon\Ri)} \boldsymbol{I}\Ri)R\Ri)^2]_j$ is bounded uniformly across $\varepsilon, z$, their product goes to $0$ as $\varepsilon\to 0$.
                        
            \item $[\mu^{\Le(\varepsilon\Ri)}\Le(1-z^{-\varepsilon}\Ri)\bpi^{\Le(\varepsilon\Ri)}_0\boldsymbol{R}^{(\varepsilon)}\Le(-z^\varepsilon \boldsymbol{F}^{(\varepsilon)} + \mu^{\Le(\varepsilon\Ri)}I\Ri)\boldsymbol{R}^{(\varepsilon)}]_j\to 0$.
            \begin{align*}
                &\Le|[\bpi^{(\varepsilon)}_0R^{(\varepsilon)}(-z^\varepsilon \boldsymbol{F}^{(\varepsilon)} + \mu^{(\varepsilon)} \boldsymbol{I})R^{(\varepsilon)}]_j\Ri|\\
                &=\Le|\sum_{i,k}\pi^{(\varepsilon)}_{i,0} \cdot [R^{(\varepsilon)}(-z^\varepsilon \boldsymbol{F}^{(\varepsilon)}+\mu^{(\varepsilon)} \boldsymbol{I})]_{i,k}\cdot r^{(\varepsilon)}_{k,j} \Ri|\\
                &=a \Le|\sum_{i,k}\tilde{C}_1 a^i \cdot \tilde{C}_3\tilde{b}^{(k-i)^+}\cdot C_3\tilde{b}^{(j-k)^+} \Ri|\\
                &=a \Le|\sum_{i=0}^\infty\sum_{k=0}^\infty(\tilde{C}_1 a^i \cdot \tilde{C}_3\tilde{b}^{(k-i)^+}\cdot C_3\tilde{b}^{(j-k)^+}) \Ri|\\
                &=a \Le|\sum_{i=0}^\infty\tilde{C}_1 a^i(\sum_{k=0}^{i-1}(\tilde{C}_3\tilde{b}^{(k-i)^+}\cdot C_3\tilde{b}^{(j-k)^+})+\sum_{k=i}^\infty(\tilde{C}_3\tilde{b}^{(k-i)^+}\cdot C_3\tilde{b}^{(j-k)^+})) \Ri|\\
                &=a \tilde{C}_1\tilde{C}_3C_3\Le|\sum_{i=0}^\infty a^i(\sum_{k=0}^{i-1}(\tilde{b}^{(k-i)^+}\cdot \tilde{b}^{(j-k)^+})+\sum_{k=i}^\infty(\tilde{b}^{(k-i)^+}\cdot \tilde{b}^{(j-k)^+})) \Ri|\\
                &=a \tilde{C}_1\tilde{C}_3C_3\Le|\sum_{i=0}^\infty a^i(\sum_{k=0}^{i-1}(1)+\sum_{k=i}^\infty(\tilde{b}^{(k-i)^+})) \Ri|\\
                &=a \tilde{C}_1\tilde{C}_3C_3|\sum_{i=0}^\infty a^i(i+\frac{1}{1-\tilde{b}})) |\\
                &=\tilde{C}_1\tilde{C}_3C_3|\frac{a}{(1-a)^2}+\frac{1}{(1-\tilde{b})(1-a)}) | < \infty
            \end{align*}
            Thus because $\Le(1-z^{-\varepsilon}\Ri)\to 0$ as $\varepsilon\to 0$ and $|[\bpi^{(\varepsilon)}_0R^{(\varepsilon)}(-z^\varepsilon \boldsymbol{F}^{(\varepsilon)} + \mu^{(\varepsilon)} \boldsymbol{I})R^{(\varepsilon)}]_j|$ is uniformly bounded across all $z,\varepsilon$, their product also goes to $0$ as $\varepsilon\to 0$.

            \item $[\mu^{\Le(\varepsilon\Ri)}\bpi^{\Le(\varepsilon\Ri)}_0\boldsymbol{R}^{(\varepsilon)}]_j\to 0$.

            \begin{align*}
                \Le|[\mu^{\Le(\varepsilon\Ri)}\bpi^{\Le(\varepsilon\Ri)}_0\boldsymbol{R}^{(\varepsilon)}]_j\Ri| &\leq \Le|\mu^{\Le(\varepsilon\Ri)}| \ |\sum_{i=0}^\infty\pi^{\Le(\varepsilon\Ri)}_{i,0} r^{\Le(\varepsilon\Ri)}_{i,j}\Ri|
                \leq |\mu^{\Le(\varepsilon\Ri)}| \ \Le|\sum_{i=0}^\infty\pi^{\Le(\varepsilon\Ri)}_{i,0} r^{\Le(\varepsilon\Ri)}_{j,j}\Ri|\\
                &=|\mu^{\Le(\varepsilon\Ri)}| \ \Le|r_{j,j}^{\Le(\varepsilon\Ri)}| \ |\sum_{i=0}^\infty \pi^{\Le(\varepsilon\Ri)}_{i,0}\Ri|
                =\varepsilon |\mu^{\Le(\varepsilon\Ri)}| \ |r_{j,j}^{\Le(\varepsilon\Ri)}|
            \end{align*}
            Because $|\mu^{\Le(\varepsilon\Ri)}|$ and $ |r_{j,j}^{\Le(\varepsilon\Ri)}|$ approach constants $|\mu^{\Le(0\Ri)}|$ and $ |r_{j,j}^{\Le(0\Ri)}|$ as $\varepsilon\to 0$, $\varepsilon |\mu^{\Le(\varepsilon\Ri)}| \ |r_{j,j}^{\Le(\varepsilon\Ri)}| \to 0$.
        \end{enumerate}
    \end{enumerate}
\end{proof}
\section{Simulation Details}
\label{app:sim}

To validate the heavy-traffic approximation for the mean queue length, we simulate the hybrid-abandonment system for $10^{8}$ jobs at each parameter set $(\rho, \theta, \mu_1)$ and record the second queue length at each arrival to the first queue. By PASTA \cite{wolff1982poisson}, these are samples of the time-average second queue length and provide an easy way to estimate the mean queue length and the standard deviation of our estimator.

To validate that the second queue length is approximately exponentially distributed and that the queues are not meaningfully correlated, we again simulated the system for $10^8$ jobs at each parameter set but instead broke the $10^8$ job-long runs into $20$ ``batches" of $5\cdot 10^6$ job-long runs. Each of these batches yielded one estimate of our desired quantity, either the correlation of the queue lengths or the KS statistic between the second queue length and the exponential or geometric random variable.

Finally, to test the efficacy of our approximation in capacity planning, we simulate the system for $40$ batches of $10^5$ jobs each at each of 400 different $\mu_2$ values for each parameter set $(\lambda, \theta, \mu_1)$. Because the relative cost parameter $c$ does not affect the queue length at any system configuration, we did not need to rerun the system at each $c$ value. Instead, the existing queue length data could be reused for each differing $c$ value.

Unaggregated data and full simulation code can be found at \url{https://github.com/CJCurtis/HybridAbandonmentSimulator}.

\end{appendices}


\bibliography{refs}

@article{zayas2023scheduling,
  title={Scheduling servers in a two-stage queue with abandonments and costs},
  author={Zayas-Cab{\'a}n, Gabriel and Cochran, Amy L},
  journal={Probability in the Engineering and Informational Sciences},
  volume={37},
  number={3},
  pages={833--851},
  year={2023}
}

@article{zayas2016dynamic,
  title={Dynamic control of a tandem system with abandonments},
  author={Zayas-Cab{\'a}n, Gabriel and Xie, Jingui and Green, Linda V and Lewis, Mark E},
  journal={Queueing Systems},
  volume={84},
  pages={279--293},
  year={2016},
  publisher={Springer}
}

@article{zayas2019policies,
  title={Policies for physician allocation to triage and treatment in emergency departments},
  author={Zayas-Caban, Gabriel and Xie, Jingui and Green, Linda V and Lewis, Mark E},
  journal={IISE Transactions on Healthcare Systems Engineering},
  volume={9},
  number={4},
  pages={342--356},
  year={2019},
  publisher={Taylor \& Francis}
}

@inproceedings{reed2004diffusion,
  title={A diffusion approximation for a generalized Jackson network with reneging},
  author={Reed, Josh and Ward, Amy R},
  booktitle={Proceedings of the 42nd annual Allerton conference on communication, control, and computing},
  year={2004}
}

@article{burman1986asymptotic,
  title={An asymptotic analysis of a queueing system with {Markov}-modulated arrivals},
  author={Burman, David Y and Smith, Donald R},
  journal={Operations Research},
  volume={34},
  number={1},
  pages={105--119},
  year={1986},
  publisher={INFORMS}
}

@article{dimitrov2011single,
  title={Single-server queueing system with {Markov}-modulated arrivals and service times},
  author={Dimitrov, Mitko},
  journal={Pliska Studia Mathematica Bulgarica},
  volume={20},
  number={1},
  pages={53p--62p},
  year={2011}
}

@article{falin1999heavy,
  title={Heavy traffic analysis of {M/G/1} type queueing systems with {Markov}-modulated arrivals},
  author={Falin, Gennadi and Falin, Anatoli},
  journal={Top},
  volume={7},
  number={2},
  pages={279--291},
  year={1999},
  publisher={Springer}
}

@article{wang2019tandem,
  title={Tandem queues with impatient customers},
  author={Wang, Jianfu and Abouee-Mehrizi, Hossein and Baron, Opher and Berman, Oded},
  journal={Performance Evaluation},
  volume={135},
  pages={102011},
  year={2019},
  publisher={Elsevier}
}

@article{jackson1957networks,
  title={Networks of waiting lines},
  author={Jackson, James R},
  journal={Operations research},
  volume={5},
  number={4},
  pages={518--521},
  year={1957},
  publisher={INFORMS}
}

@article{andradottir2023optimal,
  title={Optimal server assignment in queues with flexible servers and abandonments},
  author={Andrad{\'o}ttir, Sigr{\'u}n and Ayhan, Hayriye},
  journal={Stochastic Systems},
  volume={13},
  number={3},
  pages={360--376},
  year={2023},
  publisher={INFORMS}
}

@article{yechiali1971queuing,
  title={Queuing problems with heterogeneous arrivals and service},
  author={Yechiali, Ury and Naor, Pinhas},
  journal={Operations Research},
  volume={19},
  number={3},
  pages={722--734},
  year={1971},
  publisher={Informs}
}

@article{mou2024heavy,
  title={Heavy-traffic queue length behavior in a switch under {Markovian} arrivals},
  author={Mou, Shancong and Maguluri, Siva Theja},
  journal={Advances in Applied Probability},
  volume={56},
  number={3},
  pages={1106--1152},
  year={2024},
  publisher={Cambridge University Press}
}

@article{neuts1978m,
  title={The {M/M/1} queue with randomly varying arrival and service rates},
  author={Neuts, Marcel F},
  journal={Opsearch},
  volume={15},
  pages={139--157},
  year={1978}
}

@article{regterschot1986queue,
  title={The queue {M$|$G$|$1} with {Markov} modulated arrivals and services},
  author={Regterschot, GJK and De Smit, JHA},
  journal={Mathematics of operations research},
  volume={11},
  number={3},
  pages={465--483},
  year={1986},
  publisher={INFORMS}
}

@article{ciucu2018two,
  title={Two extensions of Kingman's GI/G/1 bound},
  author={Ciucu, Florin and Poloczek, Felix},
  journal={Proceedings of the ACM on Measurement and Analysis of Computing Systems},
  volume={2},
  number={3},
  pages={1--33},
  year={2018},
  publisher={ACM New York, NY, USA}
}

@article{grosof2023reset,
  title={The RESET and MARC techniques, with application to multiserver-job analysis},
  author={Grosof, Isaac and Hong, Yige and Harchol-Balter, Mor and Scheller-Wolf, Alan},
  journal={Performance Evaluation},
  volume={162},
  pages={102378},
  year={2023},
  publisher={Elsevier}
}

@book{latouche1999introduction,
  title={Introduction to matrix analytic methods in stochastic modeling},
  author={Latouche, Guy and Ramaswami, Vaidyanathan},
  year={1999},
  publisher={SIAM},
    address={Philadelphia, PA}
}

@article{ramaswami1980n,
  title={The {N/G/1} queue and its detailed analysis},
  author={Ramaswami, Vadyanathan},
  journal={Advances in Applied Probability},
  volume={12},
  number={1},
  pages={222--261},
  year={1980},
  publisher={Cambridge University Press}
}

@book{neuts1994matrix,
  title={Matrix-geometric solutions in stochastic models: an algorithmic approach},
  author={Neuts, Marcel F},
  year={1994},
  publisher={Johns Hopkins University Press},
  address={Baltimore, MD}
}

@book{cohen1998,
title= {Statistical Power Analysis for the Behavioral Sciences},
author = {Cohen, Jacob},
year={1988},
publisher={L. Erlbaum Associates},
address={Hillsdale, NJ}

}

@book{holt1960planning,
  author    = {Holt, Charles C. and Modigliani, Franco and Muth, John F. and Simon, Herbert A.},
  title     = {Planning Production, Inventories, and Workforce},
  publisher = {Prentice-Hall},
  address   = {Englewood Cliffs, NJ},
  year      = {1960}
}

@book{kleinrock1975queueing,
  author    = {Kleinrock, Leonard},
  title     = {Queueing Systems. Volume I: Theory},
  publisher = {Wiley},
  address   = {New York},
  year      = {1975}
}

@article{Hurtado-LangeDaniela2026MMJi,
author = {Hurtado-Lange, Daniela and Grosof, Izzy},
copyright = {http://creativecommons.org/licenses/by/4.0},
language = {eng},
title = {{Markov} Modulated {JSQ} in Heavy Traffic Via the Poisson Equation},
year = {2026},
}

@article{choudhury1994heavy,
  title={Heavy-traffic asymptotic expansions for the asymptotic decay rates in the {BMAP/G/1} queue},
  author={Choudhury, Gagan L and Whitt, Ward},
  journal={Stochastic Models},
  volume={10},
  number={2},
  pages={453--498},
  year={1994},
  publisher={Taylor \& Francis}
}

@article{mitrani2005approximate,
  title={Approximate solutions for heavily loaded {Markov}-modulated queues},
  author={Mitrani, Isi},
  journal={Performance Evaluation},
  volume={62},
  number={1-4},
  pages={117--131},
  year={2005},
  publisher={Elsevier}
}

@article{thorsdottir2016markov,
  title={{Markov}-modulated {M/G/1}-type queue in heavy traffic and its application to time-sharing disciplines},
  author={Thorsdottir, Halldora and Verloop, Ina Maria},
  journal={Queueing Systems},
  volume={83},
  number={1},
  pages={29--55},
  year={2016},
  publisher={Springer}
}

@article{ramaswami1996some,
  title={Some properties of the rate operators in level dependent quasi-birth-and-death processes with countable number of phases},
  author={Ramaswami, V and Taylor, Peter G},
  journal={Stochastic Models},
  volume={12},
  number={1},
  pages={143--164},
  year={1996},
  publisher={Taylor \& Francis}
}

@article{latouche2011queues,
  title={Queues with boundary assistance: the effects of truncation},
  author={Latouche, Guy and Nguyen, Giang T and Taylor, Peter G},
  journal={Queueing Systems},
  volume={69},
  number={2},
  pages={175--197},
  year={2011},
  publisher={Springer}
}

@article{kroese2004spectral,
  title={Spectral properties of the tandem {Jackson} network, seen as a quasi-birth-and-death process},
  author={Kroese, DP and Scheinhardt, WRW and Taylor, PG},
  year={2004}
}

@book{he2014fundamentals,
  title={Fundamentals of matrix-analytic methods},
  author={He, Qi-Ming},
  volume={365},
  year={2014},
  publisher={Springer},
  address={New York, NY}
}

@article{lian2011departure,
  title={Departure processes and busy periods of a tandem network},
  author={Lian, Zhaotong and Zhao, Ning},
  journal={Operational Research},
  volume={11},
  number={3},
  pages={245--257},
  year={2011},
  publisher={Springer}
}

@article{wolff1982poisson,
author = {Wolff, Ronald W.},
title = {Poisson Arrivals See Time Averages},
journal = {Operations Research},
volume = {30},
number = {2},
pages = {223-231},
year = {1982},
doi = {10.1287/opre.30.2.223},
}

\end{document}